\documentclass[9pt]{article}
\usepackage{amsmath,amsfonts,amssymb,amsthm,hyperref,enumerate,multicol,eufrak}
\usepackage{calligra,indentfirst,epsfig,lettrine}
\usepackage{caption}
\usepackage{relsize}
\usepackage{tabls}
\usepackage{array}
\usepackage{longtable}
\calligra
\usepackage{geometry}
\usepackage{longtable}
\usepackage{array}
\usepackage{lscape} 
\usepackage{booktabs}
\usepackage{tikz}
\usetikzlibrary{shapes.geometric, arrows}

\DeclareFontFamily{U}{mathc}{}
\DeclareFontShape{U}{mathc}{m}{it}%
{<->s*[1.03] mathc10}{}

\DeclareMathAlphabet{\mathscr}{U}{mathc}{m}{it}

\usepackage{comment}
\makeatletter
\newcommand*{\rom}[1]{\expandafter\@slowromancap\romannumeral #1@}
\makeatother
\usepackage{mathtools}
\usepackage{graphicx}
\usepackage{subfigure}
\usepackage{xcolor}
\usepackage{color, soul}
\usepackage{mathrsfs}

\usepackage{tikz}
\usetikzlibrary{shapes.geometric, arrows.meta, positioning} 

 \tikzstyle{startstop} = [rectangle, rounded corners, minimum width=1cm, minimum height=1.2cm,text width=7cm,, draw=black, fill=white!20]

  \tikzstyle{oddeven} = [ellipse, rounded corners, minimum width=3cm, minimum height=0.6cm,text centered, draw=black, fill=white!20]

 \tikzstyle{arrow}=[thick,->,>=stealth]

\DeclareMathAlphabet{\mathpzc}{OT1}{pzc}{m}{it}
\DeclarePairedDelimiter\ceil{\lceil}{\rceil}
\DeclarePairedDelimiter\floor{\lfloor}{\rfloor}
\usepackage{ upgreek }
\newcommand{\etal}{\textit{et al.}}

\newtheorem{definition}{Definition}[section]

\newtheorem{theorem}{Theorem}[section]
\newtheorem{example}{Example}[section]
\newtheorem{lemma}{Lemma}[section]
\newtheorem{remark}{Remark}[section]
\newtheorem{proposition}{Proposition}[section]
\allowdisplaybreaks
\numberwithin{equation}{section}
\begin{document} \title{{Recursive construction and enumeration of self-orthogonal and self-dual codes over finite commutative chain rings of even characteristic}}
\author{Monika Yadav{\footnote{Email address:~\urlstyle{same}\href{mailto:monikay@iiitd.ac.in}{monikay@iiitd.ac.in}}}  ~and 
  Anuradha Sharma{\footnote{Corresponding Author, Email address:~\urlstyle{same}\href{mailto:anuradha@iiitd.ac.in}{anuradha@iiitd.ac.in}} }\\
 {Department of  Mathematics, IIIT-Delhi}\\{New Delhi 110020, India}}
\date{}
\maketitle
\begin{abstract}
Let $\mathscr{R}_{e,m}$ denote a finite commutative chain ring of even characteristic  with  maximal ideal $\langle u \rangle$ of nilpotency index $e \geq 3,$  Teichm$\ddot{u}$ller set $\mathcal{T}_{m},$ and  residue field $\mathscr{R}_{e,m}/\langle u \rangle$ of order $2^m.$ Suppose that  $2 \in \langle u^{\kappa}\rangle \setminus \langle u^{\kappa+1}\rangle$ for some odd integer $\kappa$ with $3 \leq \kappa \leq  e.$   In this paper, we first develop a recursive method to construct a self-orthogonal code $\mathscr{D}_e$ of type $\{\lambda_1, \lambda_2, \ldots, \lambda_e\}$ and length $n$ over  $\mathscr{R}_{e,m}$   from a chain $\mathcal{C}^{(1)}\subseteq \mathcal{C}^{(2)} \subseteq \cdots \subseteq \mathcal{C}^{(\ceil{\frac{e}{2}})} $ of self-orthogonal codes of length $n$ over $\mathcal{T}_{m},$ and vice versa, subject to the following conditions: \begin{itemize}\item[(i)] $\dim \mathcal{C}^{(i)}=\lambda_1+\lambda_2+\cdots+\lambda_i$ for $1 \leq i \leq \ceil{\frac{e}{2}};$ \item[(ii)] the code $\mathcal{C}^{(\floor{\frac{e}{2}}-\floor{\frac{\kappa}{2}})}$ is  doubly even; and  \item[(iii)]  the all-one vector $\textbf{1} =(1,1,\ldots,1)\notin \mathcal{C}^{(\ceil{\frac{e}{2}}-\kappa)}$ provided that $2\kappa\leq e,$  $n\equiv 4\pmod 8$ and $m$ is odd,  \end{itemize} where  $\lambda_1,\lambda_2,\ldots,\lambda_e$ are non-negative integers satisfying $2\lambda_1+2\lambda_2+\cdots+2\lambda_{e-i+1}+\lambda_{e-i+2}+\lambda_{e-i+3}+\cdots+\lambda_i \leq n$ for $\ceil{\frac{e+1}{2}}\leq i\leq e,$ and   
 $\floor{\cdot }$ and $\ceil{\cdot }$ denote the floor and ceiling functions, respectively. This construction ensures that  $Tor_i(\mathscr{D}_e)=\mathcal{C}^{(i)}$ for $1 \leq i \leq \ceil{\frac{e}{2}}.$
  With the help of this recursive construction method and by applying results from group theory and finite geometry, we obtain explicit enumeration formulae for all self-orthogonal and self-dual codes of an arbitrary length over $\mathscr{R}_{e,m}.$ We also illustrate these results with some examples. In a subsequent study \cite{YSub}, we will address the complementary case where  $\kappa$ is even.
 
\end{abstract}
 \section{Introduction}\label{intro}

Self-orthogonal and self-dual codes are among the most significant and widely studied classes of linear codes. These codes exhibit strong connections with design theory \cite{Gaborit,Kennedy} as well as with the theory of modular forms and unimodular lattices \cite{HaradaBannai,Dough,P}. These codes are fundamental in constructing quantum error-correcting codes,  which protect quantum information from decoherence and noise \cite{Ashikhmin,XingJin}.
These codes are also used in designing secret-sharing schemes with desirable access structures  \cite{ Varbanov,MesnagerDougherty}. This has inspired numerous coding theorists to investigate these codes and develop construction methods for these codes \cite{Varbanov,MesnagerDougherty,XingJin,V}.

The study of self-orthogonal and self-dual codes has evolved significantly over the past few decades. In the 1990s, it was discovered that many binary non-linear codes, such as Kerdock, Preparata, Goethals, and Delsarte-Goethals codes, can be represented as Gray images of linear codes over the ring $\mathbb{Z}_4$  of integers modulo $4$ \cite{ R,sole}.  Since then, there has been growing interest in studying self-orthogonal and self-dual codes over finite commutative chain rings \cite{ Choi,Zp,GBN,Fidel,AT,K}. In particular, the explicit enumeration of self-orthogonal and self-dual codes over various finite commutative chain rings has gained significant attention, as these enumeration formulae are instrumental in classifying these codes up to monomial equivalence \cite{Choi,GBN,Galois,Y}. Below, we summarize key results in this direction.

Let  $GR(p^{\mathfrak{s}},m)$ denote the Galois ring of characteristic $p^{\mathfrak{s}}$ and cardinality $p^{\mathfrak{s}m},$ where $p$ is a prime and $m,$  $\mathfrak{s}$ are positive integers. 
By Theorem XVII.5 of \cite{Mcdonald}, any finite commutative chain ring is isomorphic to a quotient ring of the form \vspace{-1mm}\begin{equation}\label{eqCR}\mathscr{R}_{e,m}=\frac{GR(p^{\mathfrak{s}},m)[y]}{\langle h(y),p^{\mathfrak{s}-1}y^{\mathtt{t}}\rangle},\vspace{-1mm}\end{equation}
   where $h(y)$ is an Eisenstein polynomial over $GR(p^{\mathfrak{s}},m)$ of degree $\kappa,$  and the parameter $\mathtt{t}$ satisfies $1\leq \mathtt{t}\leq \kappa$ when $\mathfrak{s}\geq 2$, while  $\mathtt{t}=\kappa$ when $\mathfrak{s}=1.$  The integers $p,$ $ m,$ $\mathfrak{s},$ $\kappa$ and $\mathtt{t}$ are called invariants of the chain ring $\mathscr{R}_{e,m}.$
  
   In the special case  when $\mathfrak{s}=\kappa=\mathtt{t}=1,$ the ring 
$\mathscr{R}_{e,m}$ reduces to the finite field 
$\mathbb{F}_{p^m}$ of order $p^m,$  and the enumeration formulae for self-orthogonal and self-dual codes over $\mathbb{F}_{p^m}$ are obtained by Pless \cite{V}. Further, for a detailed study on the enumeration of self-orthogonal and self-dual codes over finite commutative chain rings of odd characteristic  (\textit{i.e.,} when $p$ is odd),
the reader is referred to \cite{b,BETTY,Zp,GBN,AT,K, Y}  and references therein. 

On the other hand,  when $p=2,$ the following important special cases have been considered in a series of papers:
\begin{itemize}\vspace{-1mm}\item If $m=1,$ $\mathfrak{s}\geq 1$ and $\kappa=1,$ then we have $\mathtt{t}=1$ and $\mathscr{R}_{e,m}\simeq \mathbb{Z}_{2^{\mathfrak{s}}}.$ Enumeration formulae for  self-dual codes over $\mathbb{Z}_{2^{\mathfrak{s}}}$ are obtained by Nagata  {\etal}  \cite{Fidel}.
\vspace{-1mm}\item If $\mathfrak{s}=1,$ then we have $\mathtt{t}=\kappa,$ $GR(2^{\mathfrak{s}},m)=\mathbb{F}_{2^m}$ and $\mathscr{R}_{e,m}\simeq \mathbb{F}_{2^m}[y]/\langle y^{\kappa}\rangle.$   Yadav and Sharma \cite{quasi} provided a recursive method to construct
and enumerate self-orthogonal and self-dual codes over $\mathbb{F}_{2^m}[y]/\langle y^{\kappa}\rangle$ from self-orthogonal  codes over $\mathbb{F}_{2^m}.$ \vspace{-1mm}\item If  $\mathfrak{s}\geq 1$ and $\kappa=1,$ then we have $\mathtt{t}=1$ and $\mathscr{R}_{e,m}\simeq GR(2^{\mathfrak{s}},m).$    Yadav and Sharma \cite{Galois} observed that the recursive construction methods and enumeration techniques employed in \cite{Fidel,quasi,Y} do not directly extend to construct and enumerate self-orthogonal and self-dual codes over Galois rings  of even characteristic.  To address this limitation, they modified their recursive methods, where they constructed  self-orthogonal and self-dual codes over the Galois ring $GR(2^{\mathfrak{s}},m)$ from doubly even codes over the Teichm$\ddot{u}$ller set $\mathcal{T}_m$ of $GR(2^{\mathfrak{s}},m).$  

However, working as in Section $3$  of  \cite{quasi}, we observe that even these modified recursive methods are insufficient to construct and enumerate self-orthogonal and self-dual codes over finite commutative chain rings of even characteristic in the case when $p=2$ and both $ \kappa, s\geq 2,$ as illustrated in Example \ref{example1}. \end{itemize}

The main goal of this paper is to provide a modified recursive method for constructing and enumerating all self-orthogonal and self-dual codes over finite commutative chain rings of even characteristic (\textit{i.e.} $p=2$) when  $\kappa\geq 3$ is odd. The complementary case, where $\kappa$ is even, will be addressed in a subsequent work \cite{YSub}, in which we develop a recursive method for constructing self-orthogonal and self-dual codes over $\mathscr{R}_{e,m}$ from a chain of self-orthogonal codes over finite fields that meet certain additional conditions. 

Together, the aforementioned works provide a complete solution to the problem of enumeration of self-orthogonal and self-dual codes over finite commutative chain rings. Notably, as shown in \cite{Jose, Lavanya}, the resulting enumeration formulae have broader applications, including the enumeration of self-orthogonal and self-dual 
quasi-abelian codes and  self-orthogonal and self-dual  Galois additive cyclic codes over finite commutative chain rings, as well as various other special classes of linear codes that decompose into direct sums of self-orthogonal and self-dual linear codes over finite commutative chain rings. Additionally, these enumeration formulae are instrumental in classifying such codes up to monomial equivalence \cite{BETTY,Choi,GBN,Galois,Y}. It is worth noting that  all monomially equivalent linear codes over finite commutative chain rings have the same homogenous and overweight distances. Consequently, these enumeration formulae are also instrumental in the effective implementation of search algorithms aimed at identifying new codes with optimal parameters \cite{Gassner,Ozbudak}.

This paper is structured
as follows: In Section \ref{prelim}, we  present the necessary preliminaries required for proving our main results.  In Section \ref{construction}, we begin by defining doubly even codes over the Teichm$\ddot{u}$ller set $\mathcal{T}_m$ of the chain ring $\mathscr{R}_{e,m}$ (see Definition \ref{d3.1}).
We then  provide  a  recursive method for  constructing a self-orthogonal (\textit{resp.} self-dual) code $\mathscr{D}_e$ of type  $\{\lambda_1, \lambda_2, \ldots, \lambda_e\}$ and  length $n$ over  $\mathscr{R}_{e,m}$ satisfying $Tor_i(\mathscr{D}_e)=\mathcal{C}^{(i)}$  from a chain \vspace{-2mm}\begin{equation*}\mathcal{C}^{(1)}\subseteq \mathcal{C}^{(2)} \subseteq \cdots \subseteq \mathcal{C}^{(\ceil{\frac{e}{2}})} \vspace{-2mm}\end{equation*} of self-orthogonal codes of length $n$ over $\mathcal{T}_m,$ and vice versa,  under the following conditions: \begin{enumerate}\vspace{-2mm}\item[(i)] $\dim \mathcal{C}^{(i)}=\lambda_1+\lambda_2+\cdots+\lambda_i$ for $1 \leq i \leq \ceil{\frac{e}{2}}$; \vspace{-1mm} \item[(ii)] the code $\mathcal{C}^{(\floor{\frac{e}{2}}-\floor{\frac{\kappa}{2}})}$ is  doubly even; and \vspace{-2mm}\item[(iii)] the all-one vector $\textbf{1} =(1,1,\ldots,1)\notin \mathcal{C}^{(\ceil{\frac{e}{2}}-\kappa)}$ provided  that $2\kappa\leq e,$  $n\equiv 4\pmod 8$ and $m$ is odd,  \end{enumerate} where  $\lambda_1,\lambda_2,\ldots,\lambda_e$ are non-negative integers satisfying $2\lambda_1+2\lambda_2+\cdots+2\lambda_{e-i+1}+\lambda_{e-i+2}+\lambda_{e-i+3}+\cdots+\lambda_i \leq n$ for $\ceil{\frac{e+1}{2}}\leq i\leq e$ (see Propositions \ref{p3.2}--\ref{p3.9KoddReplace} and Figures \ref{F1} and \ref{F2}). 
In Section \ref{counting}, we employ the recursive construction methods outlined in Figures \ref{F1} and \ref{F2} to derive explicit enumeration formulae for all self-orthogonal and self-dual codes of type $\{\lambda_1,\lambda_2,\ldots, \lambda_e\}$ and length $n$  over  $\mathscr{R}_{e,m}$ (Theorems \ref{t4.1Kodd} and  \ref{t4.2Kodd}). Using Remark \ref{Rem1} and  equation \eqref{summation}, these results further yield enumeration formulae for all self-orthogonal and self-dual codes of length $n$ over $\mathscr{R}_{e,m}.$ 
It is worth noting that  when $\mathfrak{s}=1,$ we have $e=\mathtt{t}=\kappa,$   and the ring $\mathscr{R}_{e,m}$ reduces to the quasi-Galois ring $\mathbb{F}_{2^m}[y]/\langle y^{\kappa}\rangle.$ In this special case, for any odd $\kappa \geq 3,$ the  enumeration formulae for self-orthogonal and self-dual codes over $\mathbb{F}_{2^m} \text{[}y\text{]}/\langle y^{\kappa}\rangle$  can be  directly obtained from Theorems \ref{t4.1Kodd} and \ref{t4.2Kodd}  upon  setting $\mathfrak{s}=1$ and $e=\mathtt{t}=\kappa$ (see equation \eqref{summation} and Remarks \ref{rem9.1}  and \ref{rem9.2}).

\vspace{-1mm}\section{Some preliminaries}\label{prelim}
In this section, we will present fundamental properties of finite commutative chain rings and outline the theory of linear codes over such rings, including their generator matrices, duality and torsion codes. We will also explain the structural properties of self-orthogonal and self-dual codes and highlight the limitations of existing recursive constructions for these codes over  finite commutative chain rings of even characteristic that are neither quasi-Galois nor Galois rings.
\vspace{-2mm}\subsection{Finite commutative chain rings}
A finite commutative ring with unity is called a chain ring if  its ideals are totally ordered under the set-theoretic inclusion relation.  Throughout this paper, let $p$ be a prime number, and let $m,$ $\mathfrak{s},$ $\kappa$ and $\mathtt{t}$  be positive integers satisfying $1\leq \mathtt{t} \leq \kappa$ if $\mathfrak{s}\geq 2,$ whereas $\mathtt{t}=\kappa$ if $\mathfrak{s}=1.$ Let $\mathscr{R}_{e,m}$ (as defined in \eqref{eqCR}) be a finite commutative chain ring with invariants  $p,$ $ m,$ $\mathfrak{s},$ $\kappa,$ $\mathtt{t}$ and maximal ideal $\langle u \rangle$ of nilpotency index $e=\kappa (\mathfrak{s}-1)+\mathtt{t}.$   Here, the residue field  $\overline{\mathscr{R}}_{e,m}=\mathscr{R}_{e,m}/\langle u \rangle$ is  the finite field of order $p^m.$  By Theorem XVII.5 of \cite{Mcdonald}, we note that  all ideals of $\mathscr{R}_{e,m}$ form a chain   $\{0\} \subsetneq \langle u^{e-1}\rangle \subsetneq \langle u^{e-2}\rangle \subsetneq  \ldots \subsetneq \langle u\rangle \subsetneq \langle 1 \rangle =\mathscr{R}_{e,m}$ and that the element $p \in \langle u^{\kappa} \rangle \setminus \langle u^{\kappa+1}\rangle.$ We further note, by Lemma  XVII.4 of \cite{Mcdonald}, that $|\langle u^i\rangle|=p^{m(e-i)}$ for $0\leq i\leq e,$ (throughout this paper, $|\cdot|$ denotes the cardinality function). By Theorem XVII.5 of \cite{Mcdonald} again, we see that there exists an element $\xi \in \mathscr{R}_{e,m}$ with multiplicative order  $p^m-1$. Moreover,  the cyclic subgroup generated by $\xi$ is the only subgroup of the unit group of $\mathscr{R}_{e,m}$, which is isomorphic to the multiplicative group of the residue field $\overline{\mathscr{R}}_{e,m}.$  The set \vspace{-2mm}$$\mathcal{T}_{e,m}=\{0,1,\xi,\xi^2,\ldots, \xi^{p^m-2}\}\vspace{-2mm}$$ is called the Teichm$\ddot{u}$ller set of $\mathscr{R}_{e,m}$. By Lemma  XVII.4 of \cite{Mcdonald} again, we see that each element $d\in \mathscr{R}_{e,m}$ can be uniquely expressed as $d=d_{0}+ud_1+\cdots+u^{e-1}d_{e-1},$ where $d_i\in \mathcal{T}_{e,m}$ for $0\leq i\leq e-1,$ (such a representation is called the  Teichm$\ddot{u}$ller representation of $d$). Furthermore, the element $d$  is a unit in $\mathscr{R}_{e,m}$ if and only if $d_0\neq0.$ Now, define a projection map $\pi_0: \mathscr{R}_{e,m} \rightarrow \mathcal{T}_{e,m}$ as $\pi_0(d)=d_0$ for all  $d=d_{0}+ud_1+\cdots+u^{e-1}d_{e-1},$ where $d_i\in \mathcal{T}_{e,m}$ for $0\leq i\leq e-1.$ Next, define a binary operation $\oplus$ on $\mathcal{T}_{e,m}$ as $a\oplus b=\pi_0(a+b)$ for all $a,b \in \mathcal{T}_{e,m}.$ It is easy to see that the set $\mathcal{T}_{e,m}$ is the finite field of order $p^m$ under the addition operation $\oplus$ and the usual multiplication operation of $\mathscr{R}_{e,m}$ and that the finite field $\mathcal{T}_{e,m}$  is isomorphic to the residue field $\overline{\mathscr{R}}_{e,m}$ of $\mathscr{R}_{e,m}.$ 
There also exists a canonical epimorphism  $^{-} : \mathscr{R}_{e,m}\rightarrow \overline{\mathscr{R}}_{e,m},$ defined as  $a\mapsto \bar{a}=a+\langle u\rangle$ for all $a\in \mathscr{R}_{e,m}$. Note that the restriction map $^{-}{\restriction_{\mathcal{T}_{e,m}}}:\mathcal{T}_{e,m}\rightarrow \overline{\mathscr{R}}_{e,m}$  is a field isomorphism.  
\vspace{-2mm}\subsection{Linear codes over $\mathscr{R}_{e,m}$}
Now, let $n$ be a positive integer, and let $\mathscr{R}_{e,m}^n$ be an $\mathscr{R}_{e,m}$-module consisting of all $n$-tuples over $\mathscr{R}_{e,m}.$  A linear code $\mathscr{D}$  of length $n$ over $\mathscr{R}_{e,m}$  is
defined as an $\mathscr{R}_{e,m}$-submodule of $\mathscr{R}_{e,m}^n,$ whose elements  
are referred to as codewords.  A generator matrix of the code $\mathscr{D}$ is  defined as a matrix over $\mathscr{R}_{e,m}$ whose rows form a minimal generating set of the code $\mathscr{D}.$
Further, two linear codes of length $n$ over $\mathscr{R}_{e,m}$ are said to be permutation equivalent if one code can be obtained from the other by  permuting its coordinate positions.  
By Proposition 3.2 of Norton and  S$\check{a}$l$\check{a}$gean \cite{Norton}, we note that every linear code $\mathscr{D}$ of length $n$ over $\mathscr{R}_{e,m}$ is permutation equivalent to a code with a generator matrix in the standard form
\vspace{-1mm} \begin{equation}\label{e1.1} 
 \vspace{-1mm}\begin{aligned}
\mathcal{G}=\begin{bmatrix}
	\mathtt{I}_{\lambda_1}& \mathtt{B}_{1,1}&\cdots& \mathtt{B}_{1,e-2}&\mathtt{B}_{1,e-1}& \mathtt{B}_{1,e} \\
	0 & u\mathtt{I}_{\lambda_2} &\cdots & u\mathtt{B}_{2,e-2}& u\mathtt{B}_{2,e-1} &u\mathtt{B}_{2,e} \\
	\vdots & \vdots &\vdots &\vdots& \vdots&\vdots\\
	0&0&\cdots&u^{e-2}\mathtt{I}_{\lambda_{e-1}}& u^{e-2}\mathtt{B}_{e-1,e-1}&u^{e-2}\mathtt{B}_{e-1,e}\\
	0&0&\cdots &0&u^{e-1}\mathtt{I}_{\lambda_{e}}&u^{e-1}\mathtt{B}_{e,e}
\end{bmatrix}=\begin{bmatrix}
    \mathtt{W}_1\\u\mathtt{W}_2\\ \vdots\\u^{e-2}\mathtt{W}_{e-1}\\u^{e-1}\mathtt{W}_{e}
\end{bmatrix},
\end{aligned}\end{equation}
where the columns of $\mathcal{G}$ are partitioned into blocks of sizes $\lambda_1$, $\lambda_2$, $\ldots$, $\lambda_{e},$ $ \lambda_{e+1}=n-(\lambda_1+\lambda_2+\cdots+\lambda_e),$  $\mathtt{I}_{\lambda_i}$ denotes the $\lambda_i\times\ \lambda_i$ identity matrix over $\mathscr{R}_{e,m},$ and  $\mathtt{B}_{i,j}\in \mathcal{M}_{\lambda_i\times \lambda_{j+1}}(\mathscr{R}_{e,m})$  is considered modulo $u^{j-i+1},$ \textit{i.e.,}   $\mathtt{B}_{i,j}=\mathtt{B}_{i,j}^{(0)}+u\mathtt{B}_{i,j}^{(1)}+\cdots+u^{j-i}\mathtt{B}_{i,j}^{(j-i)}$ with $\mathtt{B}_{i,j}^{(0)},\mathtt{B}_{i,j}^{(1)},\ldots,\mathtt{B}_{i,j}^{(j-i)}\in \mathcal{M}_{\lambda_i\times \lambda_{j+1}}(\mathcal{T}_{e,m})$  for $1\leq i\leq j\leq e,$ (throughout this paper, the set of all $f\times h$ matrices over a ring $R$ is denoted by $\mathcal{M}_{f\times h}(R)$). A linear code $\mathscr{D}$ of length $n$ over $\mathscr{R}_{e,m}$ is said to be of type $\{\lambda_1,\lambda_2,\ldots,\lambda_{e}\}$ if it is permutation equivalent to a code with a generator matrix in the standard form  \eqref{e1.1}.   Further, a linear code $\mathscr{D}$ of type $\{\lambda_1,\lambda_2,\ldots,\lambda_{e}\}$ over $\mathscr{R}_{e,m}$ contains $(p^m)^{\sum\limits_{i=1}^{e}(e-i+1)\lambda_i}$ codewords.

	Next, the Euclidean bilinear form on $\mathscr{R}_{e,m}^n$ is a map  $\cdot: \mathscr{R}_{e,m}^n \times \mathscr{R}_{e,m}^n \to \mathscr{R}_{e,m},$ defined as  $ \textbf{a}\cdot \textbf{b}=\sum\limits_{i=1}^{n}\text{a}_i\text{b}_i$  for all $\textbf{a}=(\text{a}_1,\text{a}_2,\ldots,\text{a}_n)$, $\textbf{b}=(\text{b}_1,\text{b}_2,\ldots,\text{b}_n) \in \mathscr{R}_{e,m}^n.$
   Note that the map $\cdot$ is  a non-degenerate and symmetric bilinear form  on $\mathscr{R}_{e,m}^n.$ 
   The (Euclidean) dual code $\mathscr{D}^{\perp}$ of a linear code $\mathscr{D}$ of length $n$ over $\mathscr{R}_{e,m}$ is defined as $\mathscr{D}^{\perp}=\{\textbf{v}_1\in \mathscr{R}_{e,m}^n ~:~ \textbf{v}_2 \cdot \textbf{v}_1 =0 ~\text{for all }\textbf{v}_2 \in \mathscr{D}\}.$ Note that the dual code $\mathscr{D}^{\perp}$ is also a linear code of  length $n$ over $\mathscr{R}_{e,m}$. Further, by Theorem 3.10 of Norton and  S$\check{a}$l$\check{a}$gean \cite{Norton}, we see that the dual code $\mathscr{D}^{\perp}$   is of type $\{n-(\lambda_1+\lambda_2+\cdots+\lambda_e),\lambda_{e},\lambda_{e-1},\ldots,\lambda_2\}$ if the code $\mathscr{D}$ is of type $\{\lambda_1,\lambda_2,\ldots, \lambda_{e}\}.$ 
  The code $\mathscr{D}$ is said to be self-orthogonal if $\mathscr{D}\subseteq \mathscr{D}^{\perp},$ while  the code $\mathscr{D}$ is said to be  self-dual if $\mathscr{D}= \mathscr{D}^{\perp}$.    
Now, the following theorem provides a necessary and sufficient condition under which a linear code of length 
$n$ over $\mathscr{R}_{e,m}$ is self-orthogonal.

\vspace{-1mm}\begin{lemma}\cite{Dougherty, Y}\label{l2.2}
For  integers $e \geq 2$ and $n \geq 1,$ let $\lambda_1,\lambda_2,\ldots,\lambda_{e+1}$ be non-negative integers satisfying  $n=\lambda_1+\lambda_2+\cdots+\lambda_{e+1}.$ 
	Let $\mathscr{D}$ be a linear code of type $\{\lambda_1,\lambda_2,\ldots,\lambda_e\}$ and  length $n$ over $\mathscr{R}_{e,m}$ with a  generator matrix $\mathcal{G}$ of the form \eqref{e1.1}. The code $\mathscr{D}$ is   self-orthogonal  if and only if 
	\vspace{-2mm}\begin{equation*}\label{e2.3a}
	\mathtt{W}_i\mathtt{W}_j^t\equiv\mathbf{0} \pmod{u^{e-i-j+2}} \text{ ~~for } 1\leq i\leq j\leq e \text{ and }  i+j\leq e+1.
	\vspace{-1mm}\end{equation*}
		Furthermore,  the code $\mathscr{D}$  is   self-dual  if and only if  it is  self-orthogonal and  $\lambda_i=\lambda_{e-i+2}$ for $1 \leq i \leq e.$ (Throughout this paper, $(\cdot)^t$ denotes the matrix transpose.)
		\end{lemma}  
\vspace{-3mm}\subsection{Quotient rings of $\mathscr{R}_{e,m}$ and Teichm$\ddot{u}$ller sets}
For an integer $\ell$ satisfying  $1 \leq \ell < e,$  one can easily see that the quotient ring $\mathscr{R}_{e,m}/\langle u^{\ell} \rangle $ is also a finite commutative chain ring with maximal ideal  $\langle u+\langle u^{\ell} \rangle \rangle $  of nilpotency index $\ell.$  From now on, we shall denote the chain ring $\mathscr{R}_{e,m}/\langle u^{\ell} \rangle $ by $\mathscr{R}_{\ell,m}$ for our convenience. 
Furthermore,  the element $\xi_{\ell} := \xi +\langle u^{\ell}\rangle\in \mathscr{R}_{\ell,m}$ has  multiplicative order $p^m-1$ and the set  $\mathcal{T}_{\ell,m}=\{0,1, \xi_{\ell}, \xi_{\ell}^2,\ldots, \xi_{\ell}^{p^m-2}\}$ is   the Teichm$\ddot{u}$ller set of $\mathscr{R}_{\ell,m}.$ The canonical epimorphism from $\mathscr{R}_{e,m}$ onto $\mathscr{R}_{\ell,m}$ is defined as $a \mapsto a+\langle u^{\ell}\rangle$ for all $a \in \mathscr{R}_{e,m}.$  In view of this, we shall identify each element $a +\langle u^{\ell}\rangle \in \mathscr{R}_{\ell,m}$ with the element $a \in \mathscr{R}_{e,m},$ performing addition and multiplication in $\mathscr{R}_{\ell,m}$ modulo $u^{\ell}.$ Specifically, we shall identify the element $\xi_{\ell}\in \mathcal{T}_{\ell,m}$ with  the element $\xi \in \mathcal{T}_{e,m}.$ Accordingly, we assume,  throughout this paper, that 
\vspace{-1.5mm}\begin{equation*}\vspace{-1mm}\mathcal{T}_{1,m}=\mathcal{T}_{2,m}=\cdots=\mathcal{T}_{e-1,m}=\mathcal{T}_{e,m}=\{0,1,\xi,{\xi}^2,\ldots, {\xi}^{p^m-2}\}=\mathcal{T}_{m}\text{ (say)}\end{equation*} and \vspace{-1mm}\begin{equation*}\vspace{-1mm}\mathscr{R}_{1,m}=\overline{\mathscr{R}}_{1,m}=\overline{\mathscr{R}}_{2,m}=\cdots=\overline{\mathscr{R}}_{e-1,m}=\overline{\mathscr{R}}_{e,m}=\mathcal{T}_{m} .\end{equation*} 
In view of this,  we see, for $1 \leq \ell \leq e,$ that each element $b \in \mathscr{R}_{\ell,m}$  can be uniquely written  as $b=b_0+ub_1 +u^2b_2 +\cdots+u^{\ell-1}b_{\ell-1},$ where $b_0,b_1,b_2,\ldots,b_{\ell-1} \in \mathcal{T}_{m}.$ Further, each matrix $Z \in \mathcal{M}_{f\times h}(\mathscr{R}_{\ell,m}) $ can be uniquely written  as $Z=Z_0+uZ_1+u^2Z_2+\cdots+u^{\ell-1}Z_{\ell-1},$ where $Z_0,Z_1,\ldots,Z_{\ell-1} \in \mathcal{M}_{f\times h}(\mathcal{T}_{m}).$  
\vspace{-2mm} \subsection{Torsion codes over $\mathcal{T}_m$ and self-orthogonal and self-dual codes over $\mathscr{R}_{e,m}$}

The set $\mathcal{T}_{m}^n,$ consisting of all $n$-tuples over $\mathcal{T}_{m}$, can be viewed  as an $n$-dimensional vector space over the finite field $\mathcal{T}_{m}$ under the component-wise addition induced by $\oplus$ and the component-wise scalar multiplication.   We further define a map $B_{m}:\mathcal{T}_{m}^n \times \mathcal{T}_{m}^n \rightarrow \mathcal{T}_{m}$  as $B_{m}(\textbf{v},\textbf{w})=\pi_0(\textbf{v}\cdot \textbf{w})=\pi_0(\sum\limits_{i=1}^{n}v_iw_i)$ for all $\textbf{v}=(v_1,v_2,\ldots,v_n),$ $\textbf{w}=(w_1,w_2,\ldots,w_n) \in \mathcal{T}_{m}^n,$ (here $\textbf{v}\cdot \textbf{w}=\sum\limits_{i=1}^{n}v_iw_i$ is viewed as an element of $\mathscr{R}_{e,m}$).  Note that the map $B_m$  is a symmetric and non-degenerate  bilinear form on $\mathcal{T}_{m}^n.$  A linear code $\mathcal{C}$ of length $n$ over $\mathcal{T}_{m}$ is defined as a $\mathcal{T}_{m}$-linear subspace of $\mathcal{T}_{m}^n.$ The dual code of $\mathcal{C}$ with respect to $B_m$ is defined as $\mathcal{C}^{\perp_{B_{m}}}=\{\textbf{w} \in \mathcal{T}_{m}^n : B_{m}(\textbf{w},\textbf{c})=0 \text{ for all } \textbf{c} \in \mathcal{C}\}.$ Note that $\mathcal{C}^{\perp_{B_{m}}}$ is also a linear code of length $n$ over $\mathcal{T}_{m}.$
  Further, a linear code $\mathcal{C}$ of length $n$ over $\mathcal{T}_{m}$ is said to be (i) self-orthogonal if  it satisfies $\mathcal{C}\subseteq \mathcal{C}^{\perp_{B_{m}}},$ and (ii) self-dual if  it satisfies $\mathcal{C}= \mathcal{C}^{\perp_{B_{m}}}.$

    For  a linear code $\mathscr{D}$ of length $n$ over $\mathscr{R}_{e,m},$ the $i$-th torsion code  of $\mathscr{D}$  is defined as  
\vspace{-2mm}\begin{equation*} Tor_i(\mathscr{D})=\{\textbf{w} \in \mathcal{T}_m^n: u^{i-1} (\textbf{w}+u\textbf{w}^{\prime})\in \mathscr{D} \text{ for some } \textbf{w}^{\prime}\in \mathscr{R}_{e,m}^n \}\vspace{-2mm}\end{equation*} for $1 \leq i \leq e.$    Note that $Tor_i(\mathscr{D})$ is a linear code of length $n$ over $\mathcal{T}_{m}$ for each $i.$ Additionally, if the code $\mathscr{D}$ has a generator matrix $\mathcal{G}$ of the form \eqref{e1.1}, then the $i$-th torsion code $Tor_i(\mathscr{D})$  has dimension  $\lambda_1+\lambda_2+\cdots + \lambda_i$  over $\mathcal{T}_m$ and  has a generator matrix 
 \vspace{-2mm}\begin{equation}\label{e1.2}\vspace{-1mm}
	\begin{aligned}
	\begin{bmatrix}
	\mathtt{I}_{\lambda_1}& \mathtt{B}_{1,1}^{(0)}& \mathtt{B}_{1,2}^{(0)}&  \cdots& \mathtt{B}_{1,i-1}^{(0)}& \cdots&\mathtt{B}_{1,e-1}^{(0)}&\mathtt{B}_{1,e}^{(0)}\\
	0&\mathtt{I}_{\lambda_2}& \mathtt{B}_{2,2}^{(0)}&\cdots & \mathtt{B}_{2,i-1}^{(0)}& \cdots&\mathtt{B}_{2,e-1}^{(0)}&\mathtt{B}_{2,e}^{(0)}\\
	\vdots&\vdots&\vdots&\vdots& \vdots& \vdots &\vdots&\vdots\\
	0&0&0&\cdots &\mathtt{I}_{\lambda_i}&\cdots &\mathtt{B}_{i,e-1}^{(0)}&\mathtt{B}_{i,e}^{(0)}
	\end{bmatrix}.
	\end{aligned}
	\end{equation}
	 Note that $Tor_i(\mathscr{D})\subseteq Tor_{i+1}(\mathscr{D})$ for $1\leq i \leq e-1$ and  $|\mathscr{D}|=\prod\limits_{i=1}^{e}|Tor_i(\mathscr{D})|.$ Now,  to count all self-orthogonal and self-dual codes over $\mathscr{R}_{e,m}$  of a given type, we need the following lemma.
		\begin{lemma}\cite{Dougherty, Y}\label{l1.2} For  a  self-orthogonal code  $\mathscr{D}$ of length $n$ over $\mathscr{R}_{e,m},$  the following hold.
				\begin{enumerate}[(a)]
			 	\vspace{-1mm}\item 	\label{l1.2a}$Tor_i(\mathscr{D})\subseteq Tor_i(\mathscr{D})^{\perp_{B_m}}~\text{for }1\leq i \leq \floor{\frac{e+1}{2}}.$
				\vspace{-1mm}\item  	\label{l1.2b}$Tor_i(\mathscr{D})\subseteq Tor_{e-i+1}(\mathscr{D})^{\perp_{B_m}}~\text{for } \floor{\frac{e+1}{2}}+1\leq i \leq e.$
					\vspace{-1mm}\end{enumerate}
			Furthermore, if the code $\mathscr{D}$ is self-dual, then  we have $Tor_i(\mathscr{D})=Tor_{e-i+1}(\mathscr{D})^{\perp_{B_m}}$ for $\ceil{\frac{e+1}{2}}\leq i\leq e.$ (Throughout this paper,  $\floor{\cdot }$ and $\ceil{\cdot }$ denote the floor and ceiling functions, respectively.)
	\end{lemma}
\begin{remark}\label{Rem1} From the above lemma, it follows that 
	if $\mathscr{D}$ is a  self-orthogonal code  of type $\{\lambda_1,\lambda_2,\ldots,\lambda_e\}$  and length $n$ over $\mathscr{R}_{e,m},$ then $2\lambda_1+2\lambda_2+\cdots+2\lambda_{e-i+1}+\lambda_{e-i+2}+\cdots+\lambda_i \leq n$ for $\ceil{\frac{e+1}{2}}\leq i\leq e,$ which implies  that  $n\geq2 (\lambda_1+\lambda_2+\cdots+\lambda_{\frac{e}{2}})+\lambda_{\frac{e}{2}+1}$ if $e$ is even, while $n\geq 2 (\lambda_1+\lambda_2+\cdots+\lambda_{\frac{e+1}{2}})$ if $e$ is odd.\end{remark}
\vspace{-2mm} \subsection{Limitations of existing constructions} 
    Yadav and Sharma  \cite{Galois,quasi,Y} counted all self-orthogonal and self-dual codes of length $n$ over $\mathscr{R}_{e,m}$ in the following three cases: (i) $p$ is an odd prime (\textit{i.e.,} the ring $\mathscr{R}_{e,m}$ is of odd characteristic), (ii) $p=2$ and $\mathfrak{s}=1$ (\textit{i.e.,} $\mathscr{R}_{e,m}$ is the quasi-Galois ring $\mathbb{F}_{2^m}[u]/\langle u^e \rangle$ of even characteristic), and (iii) $p=2$ and $\kappa=1$ (\textit{i.e.,} $\mathscr{R}_{e,m}$ is the Galois ring $GR(2^{e},m)$ of  characteristic $2^{e}$ and cardinality $2^{em}$). 
In view of Remark \ref{Rem1}, let $\lambda_1, \lambda_2, \ldots,  \lambda_{e}$ be non-negative integers satisfying $2\lambda_1+2\lambda_2+\cdots+2\lambda_{e-i+1}+\lambda_{e-i+2}+\cdots+\lambda_i \leq n$ for $\ceil{\frac{e+1}{2}}\leq i\leq e.$ When either  $p=2$ and $\mathfrak{s}=1$ or $p$ is odd, Yadav and Sharma   \cite[Sec. 4]{quasi} and  \cite[Sec. 5]{Y}   provided recursive methods for constructing self-orthogonal and self-dual codes of type $\{\lambda_1,\lambda_1,\ldots,\lambda_e\}$ and length $n$ over $\mathscr{R}_{e,m}$ from a self-orthogonal code of length $n$ and dimension $\lambda_1+\lambda_2+\cdots+\lambda_{\ceil{\frac{e}{2}}}$ over $\mathcal{T}_{m}.$   On the other hand, when $p=2$ and $\kappa=1,$ Yadav and Sharma  \cite[Sec. 4]{Galois} provided a recursive method for constructing self-orthogonal and self-dual codes of type $\{\lambda_1,\lambda_1,\ldots,\lambda_e\}$ and length $n$ over $\mathscr{R}_{e,m}$ from a doubly even code of length $n$ and dimension $\lambda_1+\lambda_2+\cdots+\lambda_{\ceil{\frac{e}{2}}}$ over $\mathcal{T}_{m}.$   It is worth noting that when $p=2$ and either $\mathfrak{s}=1$ or $\kappa=1,$ the recursive method provided in \cite[Sec. 5]{Y} does not always lift a self-orthogonal code over $\mathcal{T}_m$ to a self-orthogonal code over $\mathscr{R}_{e,m}$ (see   \cite[Sec. 1]{Galois} and \cite[Sec. 3]{quasi}). When $p=2$ and both $\mathfrak{s}, \kappa\geq 2,$ the following example illustrates that the recursive methods provided in  \cite[Sec. 4]{Galois}   and \cite[Sec. 4]{quasi} do not always lift a  self-orthogonal code over $\mathcal{T}_m$ and a doubly even code over $\mathcal{T}_m$ (see \cite[Def. 2.1]{Galois})  to a self-orthogonal code over $\mathscr{R}_{e,m},$ respectively.

    \begin{example}\label{example1} Let  $\mathscr{R}_{8,2}=\frac{GR(8,2)[y]}{\langle y^3+2,4y^2\rangle}$  be a finite commutative chain ring with  maximal ideal $\langle u \rangle $ of  nilpotency index $8,$   where $u:= y+\langle y^3+2, 4y^2 \rangle \in \mathscr{R}_{8,2}.$ Note that  $u^3+u^6=2$ in $\mathscr{R}_{8,2}.$
 Let $n=4,$ $\lambda_1=1$  and $\lambda_2=\lambda_3=\lambda_4=\lambda_5=\lambda_6=\lambda_7=\lambda_8=0.$ Let $\mathcal{C}_0$ be a self-orthogonal  code of length $4$ and dimension $\lambda_1+\lambda_2+\lambda_3+\lambda_4=1$ over $\mathcal{T}_2$ with a generator matrix $\begin{bmatrix}1& 1&1&1\end{bmatrix}.$   Note that $\pi_{1}(\textbf{d}\cdot \textbf{d})=\pi_{3}(\textbf{d}\cdot \textbf{d})=0$ for all $\textbf{d}\in \mathcal{C}_0,$ where each $\textbf{d}\cdot \textbf{d}$ is viewed as an element of $\mathscr{R}_{8,2}.$
 \begin{description}\item[I.]  ~~~Here, we will show that the code  $\mathcal{C}_0$ can not be lifted to a self-orthogonal code of type $\{1,0,0,0,0,0,0,0\}$ and length $4$ over  $\mathscr{R}_{8,2}$ using the recursive construction method provided in Section 4  of Yadav and Sharma \cite{quasi}. Towards this, let us consider a linear code $\mathcal{C}_2$ of type $\{1,0\} $ and length $4$ over $\mathscr{R}_{2,2}$ with a generator matrix $\begin{bmatrix}
    1&1&1&1
\end{bmatrix}+u\begin{bmatrix}
    0&a_1&b_1&c_1
\end{bmatrix},  $  where $a_1,b_1,c_1\in \mathcal{T}_2.$  Using the recursive construction method provided in Section 4 of \cite{quasi}, we see that the code $\mathcal{C}_0$ can be lifted to a self-orthogonal code of type $\{1,0,0,0,0,0,0,0\}$ over $\mathscr{R}_{8,2}$ if and only if $\mathcal{C}_2$ is a   self-orthogonal code   over $\mathscr{R}_{2,2}$ satisfying property $(\ast)$ (see \cite[Sec. 4]{quasi}), which holds if and only if   $a_1^2+b_1^2+c_1^2=0.$ 
In particular, for $a_1=b_1=c_1=0,$ the code $\mathcal{C}_2$ of length $4$ over $\mathscr{R}_{2,2}$ with a generator matrix $\begin{bmatrix}1& 1&1&1\end{bmatrix}$ is a self-orthogonal code satisfying property $(\ast).$ Corresponding to this particular choice of the code $\mathcal{C}_2$,  consider a linear code $\mathcal{C}_4$ of type $\{1,0,0,0\} $ and length $4$ over $\mathscr{R}_{4,2}$ with a generator matrix \vspace{-2mm}\begin{equation*}\vspace{-1mm}
    \begin{bmatrix}
    1&1&1&1
\end{bmatrix}+u^2\begin{bmatrix}
0&a_2&b_2&c_2\end{bmatrix}+u^3\begin{bmatrix}
    0&a_3&b_3&c_3
\end{bmatrix},  \end{equation*} where $a_2,b_2,c_2,a_3,b_3,c_3\in \mathcal{T}_2.$ Using again the recursive construction method provided in Section 4 of \cite{quasi},  we see that the code $\mathcal{C}_4$ can be lifted to a self-orthogonal code of type $\{1,0,0,0,0,0,0,0\}$ over $\mathscr{R}_{8,2}$ if and only if $\mathcal{C}_4$ is a   self-orthogonal code   over $\mathscr{R}_{4,2}$ satisfying property $(\ast),$ which holds if and only if  $a_2^2+b_2^2+c_2^2=0$ and $a_3^2+b_3^2+c_3^2=0.$ In particular, for $a_2=b_2=c_2=a_3=b_3=c_3=0,$ we see that the code $\mathcal{C}_4$ of length $4$ over $\mathscr{R}_{4,2}$ with a generator matrix $\begin{bmatrix}1& 1&1&1\end{bmatrix}$ is a self-orthogonal code satisfying property $(\ast).$
Further, corresponding to this particular choice of the code $\mathcal{C}_4$,  consider a linear code $\mathcal{C}_6$ of type $\{1,0,0,0,0,0\} $ and length $4$ over $\mathscr{R}_{6,2}$ with a generator matrix \vspace{-1mm}\begin{equation*}
    \vspace{-1mm}\begin{bmatrix}
    1&1&1&1
\end{bmatrix}+u^4\begin{bmatrix}
0&a_4&b_4&c_4\end{bmatrix}+u^5\begin{bmatrix}
    0&a_5&b_5&c_5
\end{bmatrix}, \end{equation*} where $a_4,b_4,c_4,a_5,b_5,c_5\in \mathcal{T}_2.$ By the recursive construction method provided in Section 4 of \cite{quasi}, we see that the code $\mathcal{C}_6$ can be lifted to a self-orthogonal code of type $\{1,0,0,0,0,0,0,0\}$ over $\mathscr{R}_{8,2}$ if and only if $\mathcal{C}_6$ is a   self-orthogonal code   over $\mathscr{R}_{6,2}$ satisfying property $(\ast),$ which holds for all  choices of $a_4,b_4,c_4,a_5,b_5,c_5\in \mathcal{T}_2.$ In particular, for $a_4=b_4=c_4=a_5=b_5=c_5=0,$ we see that the  code $\mathcal{C}_6$ of length $4$ over  $\mathscr{R}_{6,2}$ with a generator matrix $\begin{bmatrix}
    1&1&1&1
\end{bmatrix}$ is a self-orthogonal code   over $\mathscr{R}_{6,2}$ satisfying property $(\ast).$ Finally,  consider a linear code $\mathcal{C}_8$ of type $\{1,0,0,0,0,0,0,0\} $ and length $4$ over $\mathscr{R}_{8,2}$ with a generator matrix \begin{equation*}
    \begin{bmatrix}
    1&1&1&1
\end{bmatrix}+u^6\begin{bmatrix}
0&a_6&b_6&c_6\end{bmatrix}+u^7\begin{bmatrix}
    0&a_7&b_7&c_7
\end{bmatrix},  \end{equation*}  where $a_6,b_6,c_6,a_7,b_7,c_7\in \mathcal{T}_2.$ Note that the code $\mathcal{C}_8$ is not self-orthogonal for any choice of $a_6,b_6,c_6,a_7,b_7,c_7\in \mathcal{T}_2.$ This illustrates that the recursive construction method provided in Section 4  of Yadav and Sharma \cite{quasi} does not always lift a self-orthogonal code over $\mathcal{T}_m$ to a self-orthogonal code over $\mathscr{R}_{e,m}.$  

\item[II.] ~~We will next show that the code $\mathcal{C}_0$ can not be lifted to a self-orthogonal code of type $\{1,0,0,0,0,0,0,0\}$ and length $4$ over  $\mathscr{R}_{8,2}$ using the recursive construction method provided in Section 4  of Yadav and Sharma \cite{Galois}. For this, we first recall \cite[Def. 2.1]{Galois} that  a self-orthogonal code $\mathscr{D}$ over $\mathcal{T}_2$ is said to be doubly even if it satisfies $\pi_1(\textbf{d}\cdot \textbf{d})=0$ for all $\textbf{d}\in \mathscr{D},$ where  each $\textbf{d}\cdot \textbf{d}$ is viewed as an element of $\mathscr{R}_{8,2}$. 
One can easily see that $\pi_1(\textbf{v}\cdot \textbf{v})=0$ for all $\textbf{v}\in \mathcal{C}_0,$ so the code $\mathcal{C}_0$ is doubly even.   Now, let us consider a linear code $\mathcal{D}_2$ of type $\{1,0\} $ and length $4$ over $\mathscr{R}_{2,2}$ with a generator matrix $\begin{bmatrix}
    1&1&1&1
\end{bmatrix}+u\begin{bmatrix}
    0&f_1&g_1&h_1
\end{bmatrix},  $  where $f_1,g_1,h_1\in \mathcal{T}_2.$ By the recursive construction method provided in Section 4 of \cite{Galois}, we see that the code $\mathcal{D}_2$ can be lifted to a self-orthogonal code of type $\{1,0,0,0,0,0,0,0\}$ over $\mathscr{R}_{8,2}$ if and only if $\mathcal{D}_2$ is a  $\lambda_1$-doubly even self-orthogonal code   over $\mathscr{R}_{2,2}$ (see \cite[Def. 2.1]{Galois}), which  holds  if and only if $f_1^2+g_1^2+h_1^2\equiv 0\pmod u.$ In particular, for $f_1=g_1=h_1=0,$ we see that the code $\mathcal{D}_2$  of length $4$ over $\mathscr{R}_{2,2}$ with a generator matrix $\begin{bmatrix}
    1&1&1&1
\end{bmatrix}$ is a $\lambda_1$-doubly even  self-orthogonal code over $\mathscr{R}_{2,2}.$ Corresponding to this particular choice of the code $\mathcal{D}_2$, let us consider a linear code $\mathcal{D}_4$ of type $\{1,0,0,0\}$ and length $4$ over $\mathscr{R}_{4,2}$ with a generator matrix 
\vspace{-1mm}\begin{equation*}\begin{bmatrix}
    1&1&1&1
\end{bmatrix}+u^2\begin{bmatrix}
0&f_2&g_2&h_2\end{bmatrix}+u^3\begin{bmatrix}
    0&f_3&g_3&h_3
\end{bmatrix},  \end{equation*}  where $f_2,g_2,h_2,f_3,g_3,h_3\in \mathcal{T}_2.$ By the recursive construction method provided in Section 4 of \cite{Galois},  we see that the code $\mathcal{D}_4$ can be lifted to a self-orthogonal code of type $\{1,0,0,0,0,0,0,0\}$ over $\mathscr{R}_{8,2}$ if and only if $\mathcal{C}_4$ is a   $\lambda_1$-doubly even  self-orthogonal code   over $\mathscr{R}_{4,2},$  which holds if and only if  $f_2^2+g_2^2+h_2^2\equiv 0\pmod u.$  In particular, for $f_2=g_2=h_2=f_3=g_3=h_3=0,$ we see that the code $\mathcal{D}_4$ of length $4$ over $\mathscr{R}_{4,2}$ with a generator matrix $\begin{bmatrix}1& 1&1&1\end{bmatrix}$ is a $\lambda_1$-doubly even self-orthogonal code. Now, corresponding to this particular choice of the code $\mathcal{D}_4$,  consider a linear code $\mathcal{D}_6$ of type $\{1,0,0,0,0,0\} $ and length $4$ over $\mathscr{R}_{6,2}$ with a generator matrix \vspace{-1mm}\begin{equation*}
    \begin{bmatrix}
    1&1&1&1
\end{bmatrix}+u^4\begin{bmatrix}
0&f_4&g_4&h_4\end{bmatrix}+u^5\begin{bmatrix}
    0&f_5&g_5&h_5
\end{bmatrix}, \end{equation*} where $f_4,g_4,h_4,f_5,g_5,h_5\in \mathcal{T}_2.$ By the recursive construction method provided in Section 4 of \cite{Galois}, we see that the code $\mathcal{D}_6$ can be lifted to a self-orthogonal code of type $\{1,0,0,0,0,0,0,0\}$ over $\mathscr{R}_{8,2}$ if and only if $\mathcal{D}_6$ is a   $\lambda_1$-doubly even self-orthogonal code   over $\mathscr{R}_{6,2}.$ 
One can easily observe that $\mathcal{D}_6$ is not a $\lambda_1$-doubly even code for any choice of $f_4,g_4,h_4,f_5,g_5,h_5\in \mathcal{T}_2.$   This illustrates  that the recursive construction method provided in Section 4  of Yadav and Sharma \cite{Galois} does not always lift a doubly even code over $\mathcal{T}_m$ to a self-orthogonal code over $\mathscr{R}_{e,m}.$  
\end{description}
\end{example}
\vspace{-2mm}\subsection{Notations and preliminary setup}
From now on, we assume, throughout this paper, that  $p=2$  and the integer $\kappa\geq 3$ is odd, \textit{i.e.,}  the element $2 \in \langle u^{\kappa}\rangle \setminus \langle u^{\kappa+1}\rangle$ with $\kappa \geq 3$ being odd.  Accordingly, we assume, from now on, that $2 \equiv u^{\kappa}(\eta_0+u\eta_1+u^2\eta_2+\cdots+u^{e-1-\kappa}\eta_{e-1-\kappa})\pmod{u^{e}},$ where $\eta_0\in \mathcal{T}_m\setminus \{0\}$ and $\eta_1,\eta_2,\ldots,\eta_{e-1-\kappa}\in \mathcal{T}_{m}.$ In the following section, we will provide a  recursive method for lifting a self-orthogonal code over $\mathcal{T}_m$ to a self-orthogonal code over $\mathscr{R}_{e,m}$ when $\kappa\geq 3$ is odd;  the case when $\kappa$ is even will be addressed in a subsequent work \cite{YSub}. For notational convenience,  we define $\mathrm{f}_w=\floor[\big]{\frac{w}{2}}$ and  $\mathrm{c}_w=\ceil[\big]{\frac{w}{2}}$ for any positive integer $w,$  $\theta_e= 0$ if $e$ is even and $\theta_e =1$ if $e$ is odd, $s=\mathrm{f}_e=\floor[\big]{\frac{e}{2}},$ and  $\kappa_1=\mathrm{f}_\kappa=\floor[\big]{\frac{\kappa}{2}}=\frac{\kappa-1}{2}.$   
Unless otherwise stated, we will  adhere to the notations introduced in Section \ref{prelim}.
 \section{A recursive method for constructing self-orthogonal and self-dual codes over $\mathscr{R}_{e,m}$ from  a chain of self-orthogonal codes over $\mathcal{T}_m$}\label{construction}
  Throughout this section, let $n$ be a positive integer, and   let $\lambda_1, \lambda_2, \ldots,  \lambda_{e+1}$ be non-negative integers satisfying $n=\lambda_1+\lambda_2+\cdots+\lambda_{e+1}$ and  $2\lambda_1+2\lambda_2+\cdots+2\lambda_{e-i+1}+\lambda_{e-i+2}+\cdots+\lambda_i \leq n$ for $s+1\leq i\leq e.$       Let us define $\Lambda_0=0$  and 
$\Lambda_i=\lambda_1+\lambda_2+\cdots+\lambda_i $   for $1\leq i \leq e+1 .$ Note that $\Lambda_{e+1}=n.$ Further, for $ 1\leq \ell \leq e$ and  $0 \leq i \leq \ell-1,$ let us  define a mapping $\pi_i: \mathscr{R}_{\ell,m} \rightarrow \mathcal{T}_{m}$ as $\pi_i(a)=a_i$ for all  $a=a_0+ua_1+\cdots+u^{\ell-1}a_{\ell-1}\in \mathscr{R}_{\ell,m} ,$ where $a_0,a_1,\ldots,a_{\ell-1} \in \mathcal{T}_{m}.$ We further recall, from Section \ref{prelim}, that $(\mathcal{T}_m,\oplus,\cdot)$ is the finite field of order $2^m.$  Now, if  $\mathcal{C}$ is a self-orthogonal code of length $n$ over $\mathcal{T}_m,$ then for all $\textbf{v}\in \mathcal{C},$ one can easily see that $\pi_{0}(\textbf{v}\cdot \textbf{v})=0$ and $\pi_{2j+1}(\textbf{v}\cdot \textbf{v})=0$ for $ 0\leq j \leq \kappa_1-1,$ where $\kappa_1=\mathrm{f}_\kappa=\frac{\kappa-1}{2}$ and each $\textbf{v}\cdot \textbf{v}$ is viewed as an element of $\mathscr{R}_{e,m}.$ We next define a special class of self-orthogonal codes over $\mathcal{T}_m.$ 
\begin{definition}\label{d3.1}
  A  self-orthogonal  code $\mathcal{D}$ of length $n$  over $\mathcal{T}_m$ is said to be doubly even if it satisfies  $\pi_{\kappa}(\textbf{b}\cdot \textbf{b}) = 0 $  for all $\textbf{b} \in \mathcal{D},$ where each $\textbf{b}\cdot \textbf{b}$ is viewed as an element of $\mathscr{R}_{e,m}.$ 
  Moreover, a self-orthogonal code $\mathcal{D}$ of length $n$ over $\mathcal{T}_m$ is said to be $\mathfrak{d}$-doubly even if it has a $\mathfrak{d}$-dimensional doubly even subcode.
\end{definition}

For $e \geq 3$ and $2 \leq \ell \leq e,$ we further define a special class of linear codes of length $n$ over $\mathscr{R}_{\ell,m}$  as follows:

\vspace{2mm} \noindent \textbf{Linear codes over $\mathscr{R}_{\ell,m}$ satisfying property $(\mathfrak{P})$:} \textit{Let $e, \ell $ be integers satisfying $e \geq 3 ,$ $2 \leq \ell\leq  e$ and $\ell \equiv e\pmod2.$ Let us define $\gamma_{\ell}=s-\mathrm{f}_\ell=\floor[\big]{\frac{e}{2}}-\floor[\big]{\frac{\ell}{2}}.$   Let $\mathscr{D}_{\ell}$ be a linear code of type $\{\Lambda_{\gamma_{\ell}+1},\lambda_{\gamma_{\ell}+2}, \ldots, \lambda_{\gamma_{\ell}+\ell}\}$ and length $n$ over $\mathscr{R}_{\ell,m}$ with a generator matrix 
\vspace{-3mm}\small{\begin{equation*}\label{e3.0}
\mathcal{G}_{\ell}=\begin{bmatrix}
\mathtt{W}_1\\\mathtt{W}_2 \\ \vdots\\ \mathtt{W}_{\gamma_{\ell}+1}\vspace{0.5mm}\\ u\mathtt{W}_{\gamma_{\ell}+2}\\ \vdots\\u^{\ell-1}\mathtt{W}_{\gamma_{\ell}+\ell}
\end{bmatrix},\end{equation*}} \normalsize
where $\mathtt{W}_i\in \mathcal{M}_{\lambda_i\times n}(\mathscr{R}_{\ell,m})$ for $1 \leq i \leq \gamma_{\ell}+1$ and  the matrix $\mathtt{W}_{\gamma_{\ell}+j}\in \mathcal{M}_{\lambda_{\gamma_{\ell}+j}\times n}(\mathscr{R}_{\ell,m})$ is to be considered modulo $u^{\ell-j+1}$ for  $2 \leq j \leq \ell.$ 
\begin{enumerate}
    \item[(i)] When $\ell \leq \min\{\kappa-1,e-\kappa\} ,$ we say that the code $\mathscr{D}_{\ell}$ satisfies property $(\mathfrak{P})$ if 
   \hspace{-3mm} \begin{eqnarray*}
    \displaystyle \mathtt{Diag}\left([\mathtt{W}]_{\gamma_{\ell}+1-\mathrm{f}_i}[\mathtt{W}]_{\gamma_{\ell}+1-\mathrm{f}_i}^t\right) &\equiv & \mathbf{0} \pmod{u^{\ell+i}} \text{ ~~for ~} i\in \{2,4,6,\ldots,\ell-\theta_e\}~\text{ and}\\ 
   \hspace{-2mm} \displaystyle  \pi_{\ell+j}\big( \mathtt{Diag}\big([\mathtt{W}]_{\gamma_{\ell}+1-\mathrm{f}_{j+2}}[\mathtt{W}]_{\gamma_{\ell}+1-\mathrm{f}_{j+2}}^t\big) \big)&=& \mathbf{0} \text{ ~for~~ } j \in \{ \ell-1,\ell+1,\ldots, \kappa -2+\theta_e  \}.
    \end{eqnarray*}
     \item[(ii)] When $e-\kappa< \ell \leq \kappa-\mathrm{f}_{2\kappa-e}+1 ,$  we say that the code $\mathscr{D}_{\ell}$ satisfies property $(\mathfrak{P})$ if 
     \begin{eqnarray*}
   \displaystyle \mathtt{Diag}\left([\mathtt{W}]_{\gamma_{\ell}+1-\mathrm{f}_i}[\mathtt{W}]_{\gamma_{\ell}+1-\mathrm{f}_i}^t\right) &\equiv &\mathbf{0} \pmod{u^{\ell+i}} \text{ ~~~for ~} i\in \{2,4,6,\ldots,\ell-\theta_e\}\text{ and}\\ 
     \displaystyle  \pi_{\ell+j}\left( \mathtt{Diag}\big([\mathtt{W}]_{\gamma_{\ell}+1-\mathrm{f}_{j+2}}[\mathtt{W}]_{\gamma_{\ell}+1-\mathrm{f}_{j+2}}^t\big)\right)&=& \mathbf{0} \text{ ~ for~ } j \in \{ \ell-1,\ell+1,\ell+3,\ldots, e-\ell-1-\theta_e\}.
    \end{eqnarray*}
     \item[(iii)] When $\kappa \leq \ell \leq e-\kappa,$ we say that the code $\mathscr{D}_{\ell}$ satisfies property $(\mathfrak{P})$ if 
    \begin{eqnarray*}
   \displaystyle  \mathtt{Diag}\left([\mathtt{W}]_{\gamma_{\ell}+1-\mathrm{c}_i}[\mathtt{W}]_{\gamma_{\ell}+1-\mathrm{c}_i}^t\right) &\equiv &\mathbf{0} \pmod{u^{\ell+i}} \text{ ~~for ~} i\in \{2,4,6,\ldots,\kappa-1\} \cup \{\kappa\}.
    \end{eqnarray*}
    \item[(iv)] When $\ell > \max\{e-\kappa,\kappa-\mathrm{f}_{2\kappa-e}+1\},$  we say that the code $\mathscr{D}_{\ell}$ satisfies property $(\mathfrak{P})$ if 
    \begin{eqnarray*}
   \displaystyle  \mathtt{Diag}\left([\mathtt{W}]_{\gamma_{\ell}+1-\mathrm{f}_i}[\mathtt{W}]_{\gamma_{\ell}+1-\mathrm{f}_i}^t\right) &\equiv& \mathbf{0} \pmod{u^{\ell+i}} \text{ ~~for ~} i\in \{2,4,6,\ldots,e-\ell\}.
    \end{eqnarray*}(Here, all matrices are considered  over $\mathscr{R}_{e,m}$ and all matrix multiplications are performed over $\mathscr{R}_{e,m}.$ Throughout this paper,  for any square matrix $X,$  $\mathtt{Diag}(X)$ denotes the diagonal matrix whose principal diagonal is the same as that of the matrix $X.$ Also, for any positive integer $v,$ let $[Y]_v$ denote the column block matrix with $z$-th block  $Y_z$ for $1 \leq z \leq v.$)
    \end{enumerate}}
By carrying out  direct computations, one can easily observe that linear codes over $\mathscr{R}_{\ell,m}$ satisfying property $(\mathfrak{P})$ coincide with  linear codes over $\mathscr{R}_{\ell,m}$ satisfying property $(\ast)$ when $\mathfrak{s}=1$ 
    \cite[Sec. 4]{quasi}.
    Next, for  positive integers  $v$ and $w \leq e,$ let  $(\mathtt{H})_{v,w} $ denote the block matrix whose  $(i,j)$-th block is  the matrix $\mathtt{H}_{i,j}\in \mathcal{M}_{\lambda_{i}\times \lambda_{j+1}}(\mathcal{T}_{m})$   for $1 \leq i \leq v$ and $w \leq j \leq e.$ Further, if $\mathscr{D}$ is a  self-orthogonal code of type $\{\lambda_1,\lambda_2,\ldots, \lambda_e\}$ and length $n$ over $\mathscr{R}_{e,m}$,   then we see, by Lemma \ref{l1.2},     that the torsion code $Tor_{s+\theta_e}(\mathscr{D})$ is also a  self-orthogonal code of length $n$  and  dimension $\Lambda_{s+\theta_e}$ over $\mathcal{T}_{m}(=\overline{\mathscr{R}}_{e,m}).$ We now make the following observation.
\vspace{-2mm}\begin{remark}\label{r3.1}\cite{Galois,quasi}
Every  self-orthogonal code of length $n$ and dimension $\Lambda_{s+\theta_e}$ over  $\mathcal{T}_{m}$ is permutation equivalent to a self-orthogonal code with a generator matrix 
\vspace{-1.5mm}\begin{equation*}
G =\begin{bmatrix}
\mathtt{I}_{\lambda_1}&\mathtt{B}_{1,1}&\mathtt{B}_{1,2}&\cdots&\mathtt{B}_{1,s+\theta_e-1}&\mathtt{B}_{1,s+\theta_e}&\cdots &\mathtt{B}_{1,e}\\
 \mathbf{0} &\mathtt{I}_{\lambda_2} &\mathtt{B}_{2,2} &\cdots &\mathtt{B}_{2,s+\theta_e-1}&\mathtt{B}_{2,s+\theta_e}&\cdots &\mathtt{B}_{2,e}\\
	\vdots& \vdots &\vdots &\vdots &\vdots &\vdots& \vdots&\vdots\\
 \mathbf{0}& \mathbf{0}& \mathbf{0}&\cdots&  \mathtt{I}_{\lambda_{s+\theta_e}}&\mathtt{B}_{s+\theta_e,s+\theta_e}& \cdots & \mathtt{B}_{s+\theta_e,e}
\end{bmatrix}, \vspace{-1mm} \end{equation*} 
where  the columns  of $G$ are partitioned into blocks of sizes $\lambda_1,\lambda_2,\ldots, \lambda_{e+1}, $ $\mathtt{I}_{\lambda_i}$ denotes the $\lambda_i\times \lambda_i$ identity matrix over $\mathcal{T}_{m},$   $\mathtt{B}_{i,j} \in \mathcal{M}_{\lambda_i\times \lambda_{j+1}}(\mathcal{T}_{m})$  for $1 \leq i \leq s+\theta_e$ and $i \leq j \leq e,$  and the matrices 
$(\mathtt{B})_ { s,s+\theta_e+1} $, $(\mathtt{B})_{ s-1,s+\theta_e+2},\ldots, (\mathtt{B})_{2,e-1},$  $\mathtt{B}_{1,e}$ are of full row ranks. \end{remark}
Now, let $\mathscr{D}_e$ be a linear code of type $\{\lambda_1,\lambda_2,\ldots,\lambda_e\}$ and length $n$ over $\mathscr{R}_{e,m}$ with a generator matrix \vspace{-2mm}\small{\begin{equation}\label{Ge}
    \mathcal{G}_{e}=\begin{bmatrix}
            \mathtt{W}^{(e)}_1\\  u\mathtt{W}^{(e)}_2\\
              u^2\mathtt{W}^{(e)}_3\\
              \vdots\\u^{e-1}  \mathtt{W}^{(e)}_e\\
        \end{bmatrix},\vspace{-1mm}
    \end{equation}} \normalsize  where 
    the matrix  $\mathtt{W}^{(e)}_j \in \mathcal{M}_{\lambda_j\times n }(\mathscr{R}_{e,m})$ is of the form  $\mathtt{W}^{(e)}_j=\mathtt{W}^{(0)}_j+u\mathtt{V}^{(1)}_j+u^2\mathtt{V}^{(2)}_j+\cdots+u^{e-j}\mathtt{V}^{(e-j)}_j$ with $\mathtt{W}^{(0)}_j,\mathtt{V}^{(1)}_j,\mathtt{V}^{(2)}_j,\ldots, \mathtt{V}^{(e-j)}_j\in \mathcal{M}_{\lambda_j\times n}(\mathcal{T}_m)$
     for $1 \leq j \leq e.$ Now, for each integer $\ell$ satisfying $2 \leq \ell \leq e-2$ and  $\ell\equiv e\pmod 2,$    let us define $\gamma_{\ell}=s-\mathrm{f}_{\ell},$ and let  us consider a linear code $\mathscr{D}_{\ell}$ of type $\{\Lambda_{\gamma_{\ell}+1},\lambda_{\gamma_{\ell}+2}, \ldots, \lambda_{\gamma_{\ell}+\ell}\}$ and length $n$ over $\mathscr{R}_{\ell,m}$    with a generator matrix 
\vspace{-2mm}\small{\begin{equation}\label{Gl}
\mathcal{G}_{\ell}=\begin{bmatrix}
\mathtt{W}_1^{(\ell)}\\\mathtt{W}_2^{(\ell)} \\ \vdots\\ \mathtt{W}_{\gamma_{\ell}+1}^{(\ell)}\vspace{0.5mm}\\ u\mathtt{W}_{\gamma_{\ell}+2}^{(\ell)}\\ \vdots\\u^{\ell-1}\mathtt{W}_{\gamma_{\ell} +\ell}^{(\ell)}
\end{bmatrix},\end{equation} } \normalsize
where   $\mathtt{W}_h^{(\ell)} =\mathtt{W}^{(0)}_h+u\mathtt{V}_h^{(1)}+\cdots+ u^{\ell-1}\mathtt{V}_h^{(\ell-1)}$ \big(\textit{i.e.,} $\mathtt{W}_h^{(\ell)}\equiv\mathtt{W}_h^{(e)} \pmod{ u^{\ell}}$\big) for $1 \leq h\leq \gamma_{\ell}+1$ and  $\mathtt{W}_{\gamma_{\ell}+j}^{(\ell)}=\mathtt{W}^{(0)}_{\gamma_{\ell}+j}+u\mathtt{V}_{\gamma_{\ell}+j}^{(1)}+\cdots+ u^{\ell-j}\mathtt{V}_{\gamma_{\ell}+j}^{(\ell-j)}$ \big(\textit{i.e.,} $\mathtt{W}_{\gamma_{\ell}+j}^{(\ell)}\equiv \mathtt{W}_{\gamma_{\ell}+j}^{(e)} \pmod{ u^{\ell-j+1}}$\big)  for   $2 \leq j \leq \ell.$ In the following lemma, we establish a necessary condition under which a self-orthogonal (\textit{resp.} self-dual) code  over $\mathscr{R}_{\ell,m}$  can be lifted to a self-orthogonal (\textit{resp.} self-dual) code over $\mathscr{R}_{e,m},$ where the integer $\ell$ satisfies $2 \leq \ell \leq e-2$ and $\ell \equiv e \pmod 2.$

\begin{lemma}\label{l3.2} Let $\mathscr{D}_{e}$ be a linear code  of type $\{\lambda_1,\lambda_2,\ldots,\lambda_e\}$ and length $n$ over $\mathscr{R}_{e,m}$ with a generator matrix $\mathcal{G}_{e}$ (as defined in \eqref{Ge}). For   each integer $\ell$ satisfying $2 \leq \ell \leq e-2$ and $\ell\equiv e \pmod 2,$ let $\mathscr{D}_{\ell}$ be the corresponding linear code of type  $\{\Lambda_{\gamma_{\ell}+1},\lambda_{\gamma_{\ell}+2}, \ldots, \lambda_{\gamma_{\ell}+\ell}\}$ and length $n$ over $\mathscr{R}_{\ell,m}$ with a generator matrix $\mathcal{G}_{\ell}$ (as defined in \eqref{Gl}), where  $\gamma_{\ell}=s-\mathrm{f}_{\ell}.$ If the code $\mathscr{D}_{e}$ is self-orthogonal (\textit{resp.} self-dual), then $\mathscr{D}_{\ell}$ is also a self-orthogonal  (\textit{resp.} self-dual) code over $\mathscr{R}_{\ell,m}$ satisfying   property $(\mathfrak{P}).$ 
    \end{lemma}
\begin{proof} To prove the result,  let  $\ell$ be a  fixed integer satisfying $2 \leq \ell \leq e-2$ and  $\ell \equiv e\pmod 2 .$ Here, we see, working as in Theorem 5.1 of Yadav and Sharma \cite{Y},  that if the code $\mathscr{D}_e$ is  self-orthogonal (\textit{resp.} self-dual), then the code  $\mathscr{D}_{i}$ is also  self-orthogonal (\textit{resp.} self-dual) for any integer  $ i < e$  with $i\equiv e\pmod 2,$ and in particular,  the code $\mathscr{D}_{\ell}$ is  self-orthogonal (\textit{resp.} self-dual).   Now, it remains to show that  the code $\mathscr{D}_{\ell}$ satisfies property $(\mathfrak{P}).$ For this, we will distinguish the following two cases:  (I) $2\kappa \leq e,$  and (II) $2\kappa >e.$  
\begin{description}
  \item[(I)] Let  $2\kappa  \leq e.$  Here, we will consider the following three cases separately: (i) $2\leq \ell \leq \kappa-1,$ (ii) $\kappa \leq \ell \leq e-\kappa,$ and (iii) $\ell \geq e-\kappa+1.$ 
\begin{description}\item[(i)] Let $2 \leq \ell \leq \kappa-1.$  Since $\mathscr{D}_{\ell}$ is a self-orthogonal code over $\mathscr{R}_{\ell,m},$ we have  $\mathtt{Diag}\big([\mathtt{W}^{(\ell)}]_{\gamma_{\ell}+1}[\mathtt{W}^{(\ell)}]_{\gamma_{\ell}+1}^t\big)\equiv \mathbf{0}\pmod{u^{\ell}}.$  Let us write $\mathtt{Diag}\big([\mathtt{W}^{(\ell)}]_{\gamma_{\ell}+1}[\mathtt{W}^{(\ell)}]_{\gamma_{\ell}+1}^t\big)\equiv u^{\ell}\mathtt{E}^{(\ell)}+u^{\ell+1}\mathtt{E}^{(\ell+1)}+\cdots+u^{\ell+\kappa-1}\mathtt{E}^{(\ell+\kappa-1)}\pmod{u^{\ell+\kappa}},$  where $\mathtt{E}^{(\ell)},\mathtt{E}^{(\ell+1)},\ldots,\mathtt{E}^{(\ell+\kappa-1)}\in \mathcal{M}_{(\gamma_{\ell}+1)\times (\gamma_{\ell}+1)}(\mathcal{T}_m).$ Further, let $\mathtt{E}^{(b)}_{g,g}$ denote the $(g,g)$-th entry of the matrix $\mathtt{E}^{(b)}$ for  $1 \leq g \leq \gamma_{\ell}+1$ and $ \ell \leq b \leq \ell+\kappa-1.$  

Now, for  $i\in\{2,4,\ldots,\ell-\theta_e\},$  as $\mathscr{D}_{\ell+i}$ is a self-orthogonal code over $\mathscr{R}_{\ell+i,m},$ we see,  by Lemma \ref{l2.2},  that  $\mathtt{Diag}\big([\mathtt{W}^{(\ell+i)}]_{\gamma_{\ell}+1-\mathrm{f}_i}[\mathtt{W}^{(\ell+i)}]_{\gamma_{\ell}+1-\mathrm{f}_i}^t\big)\equiv \mathbf{0}\pmod{u^{\ell+i}},$  
from which one can easily observe that $$\mathtt{Diag}\big([\mathtt{W}^{(\ell)}]_{\gamma_{\ell}+1-\mathrm{f}_i}[\mathtt{W}^{(\ell)}]_{\gamma_{\ell}+1-\mathrm{f}_i}^t\big)\equiv \mathtt{Diag}\big([\mathtt{W}^{(\ell+i)}]_{\gamma_{\ell}+1-\mathrm{f}_i}[\mathtt{W}^{(\ell+i)}]_{\gamma_{\ell}+1-\mathrm{f}_i}^t\big)\equiv\mathbf{0}\pmod{u^{\ell+i}}.$$ Further, for $j \in \{\ell-1, \ell+1,\ldots,\kappa-2+\theta_e\},$  since $\mathscr{D}_{\ell+j+1+\theta_e}$ is a self-orthogonal code over $\mathscr{R}_{\ell+j+1+\theta_e,m},$ by Lemma \ref{l2.2} again, we get \begin{equation}\label{eqndiag}
\mathtt{Diag}\big([\mathtt{W}^{(\ell+j+1+\theta_e)}]_{\gamma_{\ell}+1-\mathrm{f}_{j+2}}[\mathtt{W}^{(\ell+j+1+\theta_e)}]_{\gamma_{\ell}+1-\mathrm{f}_{j+2}}^t\big)\equiv \mathbf{0}\pmod{u^{\ell+j+1+\theta_e}}.
\end{equation}   We further observe, for $1 \leq a \leq \gamma_{\ell}+1-\mathrm{f}_{j+2},$  that the $(a,a)$-th entry, say $L_a$, of the  matrix  $\mathtt{Diag}\big([\mathtt{W}^{(\ell+j+1+\theta_e)}]_{\gamma_{\ell}+1-\mathrm{f}_{j+2}}[\mathtt{W}^{(\ell+j+1+\theta_e)}]_{\gamma_{\ell}+1-\mathrm{f}_{j+2}}^t\big) $  satisfies    
\begin{eqnarray*} L_a\equiv u^{2\ell}\big(\mathtt{E}^{(2\ell)}_{a,a}+  \textbf{v}_a^{(\ell)}\cdot \textbf{v}_a^{(\ell)}\big)+u^{2\ell+1}\mathtt{E}^{(2\ell+1)}_{a,a}+\cdots+u^{\ell+j-2}\mathtt{E}^{(\ell+j-2)}_{a,a}+u^{\ell+j-1}\big(\mathtt{E}^{(\ell+j-1)}_{a,a}~~~~~~~~~~~~~~~~~\\ +    \textbf{v}_a^{(\mathtt{f}_{\ell+j-1})}\cdot \textbf{v}_a^{(\mathtt{f}_{\ell+j-1})}\big)  +u^{\ell+j}\mathtt{E}^{(\ell+j)}_{a,a}\pmod{ u^{\ell+j+1}}\end{eqnarray*}   if $\ell$ is even, while it satisfies
\begin{eqnarray*}L_a\equiv u^{2\ell-1}\mathtt{E}^{(2\ell-1)}_{a,a}+  u^{2\ell}\big(\mathtt{E}^{(2\ell)}_{a,a}+  \textbf{v}_a^{(\ell)}\cdot \textbf{v}_a^{(\ell)}\big)+u^{2\ell+1}\mathtt{E}^{(2\ell+1)}_{a,a}+\cdots+u^{\ell+j}\mathtt{E}^{(\ell+j)}_{a,a}+u^{\ell+j+1}\big(\mathtt{E}^{(\ell+j+1)}_{a,a}\\+    \textbf{v}_a^{(\mathtt{f}_{\ell+j+1})}\cdot \textbf{v}_a^{(\mathtt{f}_{\ell+j+1})}\big) \pmod{ u^{\ell+j+2}}\end{eqnarray*} if $\ell$ is odd,
where $\mathbf{v}_a^{(b)}$ denotes the $a$-th row  of the matrix $[\mathtt{V}^{(b)}]_{\gamma_{\ell}+1}$ for  $1 \leq a \leq \gamma_{\ell}+1-\mathtt{f}_{j+2}$ and $\ell \leq b \leq \mathtt{f}_{\ell+j-1+2\theta_e}.$ We next observe, by \eqref{eqndiag},  that $\mathtt{E}^{(\ell+j)}_{a,a}=0$ for $1 \leq a \leq \gamma_{\ell}+1-\mathtt{f}_{j+2},$ 
which   implies that   $$\pi_{\ell+j}\big(\mathtt{Diag}\big([\mathtt{W}^{(\ell)}]_{\gamma_{\ell}+1-\mathrm{f}_{j+2}}[\mathtt{W}^{(\ell)}]_{\gamma_{\ell}+1-\mathrm{f}_{j+2}}^t\big)\big)=\mathbf{0}$$ for $j \in \{\ell-1, \ell+1,\ldots,\kappa-2+\theta_e\}.$

\item[(ii)] Next, let $\kappa \leq \ell \leq e-\kappa. $ Here, for  $i\in\{2,4,\ldots,\kappa-1\}\cup \{\kappa\},$ as $\mathscr{D}_{\ell+i}$ is a self-orthogonal code over $\mathscr{R}_{\ell+i,m},$ by Lemma \ref{l2.2},  we get $\mathtt{Diag}\big([\mathtt{W}^{(\ell+i)}]_{\gamma_{\ell}+1-\mathrm{c}_i}[\mathtt{W}^{(\ell+i)}]_{\gamma_{\ell}+1-\mathrm{c}_i}^t\big)\equiv \mathbf{0}\pmod{u^{\ell+i}},$ which implies that $$\mathtt{Diag}\big([\mathtt{W}^{(\ell)}]_{\gamma_{\ell}+1-\mathrm{c}_i}[\mathtt{W}^{(\ell)}]_{\gamma_{\ell}+1-\mathrm{c}_i}^t\big)\equiv \mathbf{0}\pmod{u^{\ell+i}}.$$

\item[(iii)] Finally, let  $\ell \geq e-\kappa+1.$ Here, for  $i\in\{2,4,\ldots,e-\ell\},$  since $\mathscr{D}_{\ell+i}$ is a self-orthogonal code over $\mathscr{R}_{\ell+i,m},$  by Lemma \ref{l2.2}, we get $\mathtt{Diag}\big([\mathtt{W}^{(\ell+i)}]_{\gamma_{\ell}+1-\mathrm{f}_i}[\mathtt{W}^{(\ell+i)}]_{\gamma_{\ell}+1-\mathrm{f}_i}^t\big)\equiv \mathbf{0}\pmod{u^{\ell+i}},$ which implies that $$\mathtt{Diag}\big([\mathtt{W}^{(\ell)}]_{\gamma_{\ell}+1-\mathrm{f}_i}[\mathtt{W}^{(\ell)}]_{\gamma_{\ell}+1-\mathrm{f}_i}^t\big)\equiv \mathbf{0}\pmod{u^{\ell+i}}.$$ \end{description}  This shows that  the code $\mathscr{D}_{\ell}$ satisfies property $(\mathfrak{P}).$ 
\item[(II)] When $2\kappa>e,$ the desired result follows by working as in case (I).
\vspace{-4mm}\end{description}
 \vspace{-4mm}  \end{proof} 
From the above lemma, we see, for $2 \leq \ell \leq e-2$ and  $\ell \equiv e\pmod 2 ,$ that a necessary condition for lifting a linear code $\mathscr{D}_{\ell}$ over $\mathscr{R}_{\ell,m}$ to a self-orthogonal (\textit{resp.} self-dual) code $\mathscr{D}_{e}$ over $\mathscr{R}_{e,m}$ is that   $\mathscr{D}_{\ell}$ itself must be a self-orthogonal (\textit{resp.} self-dual) code over $\mathscr{R}_{\ell,m}$ and must additionally satisfy property $(\mathfrak{P}).$
We next prove the following lemma.

\begin{lemma}\label{l3.1} Let  $\ell$ be a fixed  integer satisfying  $\ell \equiv e \pmod 2$ and $4 \leq \ell \leq \kappa+2$ if $2\kappa \leq e,$ while the integer $\ell$ satisfies $\ell \equiv e \pmod 2$ and $4 \leq \ell \leq e-\kappa+1-2\theta_e $ if $2\kappa >e.$ 
  Let us define $\gamma_{\ell}=s-\mathrm{f}_{\ell}.$ If there exists  a self-orthogonal code $\mathscr{D}_{\ell}$ of type $\{\Lambda_{\gamma_{\ell}+1},\lambda_{\gamma_{\ell}+2}, \ldots, \lambda_{\gamma_{\ell}+\ell}\}$ and length $n$ over $\mathscr{R}_{\ell,m}$ satisfying property $(\mathfrak{P}),$ then there exists a $\Lambda_{s-\kappa_1}$-dimensional doubly even subcode of the torsion code $Tor_1(\mathscr{D}_{\ell}).$  
 \end{lemma}
\begin{proof}To prove the result, we assume, without any loss of generality, that the code  
$\mathscr{D}_{\ell}$ has a generator matrix $\mathcal{G}_{\ell}$ as defined in \eqref{Gl}.
Note that the torsion code $Tor_1(\mathscr{D}_{\ell})$ is a $\Lambda_{\gamma_{\ell}+1}$-dimensional linear code over 
$\mathcal{T}_m$ with a generator matrix $[\mathtt{W}^{(0)}]_{\gamma_{\ell}+1}.$ Now, let $\mathcal{C}^{(s-\kappa_1)}$ be a $\Lambda_{s-\kappa_1}$-dimensional linear subcode of 
$Tor_1(\mathscr{D}_{\ell})$  with a generator matrix $[\mathtt{W}^{(0)}]_{s-\kappa_1}.$  We assert  that  $\mathcal{C}^{(s-\kappa_1)}$ is a doubly even code over $\mathcal{T}_m.$

To prove this assertion, we see,  by Lemma \ref{l1.2}, that $\mathcal{C}^{(s-\kappa_1)}$ is a self-orthogonal code over $\mathcal{T}_m.$ We further note, by   Lemma  \ref{l2.2}, that if  $\mathscr{D}_{\ell}$ is a self-orthogonal code of type $\{\Lambda_{\gamma_{\ell}+1},\lambda_{\gamma_{\ell}+2}, \ldots, \lambda_{\gamma_{\ell}+\ell}\}$ and length $n$ over $\mathscr{R}_{\ell,m}$ satisfying property $(\mathfrak{P}),$ then 
 \begin{eqnarray*}
  \displaystyle \pi_{i}\Big(\mathtt{Diag}\big([\mathtt{W}^{(\ell)}]_{\gamma_{\ell}+1}[\mathtt{W}^{(\ell)}]_{\gamma_{\ell}+1}^t\big)\Big)&= & \mathbf{0} \text{ ~for ~} i\in \{1,3,\ldots,\ell-1-\theta_e\} \text{ and}\\
  \pi_{\ell+j-\theta_e}\left( \mathtt{Diag}\big([\mathtt{W}^{(\ell)}]_{\gamma_{\ell}+1-\mathrm{f}_{j+2-\theta_e}}[\mathtt{W}^{(\ell)}]_{\gamma_{\ell}+1-\mathrm{f}_{j+2-\theta_e}}^t\big) \right)&=& \mathbf{0} \text{ ~for~ } j \in \{1,3,\ldots,\kappa-2  \}.
    \end{eqnarray*}
  We further  note that 
$\pi_{g}\big(\mathtt{Diag}\big([\mathtt{W}^{(0)}]_{s-\kappa_1}[\mathtt{W}^{(0)}]_{s-\kappa_1}^t \big)\big)=  \pi_{g}\big(\mathtt{Diag}\big([\mathtt{W}^{(\ell)}]_{s-\kappa_1}[\mathtt{W}^{(\ell)}]_{s-\kappa_1}^t\big) \big)= \mathbf{0}$  for all $ g \in \{ 1,3,\ldots,\kappa\}.$ From this, it follows that  $\mathcal{C}^{(s-\kappa_1)}$ is a doubly even code over $\mathcal{T}_m.$ 
\end{proof}

In the following lemma, we show that if $\mathscr{D}_{e}$ is a self-orthogonal code  of type $\{\lambda_1,\lambda_2,\ldots,\lambda_e\}$ and length $n$ over $\mathscr{R}_{e,m},$ then  $Tor_{s+\theta_e}(\mathscr{D}_e)$ is a $\Lambda_{s-\kappa_1}$-doubly even  code of length $n$ and dimension $\Lambda_{s+\theta_e}$ over $\mathcal{T}_m$ and   has $Tor_{s-\kappa_1}(\mathscr{D}_{e})$ as  its $\Lambda_{s-\kappa_1}$-dimensional doubly even subcode.    

\begin{lemma}\label{l3.3} Let $\mathscr{D}_{e}$ be a self-orthogonal code  of type $\{\lambda_1,\lambda_2,\ldots,\lambda_e\}$ and length $n$ over $\mathscr{R}_{e,m}$ with a generator matrix $\mathcal{G}_{e}$ (as defined in \eqref{Ge}). The torsion code $Tor_{s+\theta_e}(\mathscr{D}_e)$ is a $\Lambda_{s-\kappa_1}$-doubly even code of length $n$ and dimension $\Lambda_{s+\theta_e}$ over $\mathcal{T}_m$ with its $\Lambda_{s-\kappa_1}$-dimensional doubly even subcode as  $Tor_{s-\kappa_1}(\mathscr{D}_{e}).$  \end{lemma}
\begin{proof}
 To prove the result,  we note, by Lemma \ref{l1.2},   that the torsion code $Tor_{s+\theta_e}(\mathscr{D}_e)$ is a self-orthogonal code of length $n$ and dimension $\Lambda_{s+\theta_e}$ over $\mathcal{T}_m.$ Here, we assert  that the torsion  code $Tor_{s-\kappa_1}(\mathscr{D}_e)$ is a doubly even code, \textit{i.e.,} $\pi_{\kappa}(v\cdot v)=0$ for all $v \in Tor_{s-\kappa_1}(\mathscr{D}_e),$ where each $v\cdot v$ is viewed as an element of $\mathscr{R}_{e,m}.$ To prove this assertion, let $a \in Tor_{s-\kappa_1}(\mathscr{D}_e)$ be fixed.  Thus, there exist $a_1,a_2,\ldots,a_{s+\kappa_1+\theta_e}\in \mathcal{T}_m^n$ such that  $u^{s-\kappa_1-1}(a+ua_1+u^2a_2+\cdots+ u^{s+\kappa_1+\theta_e}a_{s+\kappa_1+\theta_e})\in \mathscr{D}_e.$ Now, since $\mathscr{D}_e$ is a self-orthogonal code over $\mathscr{R}_{e,m},$ we have $u^{2s-2\kappa_1-2}(a+ua_1+\cdots+ u^{s+\kappa_1+\theta_e}a_{s+\kappa_1+\theta_e})\cdot (a+ua_1+\cdots+ u^{s+\kappa_1+\theta_e}a_{s+\kappa_1+\theta_e})\equiv 0 \pmod{u^e}.$ This implies that $(a+ua_1+\cdots+ u^{s+\kappa_1+\theta_e}a_{s+\kappa_1+\theta_e})\cdot (a+ua_1+\cdots+ u^{s+\kappa_1+\theta_e}a_{s+\kappa_1+\theta_e})\equiv 0 \pmod{u^{\kappa+1+\theta_e}},$ which further implies that $a\cdot a +u^2 (a_1\cdot a_1) +\cdots+ u^{\kappa-1}(a_{\kappa_1}\cdot a_{\kappa_1})+\theta_eu^{\kappa+1}(a_{\kappa_1+1}\cdot a_{\kappa_1+1}+\eta_0a\cdot a_1)\equiv 0\pmod{u^{\kappa+1+\theta_e}}.$ From this, one can easily see that $ \pi_{\kappa}(a\cdot a)= \pi_{\kappa}\left(a\cdot a +u^2 (a_1\cdot a_1) +\cdots+ u^{\kappa-1}(a_{\kappa_1}\cdot a_{\kappa_1})+\theta_eu^{\kappa+1}(a_{\kappa_1+1}\cdot a_{\kappa_1+1}+\eta_0a\cdot a_1)\right)=0.$ This shows that $Tor_{s-\kappa_1}(\mathscr{D}_e)$ is a doubly even code over $\mathcal{T}_m.$ From this, the desired result follows.
\end{proof}
From the above lemma, we see that  a necessary condition under which a linear code $\mathscr{D}_0$ of length $n$ and dimension $\Lambda_{s+\theta_e}$ over $\mathcal{T}_m$  can be lifted to a self-orthogonal code over $\mathscr{R}_{e,m}$ is that the code $\mathscr{D}_0$ must be  a $\Lambda_{s-\kappa_1}$-doubly even code over $\mathcal{T}_m.$
Next, let $\mathcal{T}$ be the Teichm$\ddot{u}$ller set  of the Galois ring $GR(2^{e},m)$ of characteristic $2^e$ and rank $m.$ Now, a linear code $C$  of length $n$ over $\mathcal{T}$ is said to be doubly even if   it satisfies $\textbf{c}\cdot \textbf{c} \equiv 0 \pmod{4}$ for all $\textbf{c} \in C,$ where each $\textbf{c}\cdot \textbf{c}$ is viewed as an element of $GR(2^{e},m)$ (see \cite[Def. 2.1]{Galois}). 
 Yadav and Sharma \cite[Th. 3.2 and 3.3]{Galois} obtained  enumeration formulae for  the number $\widehat{\sigma}_m(n;\mathfrak{d})$ of  doubly even codes of length $n$ and dimension $\mathfrak{d}$ over $\mathcal{T}$ containing the all-one vector and the number $\sigma_m(n;\mathfrak{d})$ of  doubly even codes of length $n$ and dimension $\mathfrak{d}$ over $\mathcal{T}$ that do not contain the all-one vector.
In the following lemma, we establish a 1-1 correspondence between doubly even codes over $\mathcal{T}_m$ and  doubly even codes over the  Teichm$\ddot{u}$ller set $\mathcal{T}$ of $GR(2^{e},m).$ 

\begin{lemma}\label{LEM}There is a  1-1 correspondence between doubly even codes of length $n$ over $\mathcal{T}_m$ and  doubly even codes of the same length $n$ over the  Teichm$\ddot{u}$ller set $\mathcal{T}$ of $GR(2^{e},m).$  
\end{lemma} 
\begin{proof}To prove the result, we first  note that $\mathcal{T}\simeq \mathbb{F}_{2^m}\simeq \mathcal{T}_m,$ and let $\phi$ be the field isomorphism from $\mathcal{T}$ onto $\mathcal{T}_m.$ The isomorphism $\phi$ can be extended component-wise to an isomorphism from $\mathcal{T}^n$ onto $\mathcal{T}_m^n,$ which we shall denote by $\phi$ itself for the sake of simplicity. The isomorphism $\phi$ induces a 1-1 correspondence between  linear codes of length $n$ over $\mathcal{T}$ and linear codes of length $n$ over $\mathcal{T}_m.$

Now, let $C$ be a doubly even code of length $n$  over $\mathcal{T}. $ We assert that $\phi(C)$ is a doubly even code over $\mathcal{T}_m.$ To prove this assertion, it is easy to see that the code $\phi(C)$ is self-orthogonal.   Next, let $\textbf{c}=(c_1,c_2,\ldots,c_n)\in C.$  Since the code $C$ is doubly even over $\mathcal{T},$ we have  $\textbf{c}\cdot \textbf{c}\equiv 0\pmod 4,$ where $\textbf{c}\cdot \textbf{c}$ is viewed as an element of $GR(2^e,m).$ This implies that $(\sum\limits_{i=1}^n c_i)^2-2\hspace{-1mm}\sum\limits_{1\leq i<j\leq n}\hspace{-1mm}c_ic_j \equiv 0\pmod 4,$ which further implies that $\pi_0(-\hspace{-1mm}\sum\limits_{1\leq i<j\leq n}\hspace{-1mm}c_ic_j)=0.$  As $\phi(\textbf{c})\cdot \phi(\textbf{c})$ can be viewed as an element of $\mathscr{R}_{e,m},$ we observe that  $\pi_{\kappa}\big(\phi(\textbf{c})\cdot \phi(\textbf{c})\big)=\eta_0\pi_0\big(-\hspace{-1mm}\sum\limits_{1\leq i<j\leq n}\hspace{-1mm}\phi(c_i)\phi(c_j)\big)=\eta_0\phi\big(\pi_0(-\hspace{-1mm}\sum\limits_{1\leq i<j\leq n}\hspace{-1mm}c_ic_j)\big)=\phi(0)=0,$ which implies that the code $\phi(C)$ is doubly even over $\mathcal{T}_m.$ 

Conversely, let  $D$ be a doubly even code of length $n$ over $\mathcal{T}_m.$ There exists a linear code $D^{\prime}$ of length $n$ over $\mathcal{T}$ such that $\phi(D^{\prime})=D.$ We claim that the code $D^{\prime}$ is doubly even.
It is easy to see that the code $D^{\prime}$ is  self-orthogonal.  
Now, let $\textbf{d}=(d_1,d_2,\ldots,d_n)\in D^{\prime}$ be fixed.  Since $D$ is doubly even, we have 
$\pi_{\kappa}\big(\phi(\textbf{d})\cdot \phi(\textbf{d})\big)=\eta_0 \pi_0\big(-\sum\limits_{1\leq i<j\leq n}\hspace{-1mm}\phi(d_i)\phi(d_j)\big)=0.$ This implies that $\pi_0\big(-\sum\limits_{1\leq i<j\leq n}\hspace{-1mm}d_id_j\big)=0,$ which further implies that $\textbf{d}\cdot \textbf{d}\equiv 0\pmod 4,$ where $\textbf{d}\cdot \textbf{d}$ is viewed as an element of $GR(2^e,m).$  This implies that the code $D^{\prime}$ is doubly even.
\vspace{-1mm}\end{proof}

Further, let $\mathtt{Sym}_j(\mathcal{T}_m)$ and $\mathtt{Alt}_j(\mathcal{T}_m)$ denote the sets consisting of all $j \times j$ symmetric and alternating matrices over $\mathcal{T}_m,$ respectively.  We further recall, from Section \ref{prelim}, that  $2 \equiv u^{\kappa}(\eta_0+u\eta_1+u^2\eta_2+\cdots+u^{e-1-\kappa}\eta_{e-1-\kappa})\pmod{u^{e}},$ where $\eta_0\in \mathcal{T}_m\setminus \{0\}$ and $\eta_1,\eta_2,\ldots,\eta_{e-1-\kappa}\in \mathcal{T}_{m}.$ 
Now, the following proposition considers the case where $e \geq 4$ is an even integer  and $\kappa $ is an odd integer satisfying $3\leq \kappa \leq e-1.$  It provides a method for lifting  a $\Lambda_{s-\kappa_1}$-doubly even   code  of length $n$ and dimension $\Lambda_s$ over $\mathcal{T}_{m}$ to a self-orthogonal code  of type $\{\Lambda_{s},\lambda_{s+1}\}$ and the same length $n$ over $\mathscr{R}_{2,m}$  satisfying property $(\mathfrak{P}),$  and it also enumerates  the number of distinct ways to perform this lifting. \vspace{-1mm}\begin{proposition}\label{p3.2}
 Let $e\geq 4$ be an even integer, and let $\kappa $ be an odd integer satisfying $3\leq \kappa \leq e-1.$  Let $\mathcal{C}^{(s)}$ be a $\Lambda_{s-\kappa_1}$-doubly even  code of length $n$ and dimension $\Lambda_s$ over $\mathcal{T}_{m}.$ 
The following hold.
\begin{enumerate}\vspace{-2mm}\item[(a)]  There exists a self-orthogonal code $\mathscr{D}_{2}$ of type $\{\Lambda_{s},\lambda_{s+1}\}$ and length $n$ over $\mathscr{R}_{2,m}$  satisfying property $(\mathfrak{P})$ with $Tor_1(\mathscr{D}_2)=\mathcal{C}^{(s)}.$
\vspace{-2mm}\item[(b)] Moreover,  each  $\Lambda_{s-\kappa_1}$-doubly even  code $\mathcal{C}^{(s)}$ of length $n$ and dimension $\Lambda_s$ over $\mathcal{T}_{m}$   gives rise to precisely  	\vspace{-2mm}\begin{equation*}
\displaystyle (2^m)^{\sum\limits_{i=3}^{s+1}\lambda_i\Lambda_{i-2}+\Lambda_{s}(n-\Lambda_{s+1})- \Lambda_{s-1}-\frac{\Lambda_{s}(\Lambda_{s}-1)}{2}}{\lambda_{s+1}+n-\Lambda_{s+1}-\Lambda_s \brack \lambda_{s+1}}_{2^m}
\vspace{-1mm}\end{equation*} distinct   self-orthogonal codes $\mathscr{D}_{2}$ of type $\{\Lambda_{s},\lambda_{s+1}\}$ and length $n$ over $\mathscr{R}_{2,m}$ satisfying property $(\mathfrak{P})$ with $Tor_1(\mathscr{D}_2)=\mathcal{C}^{(s)}.$   
 \end{enumerate} \end{proposition}
\vspace{-2mm}\begin{proof}
    Working as in Proposition 4.1 of Yadav and Sharma \cite{Galois}
and by applying Lemma 4.1 of Yadav and Sharma \cite{quasi}, the desired result follows. \end{proof}
 The following proposition considers the case when both $e$ and $\kappa$ are odd integers  satisfying  $3\leq \kappa \leq e-1.$ It provides a method for lifting a $\Lambda_{s-\kappa_1}$-doubly even  code of length $n$ and dimension $\Lambda_{s+1}$ over $\mathcal{T}_m$ to a self-orthogonal code of type $\{\Lambda_{s},\lambda_{s+1},\lambda_{s+2}\}$ and the same length $n$ over $\mathscr{R}_{3,m}$ satisfying property $(\mathfrak{P}),$   and it also enumerates the number of distinct ways in which this lifting can be performed. \vspace{-1mm}\begin{proposition}\label{p3.3a} Suppose that  $e$ and $\kappa$ are both odd integers satisfying  $3\leq \kappa \leq e-1.$ Let $\mathcal{C}^{(s+1)}\supseteq \mathcal{C}^{(s)}\supseteq\mathcal{C}^{(s-\kappa_1)}\supseteq\mathcal{C}^{(s-\kappa_1-1)}$ be a chain of self-orthogonal codes of length $n$ over $\mathcal{T}_m$, where $\dim  \mathcal{C}^{(s+1)}=\Lambda_{s+1},$ $\dim \mathcal{C}^{(s)}=\Lambda_s,$ $\dim \mathcal{C}^{(s-\kappa_1)}=\Lambda_{s-\kappa_1},$  $\dim \mathcal{C}^{(s-\kappa_1-1)}=\Lambda_{s-\kappa_1-1},$  and the code $\mathcal{C}^{(s-\kappa_1)}$ is doubly even.   The following hold.
 \begin{enumerate}\vspace{-2mm}\item[(a)]  There exists a  self-orthogonal code $\mathscr{D}_{3}$ of type $\{\Lambda_{s},\lambda_{s+1},\lambda_{s+2}\}$ and length $n$ over $\mathscr{R}_{3,m}$  satisfying property $(\mathfrak{P})$ with   $Tor_1(\mathscr{D}_3)=\mathcal{C}^{(s)}$   and  $Tor_2(\mathscr{D}_3)=\mathcal{C}^{(s+1)}.$ 
\item[(b)] Moreover,  for each such chain $\mathcal{C}^{(s+1)}\supseteq \mathcal{C}^{(s)}\supseteq\mathcal{C}^{(s-\kappa_1)}\supseteq\mathcal{C}^{(s-\kappa_1-1)}$   of self-orthogonal codes over $\mathcal{T}_m,$ there are precisely	\begin{equation*}
\displaystyle (2^m)^{\sum\limits_{i=3}^{s+2}\lambda_i\Lambda_{i-2}+\sum\limits_{j=4}^{s+2}\lambda_j\Lambda_{j-3}+(\Lambda_s+\Lambda_{s+1})(n-\Lambda_{s+2}-\Lambda_s)+\Lambda_s^2-\Lambda_{s-1}-\Lambda_{s-\kappa_1-1}+\omega_3}{\lambda_{s+2}+n-\Lambda_{s+2}-\Lambda_s \brack \lambda_{s+2}}_{2^m}\vspace{-2mm}
\end{equation*} distinct   self-orthogonal codes $\mathscr{D}_{3}$ of type $\{\Lambda_{s},\lambda_{s+1},\lambda_{s+2}\}$ and length $n$ over $\mathscr{R}_{3,m}$ satisfying property $(\mathfrak{P})$ with $Tor_1(\mathscr{D}_3)=\mathcal{C}^{(s)}$   and $Tor_2(\mathscr{D}_3)=\mathcal{C}^{(s+1)},$    where $\omega_3=1$ if $\textbf{1}\in \mathcal{C}^{(s-\kappa_1-1)},$ 
while $\omega_3=0$ otherwise.
 \end{enumerate} \end{proposition}
 \begin{proof} Here, by Remark \ref{r3.1},  we assume,  without any loss of generality, that the code $\mathcal{C}^{(s+1)}$  has a
generator matrix   \begin{equation*} 
\mathtt{G}_0=[\mathtt{W}^{(0)}]_{s+1}=\begin{bmatrix}
~\mathtt{W}_1^{(0)}~\\ ~\mathtt{W}_2^{(0)}~\\ ~\vdots~\\ ~\mathtt{W}_{s+1}^{(0)}~
\end{bmatrix}  =\begin{bmatrix}
\mathtt{I}_{\lambda_1}&\mathtt{B}_{1,1}^{(0)}& \mathtt{B}_{1,2}^{(0)}&\cdots&\mathtt{B}_{1,s}^{(0)}&\cdots &\mathtt{B}_{1,e-1}^{(0)}&\mathtt{B}_{1,e}^{(0)}\vspace{0.5mm}\\
 \mathbf{0} & \mathtt{I}_{\lambda_2} &  \mathtt{B}_{2,2}^{(0)} &\cdots &\mathtt{B}_{2,s}^{(0)}&\cdots & \mathtt{B}_{2,e-1}^{(0)}& \mathtt{B}_{2,e}^{(0)}\\
	\vdots& \vdots  &\vdots &\vdots&\vdots &\vdots& \vdots&\vdots\\
 \mathbf{0}& \mathbf{0}& \mathbf{0}&\cdots&\mathtt{I}_{\lambda_{s+1}}& \cdots & \mathtt{B}_{s+1,e-1}^{(0)}& \mathtt{B}_{s+1,e}^{(0)}
\end{bmatrix}, \end{equation*} where columns of the matrix $\mathtt{G}_0$  are partitioned into blocks of sizes $\lambda_1,\lambda_2,\ldots, \lambda_{e+1}, $ the matrix $\mathtt{I}_{\lambda_i}$ is the $\lambda_i\times \lambda_i$ identity matrix over $\mathcal{T}_{m},$ the matrix    $\mathtt{B}_{i,j}^{(0)} \in \mathcal{M}_{\lambda_i\times \lambda_{j+1}}(\mathcal{T}_{m})$  for $1 \leq i \leq s+1$ and $i \leq j \leq e,$ and  each of  the  matrices $(\mathtt{B}^{(0)})_{ s,s+2} $, $(\mathtt{B}^{(0)})_{ s-1,s+3},\ldots , (\mathtt{B}^{(0)})_{2,e-1}$, $\mathtt{B}_{1,e}^{(0)}$ are of full row-ranks.
We further assume, without any loss of generality,  that the code $\mathcal{C}^{(s)}$ has a generator matrix $[\mathtt{W}^{(0)}]_{s},$ the code $\mathcal{C}^{(s-\kappa_1)}$  has a generator matrix $[\mathtt{W}^{(0)}]_{s-\kappa_1},$ and
 the code $\mathcal{C}^{(s-\kappa_1-1)}$  has a generator matrix $[\mathtt{W}^{(0)}]_{s-\kappa_1-1}.$  
 Under these assumptions, we have  \begin{eqnarray*}
 [\mathtt{W}^{(0)}]_{ s}[\mathtt{W}^{(0)}]_{ s}^t & \equiv &u\mathtt{F}^{(1)}+u^2\mathtt{F}^{(2)}+\cdots+u^{\kappa+2}\mathtt{F}^{(\kappa+2)} \pmod{u^{\kappa+3}}, \\
~[\mathtt{W}^{(0)}]_{ s}\mathtt{W}_{s+1}^{(0) t}&\equiv &uP~\pmod{u^2}, 
\end{eqnarray*} 
  where     $\mathtt{F}^{(j)}\in \mathtt{Sym}_{\Lambda_s}(\mathcal{T}_{m})$ for $1 \leq j \leq \kappa+2$ and  $P\in \mathcal{M}_{\Lambda_s\times \lambda_{s+1}}(\mathcal{T}_{m}).$ Further, one can easily see that  $\mathtt{F}^{(1)}_{i,i}=\mathtt{F}^{(3)}_{i,i}=\cdots=\mathtt{F}^{(\kappa-2)}_{i,i}=0$ for $1 \leq i \leq \Lambda_s$ and $\mathtt{F}^{(\kappa)}_{r,r}=0$ for $1 \leq r \leq \Lambda_{s-\kappa_1}.$  

 Now, to prove  the result,  let us define a matrix  $\mathtt{G}_3$ over $\mathscr{R}_{3,m}$ as 
\vspace{-1.5mm}\begin{equation*} \vspace{-1mm}\mathtt{G}_3=\begin{bmatrix}
\mathtt{W}_1^{(3)}\\\mathtt{W}_2^{(3)}\\ \vdots\\  \mathtt{W}_{s}^{(3)}\\u\mathtt{W}_{s+1}^{(3)}\\ u^2\mathtt{W}_{s+2}^{(3)}
\end{bmatrix}=\begin{bmatrix}
\mathtt{W}_1^{(0)}+u\mathtt{V}_1^{(1)}+u^2\mathtt{V}_1^{(2)}\\ \mathtt{W}_2^{(0)}+u\mathtt{V}_2^{(1)}+u^2\mathtt{V}_2^{(2)}\\ \vdots\\  \mathtt{W}_{s}^{(0)}+u\mathtt{V}_s^{(1)}+u^2\mathtt{V}_s^{(2)}\\u\mathtt{W}_{s+1}^{(3)}\vspace{0.5mm}\\u^2\mathtt{W}_{s+2}^{(3)}
\end{bmatrix}, \end{equation*} where  the matrices
 $[\mathtt{V}^{(\mu)}]_{s} \in \mathcal{M}_{\Lambda_{s}\times n}(\mathcal{T}_{m}) $  for $\mu \in \{1,2\},$  $\mathtt{W}_{s+1}^{(3)} \in \mathcal{M}_{\lambda_{s+1}\times n}(\mathscr{R}_{3,m})$  and $\mathtt{W}_{s+2}^{(3)} \in \mathcal{M}_{\lambda_{s+2}\times n}(\mathcal{T}_{m})$ are of the forms
\begin{equation*}[\mathtt{V}^{(\mu)}]_{ s}=\begin{bmatrix}
\mathtt{V}_1^{(\mu)}\\\mathtt{V}_2^{(\mu)}\\ \vdots\\ \mathtt{V}_{s}^{(\mu)}
\end{bmatrix}= \begin{bmatrix}
 \mathbf{0}&\cdots&  \mathbf{0}&\mathtt{B}_{1,\mu+1}^{(\mu)}& \mathtt{B}_{1,\mu+2}^{(\mu)}&\cdots&\mathtt{B}_{1,\mu+s}^{(\mu)}&\cdots &  \mathtt{B}_{1,e}^{(\mu)}\\
 \mathbf{0} & \cdots&  \mathbf{0}& \mathbf{0}&\mathtt{B}_{2,\mu+2}^{(\mu)}&  \cdots &\mathtt{B}_{2,\mu+s}^{(\mu)}&\cdots &\mathtt{B}_{2,e}^{(\mu)}\\
\vdots &	\vdots& \vdots  &\vdots &\vdots &\vdots&\vdots&\vdots&\vdots\\
 \mathbf{0}&\cdots&   \mathbf{0}& \mathbf{0}& \mathbf{0}&\cdots&\mathtt{B}_{s,\mu+s}^{(\mu)}&\cdots &\mathtt{B}_{s,e}^{(\mu)}
\end{bmatrix},\end{equation*} 
\begin{equation*}\mathtt{W}_{s+1}^{(3)} = \mathtt{W}_{s+1}^{(0)}+u\begin{bmatrix}
 \mathbf{0}&\cdots& \mathbf{0}& \mathtt{B}_{s+1,s+2}^{(1)}&\mathtt{B}_{s+1,s+3}^{(1)}& \cdots & \mathtt{B}_{s+1,e}^{(1)}
\end{bmatrix}   \text{ and }\end{equation*} 
\begin{equation*} \mathtt{W}_{s+2}^{(3)}= \begin{bmatrix}
 \mathbf{0}&\cdots& \mathbf{0}& \mathtt{I}_{\lambda_{s+2}}& \mathtt{B}_{s+2,s+2}^{(0)}& \mathtt{B}_{s+2,s+3}^{(0)}\cdots & \mathtt{B}_{s+2,e}^{(0)}
\end{bmatrix}  \end{equation*}
with $\mathtt{B}_{i,j}^{(\mu)} \in \mathcal{M}_{\lambda_i \times \lambda_{j+1}}(\mathcal{T}_{m})$ for $1 \leq i \leq s$ and $i +\mu \leq j \leq e,$ $\mathtt{B}_{s+1,g}^{(1)} \in \mathcal{M}_{\lambda_{s+1} \times  \lambda_{g+1}}(\mathcal{T}_{m}) $ for $ s+2 \leq g \leq e,$ and 
$\mathtt{B}_{s+2,h}^{(0)} \in \mathcal{M}_{\lambda_{s+2} \times  \lambda_{h+1}}(\mathcal{T}_{m}) $ for $s+2 \leq h \leq e.$ Let $\mathscr{D}_{3}$  be a linear code of length $n$ over $\mathscr{R}_{3,m}$ with a generator matrix  $\mathtt{G}_{3}.$  Clearly,  the code $\mathscr{D}_{3}$ is  of  type $\{\Lambda_s,\lambda_{s+1},\lambda_{s+2}\}$ and   satisfies  $Tor_1(\mathscr{D}_{3})=\mathcal{C}^{(s)}$ and  $Tor_2(\mathscr{D}_{3})=\mathcal{C}^{(s+1)}.$ 
We further see,  by  Lemma \ref{l2.2},   that  the code $\mathscr{D}_{3}$ is a  self-orthogonal code over $\mathscr{R}_{3,m}$ satisfying property $(\mathfrak{P})$   
  if and only if there exist matrices    $[\mathtt{V}^{(1)}]_{ s},[\mathtt{V}^{(2)}]_{ s} ,$ $[ \mathbf{0}~\cdots~ \mathbf{0}~ \mathtt{B}_{s+1,s+2}^{(1)}~ \cdots ~ \mathtt{B}_{s+1,e}^{(1)}]  $ and $\mathtt{W}_{s+2}^{(3)}$  satisfying the following system of matrix equations:
\vspace{-1mm}\begin{eqnarray}
\mathtt{F}^{(1)}+[\mathtt{W}^{(0)}]_{s} [\mathtt{V}^{(1)}]_{ s}^t+ [\mathtt{V}^{(1)}]_{ s} [\mathtt{W}^{(0)}]_{s}^t+u\left( \mathtt{F}^{(2)}+[\mathtt{V}^{(1)}]_{ s}[\mathtt{V}^{(1)}]_{ s}^t \right. & &\nonumber\\ \left.+[\mathtt{W}^{(0)}]_{s}[\mathtt{V}^{(2)}]_{ s}^t+[\mathtt{V}^{(2)}]_{ s}[\mathtt{W}^{(0)}]_{s}^t\right)& \equiv & \mathbf{0} \pmod{u^2},\label{e3.14}\\
\displaystyle \mathtt{Diag}\left([\mathtt{W}^{(3)}]_{s-1}[\mathtt{W}^{(3)}]_{s-1}^t\right) &\equiv &  \mathbf{0} \pmod{u^{5}},~\label{e3.15}\\ 
   \hspace{-2mm} \displaystyle  \pi_{j+3}\left( \mathtt{Diag}\big([\mathtt{W}^{(3)}]_{s-\mathrm{f}_{j+2}}[\mathtt{W}^{(3)}]_{s-\mathrm{f}_{j+2}}^t\big) \right)&=& \mathbf{0} \text{ ~~for~ } j \in \{ 2,4,\ldots, \kappa -1  \},\label{e3.16}\\
 ~P+ [\mathtt{W}^{(0)}]_{s} \left[ 0~\cdots~0~ \mathtt{B}_{s+1,s+2}^{(1)}~ \cdots ~ \mathtt{B}_{s+1,e}^{(1)}\right]^t&\equiv &  \mathbf{0} \pmod u,\label{e3.17}\\
  ~[\mathtt{W}^{(0)}]_{s} \mathtt{W}^{(3) t}_{s+2}&\equiv & \mathbf{0} \pmod u.\label{e3.18}
\end{eqnarray}

First of all, we will establish the existence of the matrices  $[\mathtt{V}^{(1)}]_s$ and $[\mathtt{V}^{(2)}]_s$ satisfying the system of matrix equations \eqref{e3.14}--\eqref{e3.16} and count their choices. 
 To do this, let us suppose that  $[\mathtt{W}^{(0)}]_s=(\textbf{b}_i),$ $[\mathtt{V}^{(1)}]_s=(\textbf{x}_j)$ and $[\mathtt{V}^{(2)}]_s=(\textbf{y}_g),$ where $\textbf{b}_i$\rq{}s, $\textbf{x}_j$\rq{}s, and $\textbf{y}_g$\rq{}s denote the rows of the matrices $[\mathtt{W}^{(0)}]_s,$ $[\mathtt{V}^{(1)}]_s,$ and $[\mathtt{V}^{(2)}]_s,$ respectively. 
  We further observe, for $1\leq i \leq \Lambda_{s},$ that the $(i,i)$-th entry, say $J_i$, of the  matrix  $\mathtt{Diag}\left([\mathtt{W}^{(3)}]_{s}[\mathtt{W}^{(3)}]_{s}^t\right)$  satisfies
 \begin{eqnarray*}\label{e3.19} J_i\equiv u^2\left(\mathtt{F}^{(2)}_{i,i}+\textbf{x}_{i}\cdot \textbf{x}_{i}\right)+u^3\mathtt{F}_{i,i}^{(3)}+u^4\left(\mathtt{F}^{(4)}_{i,i}+\textbf{y}_{i}\cdot \textbf{y}_{i}\right)+u^5\mathtt{F}_{i,i}^{(5)}+\cdots+u^{\kappa}\mathtt{F}_{i,i}^{(\kappa)} +u^{\kappa+1}\left(\mathtt{F}_{i,i}^{(\kappa+1)}+\eta_0\textbf{b}_i\cdot \textbf{x}_{i}\right)\nonumber\\ + u^{\kappa+2}\left(\mathtt{F}_{i,i}^{(\kappa+2)}+\eta_0\textbf{b}_i\cdot \textbf{y}_{i}+\eta_1\textbf{b}_i\cdot \textbf{x}_{i}\right)\pmod{u^{\kappa+3}}.
\end{eqnarray*}
 If  $w$ is the smallest positive integer satisfying  $\kappa+1 \leq 2^{w},$ then we observe, for $1\leq i \leq \Lambda_s,$ that 
\begin{eqnarray*}\label{e3.20}
 J_i\equiv  u^2 \left((\mathtt{F}^{(2)}_{i,i})^{2^{-w}}+(\textbf{1}\cdot \textbf{x}_{i})^{2^{-(w-1)}}\right)^{2^{w}}+u^3\mathtt{F}_{i,i}^{(3)}+u^4\left((\mathtt{F}^{(4)}_{i,i})^{2^{-w}}+(\textbf{1}\cdot \textbf{y}_{i})^{2^{-(w-1)}}\right)^{2^w}+u^5\mathtt{F}_{i,i}^{(5)}+\cdots~~~~~~\nonumber\\+u^{\kappa}\mathtt{F}_{i,i}^{(\kappa)} +u^{\kappa+1}\left(\mathtt{F}_{i,i}^{(\kappa+1)}+\eta_0\textbf{b}_i\cdot \textbf{x}_{i}\right)+u^{\kappa+2}\Big(\mathtt{F}_{i,i}^{(\kappa+2)}+\eta_0\textbf{b}_i\cdot \textbf{y}_{i}+\eta_1\textbf{b}_i\cdot \textbf{x}_{i}~~\nonumber\\ -\eta_0\Big((\textbf{1}\cdot \textbf{x}_i)(\mathtt{F}_{i,i}^{(2)})^{2^{-1}}+\sum\limits_{1 \leq h <g \leq n}\texttt{x}_{i,h}\texttt{x}_{i,g} \Big)\Big)\pmod{u^{\kappa+3}},\end{eqnarray*}
where $\textbf{x}_i=(\texttt{x}_{i,1},\texttt{x}_{i,2},\ldots,\texttt{x}_{i,n}) \in \mathcal{T}_{m}^n.$ 
It is easy to see that  $\big((\mathtt{F}^{(2)}_{i,i})^{2^{-w}}+(\textbf{1}\cdot \textbf{x}_{i})^{2^{-(w-1)}}\big)^{2^{w}} \in \mathcal{T}_{m}$  and $ \big((\mathtt{F}^{(4)}_{i,i})^{2^{-w}}+(\textbf{1}\cdot \textbf{y}_{i})^{2^{-(w-1)}}\big)^{2^{w}}\in \mathcal{T}_{m}$ for $1 \leq i \leq \Lambda_{s}.$ 
   In view of this, we see that the system  of  matrix equations \eqref{e3.14}-\eqref{e3.16} is  equivalent to the following system of equations in unknowns $ \textbf{x}_1, \textbf{x}_2,\ldots, \textbf{x}_{\Lambda_{s}}$ and $ \textbf{y}_1, \textbf{y}_2,\ldots, \textbf{y}_{\Lambda_{s}}$over $\mathcal{T}_{m}$:  \begin{eqnarray}
\mathtt{F}_{i,j}^{(1)}+\textbf{b}_i\cdot \textbf{x}_j+\textbf{b}_j \cdot \textbf{x}_i  +u\left(\mathtt{F}_{i,j}^{(2)}+\textbf{x}_i\cdot\textbf{x}_j+\textbf{b}_i\cdot \textbf{y}_j+\textbf{b}_j \cdot \textbf{y}_i \right) &\equiv& 0 \pmod{u^2} \text{ for } 1 \leq i <j \leq \Lambda_s,\label{e3.21}\\ 
(\mathtt{F}^{(2)}_{i,i})^{2^{-1}}+\textbf{1}\cdot \textbf{x}_i & \equiv & 0 \pmod{u} \text{ ~~for } 1 \leq i \leq \Lambda_{s},~~\label{e3.22}\\
(\mathtt{F}^{(4)}_{i,i})^{2^{-1}} +\textbf{1}\cdot \textbf{y}_i & \equiv & 0\pmod{u} \text{~~ for } 1 \leq i \leq \Lambda_{s-1},\label{e3.23}\\
\hspace{-8mm}\mathtt{F}_{i,i}^{(\kappa+2)}+\eta_0\textbf{b}_i\cdot \textbf{y}_{i}+\eta_1\textbf{b}_i\cdot \textbf{x}_{i}-\eta_0\Big((\textbf{1}\cdot \textbf{x}_i)(\mathtt{F}_{i,i}^{(2)})^{2^{-1}}\hspace{-1mm}+\hspace{-3mm}\sum\limits_{1 \leq h <g \leq n}\hspace{-2mm}\texttt{x}_{i,h}\texttt{x}_{i,g} \Big) &\equiv&0 \pmod{u} \text{ for } 1 \leq i \leq \Lambda_{s-\kappa_1-1}. \label{e3.24}
\end{eqnarray}
Working as in Lemma 4.1  of Yadav and Sharma   \cite{quasi},  we see that there exist vectors $\textbf{x}_1, \textbf{x}_2,\ldots, \textbf{x}_{\Lambda_{s}}$
satisfying the following system of equations: 
 \begin{eqnarray}
\mathtt{F}_{i,j}^{(1)}+\textbf{b}_i\cdot \textbf{x}_j+\textbf{b}_j \cdot \textbf{x}_i  &\equiv& 0 \pmod{u} \text{ ~~for } 1 \leq i <j \leq \Lambda_s,~~~~~~~\label{e3.25}\\ 
(\mathtt{F}^{(2)}_{i,i})^{2^{-1}}+\textbf{1}\cdot \textbf{x}_i & \equiv & 0 \pmod{u} \text{~~ for } 1 \leq i \leq \Lambda_{s},\label{e3.26}
 \end{eqnarray}
and hence there exists a matrix  $\textbf{[}\mathtt{V}^{(1)}\textbf{]}_{s} \in \mathcal{M}_{\Lambda_s\times n}(\mathcal{T}_{m}) $ satisfying \eqref{e3.25} and \eqref{e3.26}. Moreover, such a matrix $\textbf{[}\mathtt{V}^{(1)}\textbf{]}_{s} \in \mathcal{M}_{\Lambda_s\times n}(\mathcal{T}_{m}) $ has precisely  $(2^m)^{\sum\limits_{i=3}^{s+2}\lambda_i\Lambda_{i-2}+\Lambda_{s}(n-\Lambda_{s+2})-\Lambda_s-\frac{\Lambda_{s}(\Lambda_{s}-1)}{2}}  $ distinct choices. 
Next, for a given choice of the matrix $[\mathtt{V}^{(1)}]_{s}$  satisfying \eqref{e3.25} and \eqref{e3.26}, we can write
\begin{equation}\label{e3.28}
 \mathtt{F}^{(1)}+[\mathtt{W}^{(0)}]_{s}[\mathtt{V}^{(1)}]_{s}^t+[\mathtt{V}^{(1)}]_{s}[\mathtt{W}^{(0)}]_{s}^t ~\equiv ~ uQ \pmod{u^2}
\end{equation}
for some  $Q\in \mathtt{Sym}_{\Lambda_{s}}(\mathcal{T}_{m}) .$ 
This, by \eqref{e3.21}, implies that
\begin{eqnarray}
    Q+\mathtt{F}^{(2)}+[\mathtt{W}^{(0)}]_{s}[\mathtt{V}^{(2)}]_{s}^t+[\mathtt{V}^{(2)}]_{s}[\mathtt{W}^{(0)}]_{s}^t+[\mathtt{V}^{(1)}]_{s}[\mathtt{V}^{(1)}]_{s}^t&\equiv&  \mathbf{0}\pmod{u}.\label{e3.29}
\end{eqnarray}
Note that $\mathcal{C}^{(s-\kappa_1-1)}$ is a doubly even code over $\mathcal{T}_m.$ Further, if $\textbf{1}\in \mathcal{C}^{(s-\kappa_1-1)},$ then by Theorem 3.2 of Yadav and Sharma \cite{Galois}, we see that   $n\equiv 0,4\pmod 8.$  Now, working  as in Lemma 4.1 of Yadav and Sharma \cite{Galois}, we see that the matrix
$[\mathtt{V}^{(2)}]_{s}$ satisfying \eqref{e3.23}, \eqref{e3.24} and \eqref{e3.29} has precisely 
\begin{equation*}
(2^m)^{\sum\limits_{i=4}^{s+2}\lambda_i\Lambda_{i-3}+\Lambda_{s}(n-\Lambda_{s+2})-\Lambda_{s-1}-\Lambda_{s-\kappa_1-1}-\frac{\Lambda_{s}(\Lambda_{s}-1)}{2}+\omega_3}\end{equation*} distinct choices, where $\omega_3=1$ if $\textbf{1}\in \mathcal{C}^{(s-\kappa_1-1)},$ 
 while $\omega_3=0$ otherwise. Further, working as in Proposition 4.2 of Yadav and Sharma \cite{Galois}, we see that there exist matrices 
$[  \mathbf{0}~\cdots~ \mathbf{0}~ \mathtt{B}_{s+1,s+2}^{(1)}~ \cdots ~ \mathtt{B}_{s+1,e}^{(1)}]  $ and $\mathtt{W}_{s+2}^{(3)}$  
satisfying \eqref{e3.17} and \eqref{e3.18} and that such matrices $[  \mathbf{0}~\cdots~ \mathbf{0}~ \mathtt{B}_{s+1,s+2}^{(1)}~ \cdots ~ \mathtt{B}_{s+1,e}^{(1)}]  $ and $\mathtt{W}_{s+2}^{(3)}$  have  precisely \begin{equation*}
  (2^m)^{\lambda_{s+1}(n-\Lambda_{s+2}-\Lambda_{s})}{\lambda_{s+2}+n-\Lambda_{s+2}-\Lambda_{s} \brack \lambda_{s+2}}_{2^m}
\end{equation*}  distinct choices. 
From this,  the desired result follows immediately. 
 \end{proof}
From this point on, we shall  distinguish the following two cases: (i) $2 \kappa \leq e,$ and (ii) $2 \kappa > e.$
The following proposition lifts a  self-orthogonal code  of type $\{\Lambda_{\gamma_{\ell}+2},\lambda_{\gamma_{\ell}+3}, \ldots, \lambda_{\gamma_{\ell}+\ell-1}\}$   and  length $n$ over $\mathscr{R}_{\ell-2,m}$ satisfying  property $(\mathfrak{P})$ to a self-orthogonal code  of type $\{\Lambda_{\gamma_{\ell}+1},\lambda_{\gamma_{\ell}+2}, \ldots, \lambda_{\gamma_{\ell}+\ell}\}$ and the same length $n$ over $\mathscr{R}_{\ell,m}$ satisfying  property $(\mathfrak{P}),$ where the integer $\ell $ satisfies  $\ell \equiv e \pmod 2$ and $4 \leq \ell \leq \kappa$ if $2\kappa \leq e,$ while   $4 \leq \ell \leq e-\kappa+1-2\theta_e $ if $2\kappa >e.$ Additionally, it enumerates the number of possible ways to perform this lifting.
 \begin{proposition}\label{p3.4Kodd} Let $\ell$ be a fixed  integer satisfying  $\ell \equiv e \pmod 2$ and $4 \leq \ell \leq \kappa$ if  $2\kappa \leq e,$ while  $4 \leq \ell \leq e-\kappa+1-2\theta_e $ if $2\kappa >e.$ 
   Let us define $\gamma_{\ell}=s-\mathrm{f}_{\ell}.$  Let $\mathscr{D}_{\ell-2}$ be a self-orthogonal code of type $\{\Lambda_{\gamma_{\ell}+2},\lambda_{\gamma_{\ell}+3}, \ldots, \lambda_{\gamma_{\ell}+\ell-1}\}$   and  length $n$ over $\mathscr{R}_{\ell-2,m}$ satisfying  property $(\mathfrak{P}).$  
Let   $\mathcal{C}^{(\gamma_{\ell}+1)}$ and $\mathcal{C}^{(\gamma_{\ell}+1-\kappa_1-\theta_e)}$ be linear subcodes of   $Tor_{1}(\mathscr{D}_{\ell-2})$ such that $\mathcal{C}^{(\gamma_{\ell}+1-\kappa_1-\theta_e)}\subseteq \mathcal{C}^{(\gamma_{\ell}+1)},$  $\dim \mathcal{C}^{(\gamma_{\ell}+1)}= \Lambda_{\gamma_{\ell}+1},$  $\dim \mathcal{C}^{(\gamma_{\ell}+1-\kappa_1-\theta_e)}=\Lambda_{\gamma_{\ell}+1-\kappa_1-\theta_e}$ and the code $\mathcal{C}^{(\gamma_{\ell}+1-\kappa_1-\theta_e)}$ is doubly even, (note that such a pair of codes $\mathcal{C}^{(\gamma_{\ell}+1)}$ and $\mathcal{C}^{(\gamma_{\ell}+1-\kappa_1-\theta_e)}$ exists for each choice of $\mathscr{D}_{\ell-2},$ by Lemma \ref{l3.1}).
The following hold.
 \begin{enumerate}\item[(a)] There exists a  self-orthogonal code $\mathscr{D}_{\ell}$ of type $\{\Lambda_{\gamma_{\ell}+1},\lambda_{\gamma_{\ell}+2}, \ldots, \lambda_{\gamma_{\ell}+\ell}\}$ and length $n$ over $\mathscr{R}_{\ell,m}$ satisfying  property $(\mathfrak{P})$ with     $Tor_{1}(\mathscr{D}_{\ell})=\mathcal{C}^{(\gamma_{\ell}+1)}$ and $Tor_{j+1}(\mathscr{D}_{\ell})=Tor_{j}(\mathscr{D}_{\ell-2})$ for $1 \leq j \leq  \ell-2.$
\item[(b)] Moreover,  each such choice of the codes $\mathscr{D}_{\ell-2},$ $\mathcal{C}^{(\gamma_{\ell}+1)}$ and $\mathcal{C}^{(\gamma_{\ell}+1-\kappa_1-\theta_e)})$ 
gives rise to  precisely
\begin{eqnarray*}
\displaystyle (2^m)^{\sum\limits_{i=\ell}^{\gamma_{\ell}+\ell}\lambda_i\Lambda_{i-\ell+1}+\sum\limits_{j=\ell+1}^{\gamma_{\ell}+\ell}\lambda_j\Lambda_{j-\ell}+Y_{\ell}}  {\lambda_{\gamma_{\ell}+\ell}+n-\Lambda_{\gamma_{\ell}+\ell}-\Lambda_{\gamma_{\ell}+1} \brack \lambda_{\gamma_{\ell}+\ell}}_{2^m}
\end{eqnarray*}  distinct  self-orthogonal codes $\mathscr{D}_{\ell}$ of type $\{\Lambda_{\gamma_{\ell}+1},\lambda_{\gamma_{\ell}+2}, \ldots, \lambda_{\gamma_{\ell}+\ell}\}$ and length $n$ over $\mathscr{R}_{\ell,m}$ satisfying property $(\mathfrak{P})$ with $Tor_{1}(\mathscr{D}_{\ell})=\mathcal{C}^{(\gamma_{\ell}+1)}$ and $Tor_{j+1}(\mathscr{D}_{\ell})=Tor_{j}(\mathscr{D}_{\ell-2})$ for $1 \leq j \leq  \ell-2,$ where  $Y_{\ell}=(\Lambda_{\gamma_{\ell}+\ell-1}+\Lambda_{\gamma_{\ell}+1})(n-\Lambda_{\gamma_{\ell}+\ell}-\Lambda_{\gamma_{\ell}+1}) +\Lambda_{\gamma_{\ell}+1}^2 +\Lambda_{\gamma_{\ell}+1}-\Lambda_{s-2\mathrm{f}_{\ell}+2}-\Lambda_{s-2\mathrm{f}_{\ell}+1}-\Lambda_{\gamma_{\ell}+1-\kappa_1-\theta_e} +\omega_{\ell}$ with 
$\omega_{\ell}=1$ if $\textbf{1} \in \mathcal{C}^{(\gamma_{\ell}+1-\kappa_1-\theta_e)},$  while $\omega_{\ell}=0$ otherwise.
 \end{enumerate} \end{proposition}
\begin{proof} To prove the result, we assume, without any loss of generality, that  the code  
$\mathscr{D}_{\ell-2}$ has a generator matrix 
  \vspace{-1mm}\begin{equation*}\label{e00.1Kodd}
\mathtt{G}_{\ell-2}=\begin{bmatrix}
\mathtt{W}_1^{(\ell-2)}\\\mathtt{W}_2^{(\ell-2)} \\ \vdots\\ \mathtt{W}_{\gamma_{\ell}+2}^{(\ell-2)}\vspace{0.5mm}\\ u\mathtt{W}_{\gamma_{\ell}+3}^{(\ell-2)}\vspace{0.5mm}\\  \vdots\\u^{\ell-3}\mathtt{W}_{\gamma_{\ell}+\ell-1}^{(\ell-2)}
\end{bmatrix}=\begin{bmatrix}
\mathtt{W}_1^{(0)}+u\mathtt{V}_1^{(1)}+u^2\mathtt{V}_1^{(2)}+\cdots +u^{\ell-3}\mathtt{V}_1^{(\ell-3)}\\\mathtt{W}_2^{(0)}+u\mathtt{V}_2^{(1)}+u^2\mathtt{V}_2^{(2)}+\cdots +u^{\ell-3}\mathtt{V}_2^{(\ell-3)} \\ \vdots\\ \mathtt{W}_{\gamma_{\ell}+2}^{(0)}+u\mathtt{V}_{\gamma_{\ell}+2}^{(1)}+u^2\mathtt{V}_{\gamma_{\ell}+2}^{(2)}+\cdots +u^{\ell-3}\mathtt{V}_{\gamma_{\ell}+2}^{(\ell-3)}\vspace{0.5mm}\\ u\mathtt{W}_{\gamma_{\ell}+3}^{(\ell-2)}\vspace{0.5mm}\\ \vdots\\u^{\ell-3}\mathtt{W}_{\gamma_{\ell}+\ell -1}^{(\ell-2)}
\end{bmatrix},\end{equation*} 
 where  
\begin{equation*}\label{e00.2Kodd}
[\mathtt{W}^{(0)}]_{\gamma_{\ell}+2}=\begin{bmatrix}
\mathtt{W}_1^{(0)}\\\mathtt{W}_2^{(0)}\\ \vdots\\ \mathtt{W}_{\gamma_{\ell}+2}^{(0)}
\end{bmatrix}  =\begin{bmatrix}
{\mathtt{I}}_{\lambda_1}&{\mathtt{B}}_{1,1}^{(0)}&\cdots&{\mathtt{B}}_{1,\gamma_{\ell}+1}^{(0)}&\cdots & {\mathtt{B}}_{1,e-1}^{(0)}& {\mathtt{B}}_{1,e}^{(0)}\\
\mathbf{0} &{\mathtt{I}}_{\lambda_2}  &\cdots &{\mathtt{B}}_{2,\gamma_{\ell}+1}^{(0)}&\cdots & {\mathtt{B}}_{2,e-1}^{(0)}&{\mathtt{B}}_{2,e}^{(0)}\\
	\vdots& \vdots   &\vdots&\vdots &\vdots& \vdots&\vdots\\
\mathbf{0}&\mathbf{0}&\cdots& {\mathtt{I}}_{\lambda_{\gamma_{\ell}+2}}& \cdots & {\mathtt{B}}_{\gamma_{\ell}+2,e-1}^{(0)}&{\mathtt{B}}_{\gamma_{\ell}+2,e}^{(0)}
\end{bmatrix}\end{equation*} 
with  ${\mathtt{I}}_{\lambda_i}$ as the $\lambda_i\times \lambda_i$ identity matrix over ${\mathcal{T}}_{m},$    $ {\mathtt{B}}_{i,j}^{(0)} \in \mathcal{M}_{\lambda_i\times \lambda_{j+1}}({\mathcal{T}}_{m})$  for $1 \leq i \leq \gamma_{\ell}+2$ and $i \leq j \leq e,$
$[\mathtt{V}^{(a)}]_{\gamma_{\ell}+2}\in \mathcal{M}_{\Lambda_{\gamma_{\ell}+2}\times n}(\mathcal{T}_{m})$ for $1 \leq a \leq \ell-3,$ and  the matrix $\mathtt{W}_{\gamma_{\ell}+g}^{(\ell-2)}\in \mathcal{M}_{\lambda_{\gamma_{\ell}+g}\times n}(\mathscr{R}_{\ell-2,m})$  is to be considered modulo $u^{\ell-g}$   for  $3 \leq g \leq \ell-1.$ Here,  we note that the torsion code $Tor_1(\mathscr{D}_{\ell-2})$ is a $\Lambda_{\gamma_{\ell}+2}$-dimensional code over $\mathcal{T}_{m}$ with a generator matrix $[\mathtt{W}^{(0)}]_{\gamma_{\ell}+2}.$   
 We further assume, without any loss of generality,  that the  code  $\mathcal{C}^{(\gamma_{\ell}+1)}$ has a generator matrix $[\mathtt{W}^{(0)}]_{\gamma_{\ell}+1}$
and   the code $\mathcal{C}^{(\gamma_{\ell}+1-\kappa_1-\theta_e)}$   has a generator matrix $[\mathtt{W}^{(0)}]_{\gamma_{\ell}+1-\kappa_1-\theta_e} .$ 
Additionally, by Remark \ref{r3.1}, we  assume that the matrix $(\mathtt{B}^{(0)})_ { \gamma_{\ell}+1,\gamma_{\ell}+\ell} $ is of full row-rank. 
   
 Now, as $\mathscr{D}_{\ell-2}$ is a self-orthogonal code over $\mathscr{R}_{\ell-2,m}$ satisfying   property $(\mathfrak{P}),$  we have
\begin{eqnarray*}
[\mathtt{W}^{(\ell-2)}]_{\gamma_{\ell}+1}[\mathtt{W}^{(\ell-2)}]_{\gamma_{\ell}+1}^t &\equiv &\mathbf{0} \pmod{u^{\ell-2}},\label{}\\
  \displaystyle \mathtt{Diag}\left([\mathtt{W}^{(\ell-2)}]_{\gamma_{\ell}+2-\mathrm{f}_i}[\mathtt{W}^{(\ell-2})]_{\gamma_{\ell}+2-\mathrm{f}_i}^t\right) &\equiv & \mathbf{0} \pmod{u^{\ell-2+i}} \text{ ~~for ~} i\in \{2,4,6,\ldots,\ell-2-\theta_e\},~\\ 
  \displaystyle  \pi_{\ell-2+j}\left( \mathtt{Diag}\big([\mathtt{W}^{(\ell-2)}]_{\gamma_{\ell}+2-\mathrm{f}_{j+2}}[\mathtt{W}^{(\ell-2)}]_{\gamma_{\ell}+2-\mathrm{f}_{j+2}}^t\big) \right)&=&\mathbf{0} \text{ ~~for~~ } j \in \{ \ell-3,\ell-1,\ldots, \kappa -2+\theta_e \},\\
  ~[\mathtt{W}^{(\ell-2)}]_{ \gamma_{\ell}+1} \mathtt{W}_{\gamma_{\ell}+g}^{(\ell-2)t}&\equiv &\mathbf{0}  \pmod{u^{\ell-g}} \text{ ~for } 2 \leq g \leq \ell-1 ,\\
\mathtt{W}_{\gamma_{\ell}+i}^{(\ell-2)}
 \mathtt{W}_{\gamma_{\ell}+j}^{(\ell-2)t}&\equiv &\mathbf{0}  \pmod{u^{\ell+2-i-j}}~\text{ for } 2 \leq i,j\leq \ell-1 \text{ and } i+j\leq \ell+1.\end{eqnarray*} 
   This implies that \begin{eqnarray*}
    [\mathtt{W}^{(\ell-2)}]_{\gamma_{\ell}+1} [\mathtt{W}^{(\ell-2)}]_{\gamma_{\ell}+1}^t&\equiv&  u^{\ell-2}\mathtt{F}^{(\ell-2)}+u^{\ell-1}\mathtt{F}^{(\ell-1)}+\cdots+u^{\ell+\kappa-1}\mathtt{F}^{(\ell+\kappa-1)} \pmod{u^{\ell+\kappa}},\\ ~[\mathtt{W}^{(\ell-2)}]_{ \gamma_{\ell}+1} \mathtt{W}_{\gamma_{\ell}+g}^{(\ell-2)t}&\equiv &u^{\ell-g}P_{g}  \pmod{u^{\ell+1-g}} \text{ ~for } 2 \leq g \leq \ell-1 ,\\
\mathtt{W}_{\gamma_{\ell}+i}^{(\ell-2)} \mathtt{W}_{\gamma_{\ell}+j}^{(\ell-2)t}&\equiv &\mathbf{0} \pmod{u^{\ell+2-i-j}} ~\text{ for } 2 \leq i,j\leq \ell-1 \text{ and } i+j\leq \ell+1,
\end{eqnarray*} where $\mathtt{F}^{(\ell-2)},\mathtt{F}^{(\ell-1)}, \ldots, \mathtt{F}^{(\ell+\kappa-1)} \in \mathtt{Sym}_{\Lambda_{\gamma_{\ell}+1}}(\mathcal{T}_{m})$ with $\mathtt{F}_{h,h}^{(\ell-4+i)}=\mathtt{F}^{(\ell-3+i)}_{h,h}=0 $ for $1 \leq h \leq \Lambda_{\gamma_{\ell}+2-\mathrm{f}_i }$ and $ i \in \{ 2,4,\ldots,\ell-2-\theta_e\} $ and  $\mathtt{F}_{a,a}^{(\ell-2+j)}=0 $ for $ 1 \leq a \leq \Lambda_{\gamma_{\ell}+2-\mathrm{f}_{j+2}} $ and $ j\in \{ \ell-3,\ell-1, \ell+1,\ldots, \kappa-2+\theta_e\},$ and the matrix $P_{g} \in \mathcal{M}_{\Lambda_{\gamma_{\ell}+1} \times \lambda_{\gamma_{\ell}+g}}(\mathcal{T}_{m})  $ for $ 2 \leq g \leq \ell-1.$   
   
Now, to establish the result,  we define a matrix  $\mathtt{G}_{\ell}$ over $\mathscr{R}_{\ell,m}$ as follows:
\vspace{-1mm}\begin{equation*} \mathtt{G}_{\ell}=\begin{bmatrix}
\mathtt{W}_1^{(\ell)}\\\mathtt{W}_2^{(\ell)}\\ \vdots\\ \mathtt{W}_{\gamma_{\ell}+1}^{(\ell)}\\ \vspace{0.5mm}u\mathtt{W}_{\gamma_{\ell}+2}^{(\ell)}\\ \vdots\\u^{\ell-1}\mathtt{W}_{\gamma_{\ell}+\ell}^{(\ell)}
\end{bmatrix}=\begin{bmatrix} \mathtt{W}_1^{(\ell-2)}+u^{\ell-2}\mathtt{V}_1^{(\ell-2)}+u^{\ell-1}\mathtt{V}_1^{(\ell-1)}\\ \mathtt{W}_2^{(\ell-2)}+u^{\ell-2}\mathtt{V}_2^{(\ell-2)}+u^{\ell-1}\mathtt{V}_2^{(\ell-1)}\\ \vdots\\ \mathtt{W}_{\gamma_{\ell}+1}^{(\ell-2)}+u^{\ell-2}\mathtt{V}_{\gamma_{\ell}+1}^{(\ell-2)}+u^{\ell-1}\mathtt{V}_{\gamma_{\ell}+1}^{(\ell-1)}\\ \vspace{0.5mm}u\mathtt{W}_{\gamma_{\ell}+2}^{(\ell)}\\ \vdots\\u^{\ell-1}\mathtt{W}_{\gamma_{\ell}+\ell}^{(\ell)}   \end{bmatrix},\end{equation*} 
where the matrices  $[\mathtt{V}^{(\mu)}]_{\gamma_{\ell}+1}\in \mathcal{M}_{\Lambda_{\gamma_{\ell}+1}\times n}(\mathcal{T}_{m})$ for $ \mu \in \{\ell-2,\ell-1\},$   $\mathtt{W}_{\gamma_{\ell}+g}^{(\ell)}\in \mathcal{M}_{\lambda_{\gamma_{\ell}+g}\times n}(\mathscr{R}_{\ell,m})$   for  $2 \leq g \leq \ell-1$ and  $\mathtt{W}_{\gamma_{\ell}+\ell}^{(\ell)}\in \mathcal{M}_{\lambda_{\gamma_{\ell}+\ell}\times n}(\mathcal{T}_{m})$ are of the forms 
\begin{equation*}\begin{bmatrix}
\mathtt{V}_1^{(\mu)}\\\mathtt{V}_2^{(\mu)}\\ \vdots\\ \mathtt{V}_{\gamma_{\ell}+1}^{(\mu)}
\end{bmatrix}= \begin{bmatrix}
\mathbf{0}&\cdots&\mathbf{0}&\mathtt{B}_{1,\mu+1}^{(\mu)}& \mathtt{B}_{1,\mu+2}^{(\mu)}&\cdots&\mathtt{B}_{1,\gamma_{\ell}+1+\mu}^{(\mu)}&\cdots &  \mathtt{B}_{1,e}^{(\mu)}\\
\mathbf{0} & \cdots&\mathbf{0}&\mathbf{0}& \mathtt{B}_{2,\mu+2}^{(\mu)}&  \cdots &\mathtt{B}_{2,\gamma_{\ell}+1+\mu}^{(\mu)}&\cdots &\mathtt{B}_{2,e}^{(\mu)}\\
	\vdots& \cdots &\vdots &\vdots &\vdots&\vdots & \vdots&\vdots&\vdots\\
\mathbf{0}&\cdots&\mathbf{0}&\mathbf{0}&\mathbf{0}&\cdots&  \mathtt{B}_{\gamma_{\ell}+1,\gamma_{\ell}+1+\mu}^{(\mu)}&\cdots &  \mathtt{B}_{\gamma_{\ell}+1,e}^{(\mu)}\\
\end{bmatrix},\end{equation*} 
 \begin{equation*} \mathtt{W}_{\gamma_{\ell}+g}^{(\ell)}= \mathtt{W}_{\gamma_{\ell}+g}^{(\ell-2)}+u^{\ell-g}\begin{bmatrix}
\mathbf{0}&\cdots&\mathbf{0}& \mathtt{B}_{\gamma_{\ell}+g,\gamma_{\ell}+\ell}^{(\ell-g)}&\cdots & \mathtt{B}_{\gamma_{\ell}+g,e}^{(\ell-g)}
\end{bmatrix} \text{ and } \end{equation*} 
\begin{equation*} \mathtt{W}_{\gamma_{\ell}+\ell}^{(\ell)}= \begin{bmatrix}
\mathbf{0}&\cdots&\mathbf{0}&\mathtt{I}_{\lambda_{\gamma_{\ell}+\ell}}& \mathtt{B}_{\gamma_{\ell}+\ell,\gamma_{\ell}+\ell}^{(0)}& \cdots & \mathtt{B}_{\gamma_{\ell}+\ell,e}^{(0)}
\end{bmatrix} \end{equation*}
with   $\mathtt{B}_{i,j}^{(\mu)} \in \mathcal{M}_{\lambda_i \times \lambda_{j+1}}(\mathcal{T}_{m})$ for $1 \leq i \leq \gamma_{\ell}+1$ and $ i+\mu \leq j \leq e,$   $\mathtt{B}_{\gamma_{\ell}+g,h}^{(\ell-g)} \in \mathcal{M}_{\lambda_{\gamma_{\ell}+g} \times  \lambda_{h+1}}(\mathcal{T}_{m}) \text{ for } \gamma_{\ell}+\ell \leq h \leq e$  and $\mathtt{B}_{\gamma_{\ell}+\ell,a}^{(0)} \in \mathcal{M}_{\lambda_{\gamma_{\ell}+\ell} \times  \lambda_{a+1}}(\mathcal{T}_{m}) \text{ for } \gamma_{\ell}+\ell \leq a \leq e. $
Next, let  $\mathscr{D}_{\ell}$  be a  linear code of length $n$ over $\mathscr{R}_{\ell,m}$ with a generator matrix $\mathtt{G}_{\ell}.$   Note that the code $\mathscr{D}_{\ell}$ is  of  type $\{\Lambda_{\gamma_{\ell}+1},\lambda_{\gamma_{\ell}+2}, \ldots, \lambda_{\gamma_{\ell}+\ell}\}$ over $\mathscr{R}_{\ell,m}$ and satisfies $Tor_1(\mathscr{D}_{\ell})=\mathcal{C}^{(\gamma_{\ell}+1)}$ and $Tor_{j+1}(\mathscr{D}_{\ell})=Tor_{j}(\mathscr{D}_{\ell-2})$ for $1 \leq j \leq \ell-2.$  By  Lemma \ref{l2.2},  we see that the code  $\mathscr{D}_{\ell}$ is a self-orthogonal code over $\mathscr{R}_{\ell,m}$ satisfying property $(\mathfrak{P})$   if and only if there exist matrices  $[\mathtt{V}^{(\ell-2)}]_{\gamma_{\ell}+1},$ $[\mathtt{V}^{(\ell-1)}]_{\gamma_{\ell}+1},$ $\left[\mathbf{0}~\cdots~\mathbf{0}~\mathtt{B}_{\gamma_{\ell}+g,\gamma_{\ell}+\ell}^{(\ell-g)}~\cdots ~ \mathtt{B}_{\gamma_{\ell}+g,e}^{(\ell-g)}\right]$ for $2 \leq g \leq \ell-1$ and  $\mathtt{W}_{\gamma_{\ell}+\ell}^{(\ell)}$
 satisfying the following system of matrix equations:
\small{\begin{eqnarray}
 \mathtt{F}^{(\ell-2)}+[\mathtt{W}^{(0)}]_{\gamma_{\ell}+1}[\mathtt{V}^{(\ell-2)}]_{\gamma_{\ell}+1}^t+[\mathtt{V}^{(\ell-2)}]_{\gamma_{\ell}+1}[\mathtt{W}^{(0)}]_{\gamma_{\ell}+1}^t+
u\left([\mathtt{W}^{(0)}]_{\gamma_{\ell}+1}[\mathtt{V}^{(\ell-1)}]_{\gamma_{\ell}+1}^t\right.~~~~&&\nonumber\\\left.+[\mathtt{V}^{(\ell-1)}]_{\gamma_{\ell}+1}[\mathtt{W}^{(0)}]_{\gamma_{\ell}+1}^t +[\mathtt{V}^{(1)}]_{\gamma_{\ell}+1}[\mathtt{V}^{(\ell-2)}]_{\gamma_{\ell}+1}^t+[\mathtt{V}^{(\ell-2)}]_{\gamma_{\ell}+1}[\mathtt{V}^{(1)}]_{\gamma_{\ell}+1}^t+\mathtt{F}^{(\ell-1)}\right)&\equiv& \mathbf{0}\pmod{u^2},\label{e3.30Kodd}\\
\mathtt{Diag}\left(\mathtt{E}^{(2\ell-4)}+[\mathtt{V}^{(\ell-2)}]_{s-2\mathrm{f}_{\ell}+2}[\mathtt{V}^{(\ell-2)}]_{s-2\mathrm{f}_{\ell}+2}^t\right) &\equiv & \mathbf{0} \pmod{u},\label{e3.31Kodd}~~~~~~~\\ 
\mathtt{Diag}\left(\mathtt{E}^{(2\ell-2)}+[\mathtt{V}^{(\ell-1)}]_{s-2\mathrm{f}_{\ell}+1}[\mathtt{V}^{(\ell-1)}]_{s-2\mathrm{f}_{\ell}+1}^t\right) &\equiv & \mathbf{0} \pmod{u},\label{e3.32Kodd}\\
\mathtt{Diag}\left(\mathtt{E}^{(\ell+\kappa-1)}+\eta_1[\mathtt{W}^{(0)}]_{\gamma_{\ell}-\kappa_1}[\mathtt{V}^{(\ell-2)}]_{\gamma_{\ell}-\kappa_1}^t\right.~~~~~~~&&\nonumber\\\left.+\eta_0\left([\mathtt{W}^{(0)}]_{\gamma_{\ell}-\kappa_1}[\mathtt{V}^{(\ell-1)}]_{\gamma_{\ell}-\kappa_1}^t+[\mathtt{V}^{(1)}]_{\gamma_{\ell}-\kappa_1}[\mathtt{V}^{(\ell-2)}]_{\gamma_{\ell}-\kappa_1}^t\right)\right) &\equiv & \mathbf{0} \pmod{u} \text{ if } e \text{ is odd},\label{e3.33bKodd}~~~~~~~\\
\mathtt{Diag}\left(\mathtt{E}^{(\ell+\kappa-2)}+\eta_0[\mathtt{W}^{(0)}]_{\gamma_{\ell}+1-\kappa_1}[\mathtt{V}^{(\ell-2)}]_{\gamma_{\ell}+1-\kappa_1}^t\right) &\equiv & \mathbf{0} \pmod{u}\text{ if } e \text{ is even},\label{e3.33Kodd}~~~~~~~\\ P_{g}+[\mathtt{W}^{(0)}]_{ \gamma_{\ell}+1} \begin{bmatrix} \mathbf{0}~ \cdots ~\mathbf{0} ~\mathtt{B}_{\gamma_{\ell}+g,\gamma_{\ell}+\ell}^{(\ell-g)} ~\cdots ~ \mathtt{B}_{\gamma_{\ell}+g,e}^{(\ell-g)}\end{bmatrix} ^t&\equiv &\mathbf{0}   \pmod{u}, ~~\label{e3.34Kodd}\\
~[\mathtt{W}^{(0)}]_{\gamma_{\ell}+1} \mathtt{W}_{\gamma_{\ell}+\ell}^{(\ell)t} &\equiv& \mathbf{0} \pmod{u},\label{e3.35Kodd}\end{eqnarray}}\normalsize
 where   the matrices $\mathtt{E}^{(2\ell-4)},$ $\mathtt{E}^{(2\ell-2)}$ and $\mathtt{E}^{(\ell+\kappa-2+\theta_e)}$ are  of orders $\Lambda_{s-2\mathrm{f}_{\ell}+2} \times \Lambda_{s-2\mathrm{f}_{\ell}+2},$  $\Lambda_{s-2\mathrm{f}_{\ell}+1} \times \Lambda_{s-2\mathrm{f}_{\ell}+1}$ and $\Lambda_{\gamma_{\ell}+1-\kappa_1-\theta_e} \times \Lambda_{\gamma_{\ell}+1-\kappa_1-\theta_e}$ over $\mathcal{T}_{m}$ whose rows are the  first $\Lambda_{s-2\mathrm{f}_{\ell}+2},$ $\Lambda_{s-2\mathrm{f}_{\ell}+1}$ and $\Lambda_{\gamma_{\ell}+1-\kappa_1-\theta_e}$  rows of the matrices $\mathtt{F}^{(2\ell-4)},$ $ \mathtt{F}^{(2\ell-2)}$ and $ \mathtt{F}^{(\ell+\kappa-2+\theta_e)}, $   respectively.

 First of all, we will establish the existence of the   matrices $[\mathtt{V}^{(\ell-2)}]_{\gamma_{\ell}+1}$ and  $[\mathtt{V}^{(\ell-1)}]_{\gamma_{\ell}+1}$  satisfying the system of matrix equations \eqref{e3.30Kodd}--\eqref{e3.33Kodd}.
To do this,  we will distinguish the following two cases: (I) $\ell$ is even, and (II) $\ell$ is odd. 

\begin{itemize}

\item[(I)] Let $\ell$ be even.  In this case, we will  first show that there exists a matrix  $[\mathtt{V}^{(\ell-2)}]_{\gamma_{\ell}+1} \in \mathcal{M}_{\Lambda_{\gamma_{\ell}+1}\times n}(\mathcal{T}_{m}) $ satisfying the  following system of matrix equations:  
\begin{eqnarray}
  \mathtt{F}^{(\ell-2)}+[\mathtt{W}^{(0)}]_{\gamma_{\ell}+1}[\mathtt{V}^{(\ell-2)}]_{\gamma_{\ell}+1}^t+[\mathtt{V}^{(\ell-2)}]_{\gamma_{\ell}+1}[\mathtt{W}^{(0)}]_{\gamma_{\ell}+1}^t &\equiv & \mathbf{0} \pmod{u},\label{e3.40p3.5Kodd}\\
  \mathtt{Diag}\left(\mathtt{E}^{(2\ell-4)}+[\mathtt{V}^{(\ell-2)}]_{s-2\mathrm{f}_{\ell}+2}[\mathtt{V}^{(\ell-2)}]_{s-2\mathrm{f}_{\ell}+2}^t\right) &\equiv & \mathbf{0} \pmod{u}\label{e3.41p3.5Kodd},\\
\mathtt{Diag}\left(\mathtt{E}^{(\ell+\kappa-2)}+\eta_0[\mathtt{W}^{(0)}]_{\gamma_{\ell}+1-\kappa_1} [\mathtt{V}^{(\ell-2)}]_{\gamma_{\ell}+1-\kappa_1}^t\right) &\equiv & \mathbf{0} \pmod{u}.\label{e3.42p3.5Kodd}
\end{eqnarray} 
Towards this, let  $[\mathtt{W}^{(0)}]_{\gamma_{\ell}+1}=(\textbf{b}_i)$ and $[\mathtt{V}^{(\ell-2)}]_{\gamma_{\ell}+1}=(\textbf{v}_j),$ where $\textbf{b}_i$ and $\textbf{v}_j$ denote the $i$-th and $j$-th  rows of the matrices $[\mathtt{W}^{(0)}]_{\gamma_{\ell}+1}$ and $[\mathtt{V}^{(\ell-2)}]_{\gamma_{\ell}+1},$ respectively, for each $i$ and $j.$ 
Further, let $\mathtt{F}^{(\ell-2)}_{i,j},\mathtt{E}^{(2\ell-4)}_{i,j}, \mathtt{E}^{(\ell+\kappa-2)}_{i,j}\in \mathcal{T}_{m}$ denote the $(i,j)$-th entries of the matrices  $\mathtt{F}^{(\ell-2)},\mathtt{E}^{(2\ell-4)},$  $\mathtt{E}^{(\ell+\kappa-2)},$ respectively. Thus, the system  of  matrix equations \eqref{e3.40p3.5Kodd}--\eqref{e3.42p3.5Kodd} is  equivalent to the following system of equations in unknowns $ \textbf{v}_1, \textbf{v}_2,\ldots, \textbf{v}_{\Lambda_{\gamma_{\ell}+1}}$ over $\mathcal{T}_{m}$:  \begin{eqnarray}
\textbf{b}_i\cdot \textbf{v}_j+\textbf{b}_j \cdot \textbf{v}_i  &\equiv& \mathtt{F}^{(\ell-2)}_{i,j} \pmod{u} \text{ ~~for } 1 \leq i <j \leq \Lambda_{\gamma_{\ell}+1},\label{e3.6p3.5Kodd}\\ 
\textbf{1}\cdot \textbf{v}_i & \equiv & (\mathtt{E}^{(2\ell-4)}_{i,i})^{2^{m-1}} \pmod{u} \text{~~ for } 1 \leq i \leq \Lambda_{s-2\mathrm{f}_{\ell}+2},\label{e3.7p3.5Kodd}\\
\textbf{b}_i \cdot \textbf{v}_i & \equiv & \eta_0^{-1}\mathtt{E}_{i,i}^{(\ell+\kappa-2)} \pmod{u} \text{ ~~for } 1 \leq i \leq \Lambda_{\gamma_{\ell}+1-\kappa_1}. \label{e3.8p3.5Kodd}
\end{eqnarray} 
For each integer $j$ satisfying $1 \leq j \leq \Lambda_{\gamma_{\ell}+1},$ one can easily see that there exists a unique integer $r_{j}$ satisfying $1 \leq r_{j}\leq \gamma_{\ell}+1$ and $\Lambda_{r_{j}-1}+1 \leq j \leq \Lambda_{r_{j}}.$ The corresponding unknown vector $\textbf{v}_{j}$ can be written as $\textbf{v}_{j}=(\textbf{0}~ \textbf{v}_{j}^{n-\Lambda_{r_{j}+\ell-2}}),$ where $\textbf{0}$ denotes the zero vector of 
length $\Lambda_{r_{j}+\ell-2}$ and $\textbf{v}_{j}^{n-\Lambda_{r_{j}+\ell-2}}$ denotes the vector of length $n-\Lambda_{r_{j}+\ell-2}$ obtained   by omitting the first $\Lambda_{r_{j}+\ell-2}$ coordinates of $\textbf{v}_{j}$. This implies,  for $\Lambda_{r_{j}-1}+1 \leq j \leq \Lambda_{r_{j}}$, that the first $\Lambda_{r_{j}+\ell-2}$ coordinates of $\textbf{v}_{j}$ are zero,  leaving  $n-\Lambda_{r_{j}+\ell-2}$   variables
 in $\textbf{v}_{j}.$  For $1 \leq j \leq \Lambda_{\gamma_{\ell}+1},$ let $\widehat{\textbf{v}}_{j}=\textbf{v}_{j}^{n-\Lambda_{r_{j}+\ell-2}}$ 
and  $\widehat{\textbf{b}}_{j}=\textbf{b}_{j}^{n-\Lambda_{r_{j}+\ell-2}}$ be the  vectors of length $n-\Lambda_{r_{j}+\ell-2}$ obtained from $\textbf{v}_{j}$ and  $\textbf{b}_{j}$ by omitting their first $\Lambda_{r_{j}+\ell-2}$ coordinates, respectively. 
Consequently, the system   of equations \eqref{e3.6p3.5Kodd}--\eqref{e3.8p3.5Kodd} is equivalent to the following system of equations in unknowns $\widehat{\textbf{v}}_1, \widehat{\textbf{v}}_2,\ldots, \widehat{\textbf{v}}_{\Lambda_{\gamma_{\ell}+1}}$ over $\mathcal{T}_{m}$:   
\begin{eqnarray}
\widehat{\textbf{b}}_i\cdot \widehat{\textbf{v}}_j+\widehat{\textbf{b}}_j \cdot \widehat{\textbf{v}}_i  &\equiv& \mathtt{F}^{(\ell-2)}_{i,j} \pmod{u} \text{ ~~for } 1 \leq i <j \leq \Lambda_{\gamma_{\ell}+1},\label{e3.9p3.5Kodd}\\ 
\widehat{\textbf{1}}\cdot \widehat{\textbf{v}}_i & \equiv & (\mathtt{E}^{(2\ell-4)}_{i,i})^{2^{m-1}} \pmod{u} \text{~~ for } 1 \leq i \leq \Lambda_{s-2\mathrm{f}_{\ell}+2},\label{e3.10p3.5Kodd}\\
\widehat{\textbf{b}}_i \cdot \widehat{\textbf{v}}_i   & \equiv & \eta_0^{-1}\mathtt{E}_{i,i}^{(\ell+\kappa-2)} \pmod{u} \text{ ~~for } 1 \leq i \leq \Lambda_{\gamma_{\ell}+1-\kappa_1}, \label{e3.11p3.5Kodd}
\end{eqnarray} 
 where $\widehat{\textbf{1}}$ denotes the all-one vector having the same length as that of the vector $\widehat{\textbf{v}}_i$ for each $i.$ We further express the above system of equations \eqref{e3.9p3.5Kodd}--\eqref{e3.11p3.5Kodd} in the matrix form as
\small{\begin{equation}\label{e3.12p3.5Kodd}\mathcal{N}\begin{bmatrix}
	\widehat{\textbf{v}}_1^{t}\\ \widehat{\textbf{v}}_2^{t}\\ \widehat{\textbf{v}}_3^{t}\\ \vdots\\ \vdots\\ \widehat{\textbf{v}}_{\Lambda_{\gamma_{\ell}+1-\kappa_1}}^{t}\vspace{0.50mm}\\ \vdots\\ \widehat{\textbf{v}}_{\Lambda_{s-2\mathrm{f}_{\ell}+2}+1}^{t} \\ \vdots\\ \widehat{\textbf{v}}_{\Lambda_{\gamma_{\ell}+1}-2}^{t}\\ \widehat{\textbf{v}}_{\Lambda_{\gamma_{\ell}+1}-1}^{t} \\ \vspace{0.75mm}\widehat{\textbf{v}}_{\Lambda_{\gamma_{\ell}+1}}^{t}
	\end{bmatrix}\equiv\begin{bmatrix}
	(\mathtt{E}_{1,1}^{(2\ell-4)})^{2^{m-1}}\\ \vdots\\ (\mathtt{E}_{\Lambda_{s-2\mathrm{f}_{\ell}+2},\Lambda_{s-2\mathrm{f}_{\ell}+2}}^{(2\ell-4)})^{2^{m-1}}\vspace{0.5mm}\\ \eta_0^{-1}\mathtt{E}^{(\ell+\kappa-2)}_{1,1}\vspace{-0.25mm}\\ \vdots\vspace{-0.25mm}\\ \eta_0^{-1}\mathtt{E}^{(\ell+\kappa-2)}_{\Lambda_{\gamma_{\ell}+1-\kappa_1},\Lambda_{\gamma_{\ell}+1-\kappa_1}}\vspace{0.5mm}\\\mathtt{F}^{(\ell-2)}_{1,2}\\\vdots\\\mathtt{F}^{(\ell-2)}_{1,\Lambda_{\gamma_{\ell}+1}}\vspace{-0.25mm}\\ \vdots\vspace{-0.25mm} \\\mathtt{F}^{(\ell-2)}_{\Lambda_{\gamma_{\ell}+1}-1,\Lambda_{\gamma_{\ell}+1}}
	\end{bmatrix}\pmod u,\end{equation}}\normalsize
where the matrix  $\mathcal{N}$ of order $\Big( \Lambda_{s-2\mathrm{f}_\ell+2}+\Lambda_{\gamma_{\ell}+1-\kappa_1}+\frac{\Lambda_{\gamma_{\ell}+1}(\Lambda_{\gamma_{\ell}+1}-1)}{2}\Big) \times \Big(\sum\limits_{i=\ell}^{\gamma_{\ell}+\ell}\lambda_i\Lambda_{i-\ell+1}+\Lambda_{\gamma_{\ell}+1}(n-\Lambda_{\gamma_{\ell}+\ell})\Big)$ is  given by  
\smaller{\begin{equation*}
  \mathcal{N}=\begin{bmatrix}
	\widehat{\textbf{1}} &&&\\
    	&\widehat{\textbf{1}} &&&\\
	& & \ddots  &  \\
	& & & \widehat{\textbf{1}}\\
    	& & &&\ddots  &  \\
    & & & && \widehat{\textbf{1}}\\
	\widehat{\textbf{b}}_1 &&&\\
    	&\widehat{\textbf{b}}_2 &&&\\
	& & \ddots  &  \\
	& & &  \widehat{\textbf{b}}_{\Lambda_{\gamma_{\ell}+1}-\kappa_1}&\\
\widehat{\textbf{b}}_{2}&\widehat{\textbf{b}}_{1}&&& \\
	\vdots&\vdots  & \vdots &\vdots\\
    \widehat{\textbf{b}}_{\Lambda_{s-2\mathrm{f}_{\ell}+2}}&&&&&\widehat{\textbf{b}}_{1}&&& \\
    	\vdots&\vdots  & \vdots &\vdots\\
\widehat{\textbf{b}}_{\Lambda_{\gamma_{\ell}+1}}&&&&&&\cdots&&\widehat{\textbf{b}}_{1} \\
	&&&\vdots&\vdots  & \vdots&\cdots&\vdots&\vdots \\
    	
	& && &&&\cdots& \widehat{\textbf{b}}_{\Lambda_{\gamma_{\ell}+1}}& \widehat{\textbf{b}}_{\Lambda_{\gamma_{\ell}+1}-1}\\
	\end{bmatrix}.\vspace{-1mm}
\end{equation*}}\normalsize
Furthermore,  working as in Lemma 4.1 of Yadav and Sharma \cite{Galois}, we observe that the rows of the matrix $\mathcal{N}$  are linearly independent over $\mathcal{T}_{m}$  if and only if $\textbf{1}\notin \mathcal{C}^{(\gamma_{\ell}+1-\kappa_1)}.$ Accordingly, we will consider the following two cases separately:  (i) $\textbf{1}\notin \mathcal{C}^{(\gamma_{\ell}+1-\kappa_1)},$ and (ii)  $\textbf{1}\in \mathcal{C}^{(\gamma_{\ell}+1-\kappa_1)}.$

\begin{itemize}\vspace{-1mm}\item[(i)] Let  $\textbf{1}\notin \mathcal{C}^{(\gamma_{\ell}+1-\kappa_1)}.$ In this case, the  rows of the matrix $\mathcal{N}$ are linearly independent over $\mathcal{T}_{m},$  and hence the row-rank of the matrix $\mathcal{N}$ is  $\Lambda_{s-2\mathrm{f}_\ell+2}+\Lambda_{\gamma_{\ell}+1-\kappa_1}+\frac{\Lambda_{\gamma_{\ell}+1}(\Lambda_{\gamma_{\ell}+1}-1)}{2}.$ Thus, the matrix equation  \eqref{e3.12p3.5Kodd} always has a solution. In fact,    the number of  solutions  of  the matrix equation \eqref{e3.12p3.5Kodd}, and hence the number of choices for the matrix $[\mathtt{V}^{(\ell-2)}]_{\gamma_{\ell}+1} \in \mathcal{M}_{\Lambda_{\gamma_{\ell}+1}\times n}(\mathcal{T}_{m})$ is  given by 
\vspace{-2mm}\begin{equation*}
   \displaystyle (2^m)^{\sum\limits_{i=\ell}^{\gamma_{\ell}+\ell}\lambda_i\Lambda_{i-\ell+1}+\Lambda_{\gamma_{\ell}+1}(n-\Lambda_{\gamma_{\ell}+\ell})-\Lambda_{s-2\mathrm{f}_{\ell}+2}-\Lambda_{\gamma_{\ell}+1-\kappa_1}-\frac{\Lambda_{\gamma_{\ell}+1}(\Lambda_{\gamma_{\ell}+1}-1)}{2} }. 
\vspace{-2mm}\end{equation*} 
\vspace{-1mm}\item[(ii)] Now, let us assume that $\textbf{1}\in \mathcal{C}^{(\gamma_{\ell}+1-\kappa_1)}.$  In this case, $\textbf{1}$ belongs to $\mathcal{T}_m$-span of the vectors  $\textbf{b}_1,\textbf{b}_2,\ldots,\\  \textbf{b}_{\Lambda_{\gamma_{\ell}+1-\kappa_1}}.$ Since the  vectors $\textbf{b}_1,\textbf{b}_2,\ldots, \textbf{b}_{\Lambda_{\gamma_{\ell}+1-\kappa_1}}$  are linearly independent over $\mathcal{T}_{m},$ there exist unique scalars $\alpha_{1},\alpha_{2},\ldots,\alpha_{\Lambda_{\gamma_{\ell}+1-\kappa_1}} \in \mathcal{T}_{m}$ such that $\alpha_{1} \textbf{b}_{1}+\alpha_{2} \textbf{b}_{2}+\cdots+\alpha_{\Lambda_{\gamma_{\ell}+1-\kappa_1}} \textbf{b}_{\Lambda_{\gamma_{\ell}+1-\kappa_1}} \equiv  \textbf{1}\pmod u.$ 
Moreover, all the rows of the matrix $\mathcal{N}$ except the last row are linearly independent over $\mathcal{T}_{m}.$ This implies that the row-rank of the matrix $\mathcal{N}$ is  $\Lambda_{\gamma_{\ell}+1-\kappa_1}+\Lambda_{s-2\mathrm{f}_\ell+2}+\frac{\Lambda_{\gamma_{\ell}+1}(\Lambda_{\gamma_{\ell}+1}-1)}{2}-1.$ Further, one can easily see that  the matrix equation \eqref{e3.12p3.5Kodd} has a solution if and only if 
\vspace{-1mm}\begin{equation}\label{e3.13p3.5Kodd}\vspace{-1mm}
   \displaystyle   \sum\limits_{h=1}^{\Lambda_{\gamma_{\ell}+1-\kappa_1}}\alpha_h(\mathtt{E}^{(2\ell-4)}_{h,h})^{2^{m-1}} +\eta_0^{-1}\sum\limits_{g=1}^{\Lambda_{\gamma_{\ell}+1-\kappa_1}}\alpha_g^2 \mathtt{E}^{(\ell+\kappa-2)}_{g,g} + \hspace{-4mm}\sum\limits_{~~~~~1\leq i< j \leq \Lambda_{\gamma_{\ell}+1-\kappa_1}}\hspace{-4mm}\alpha_i \alpha_j \mathtt{F}_{i,j}^{(\ell-2)}~ \equiv ~0\pmod{u}.
 \end{equation}
 As  $\mathcal{C}^{(\gamma_{\ell}+1-\kappa_1)}$ is a doubly even code over $\mathcal{T}_m$ satisfying $\textbf{1}\in \mathcal{C}^{(\gamma_{\ell}+1-\kappa_1)},$ we see,   by Theorem 3.2 of Yadav and Sharma \cite{Galois}, that $n\equiv 0,4\pmod 8.$  We next assert that the equation  \eqref{e3.13p3.5Kodd} holds (\textit{i.e.,} the matrix equation \eqref{e3.12p3.5Kodd}  always has a solution).       To prove this assertion, let $[\mathtt{V}^{(a)}]_{\gamma_{\ell}+1} = (\textbf{v}^{(a)}_j)$ for $1 \leq a \leq \ell-3,$  where  $\textbf{v}^{(a)}_j$ denotes the 
 $j$-th row of the matrix $[\mathtt{V}^{(a)}]_{\gamma_{\ell}+1}$  for each $j$ and $a.$  Further, let us write \vspace{-1mm}\begin{equation*}\vspace{-1mm}
     \textbf{1} \equiv \sum\limits_{j=1}^{\Lambda_{\gamma_{\ell}+1-\kappa_1}} \hspace{-2mm}\alpha_j \left(\textbf{b}_j+u\textbf{v}^{(1)}_j+ \cdots+u^{\ell-3}\textbf{v}^{(\ell-3)}_j\right)+ u^{\ell-2}\textbf{d}_{\ell-2}+\cdots+u^{\ell+\kappa-2} \textbf{d}_{\ell+\kappa-2} \pmod{u^{\ell+\kappa-1}}
 \end{equation*} for some $\textbf{d}_{\ell-2},\textbf{d}_{\ell-1},\ldots,\textbf{d}_{\ell+\kappa-2} \in \mathcal{T}_{m}^n.$ Note that  $\textbf{1}\cdot \textbf{1} \equiv n \pmod{ u^{\ell+\kappa-1}}.$ As $n\equiv 0,4\pmod 8,$ we must have  $\textbf{1}\cdot \textbf{1} \equiv 0 \pmod{ u^{\ell+\kappa-1}}. $   From this, we get 
\vspace{-1mm}\begin{equation}\label{eta_0p3.5Kodd}\vspace{-1mm}
 \sum\limits_{h=1}^{\Lambda_{\gamma_{\ell}+1-\kappa_1}}\alpha_h^2 \mathtt{E}_{h,h}^{(2\ell-4)}+\textbf{d}_{\ell-2}\cdot \textbf{d}_{\ell-2}  \equiv 0\pmod{u}   \end{equation}       and \vspace{-1mm}\begin{equation}\label{eta_1p3.5Kodd}\vspace{-1mm}
   \eta_0^{-1} \sum\limits_{g=1}^{\Lambda_{\gamma_{\ell}+1-\kappa_1}}\alpha_g^2 \mathtt{E}^{(\ell+\kappa-2)}_{g,g}+ \hspace{-5mm}\sum\limits_{~~~~~1\leq i< j \leq \Lambda_{\gamma_{\ell}+1-\kappa_1}}\hspace{-5mm}\alpha_i \alpha_j \mathtt{F}_{i,j}^{(\ell-2)}+\textbf{1}\cdot \textbf{d}_{\ell-2} \equiv 0\pmod{u}.  
 \end{equation} 
By \eqref{eta_0p3.5Kodd} and \eqref{eta_1p3.5Kodd},  we see that the equation \eqref{e3.13p3.5Kodd} holds, and hence the matrix equation \eqref{e3.12p3.5Kodd}  always has a solution. 
  Further, 
  the number of solutions of the matrix equation \eqref{e3.12p3.5Kodd}, and hence the number of choices for the matrix  $[\mathtt{V}^{(\ell-2)}]_{\gamma_{\ell}+1} \in \mathcal{M}_{\Lambda_{\gamma_{\ell}+1} \times n } (\mathcal{T}_{m})$ is given by 
 \vspace{-1mm} \begin{equation*}\vspace{-1mm}
   \displaystyle (2^m)^{\sum\limits_{i=\ell}^{\gamma_{\ell}+\ell}\lambda_i\Lambda_{i-\ell+1}+\Lambda_{\gamma_{\ell}+1}(n-\Lambda_{\gamma_{\ell}+\ell})-\Lambda_{s-2\mathrm{f}_\ell+2}-\Lambda_{\gamma_{\ell}+1-\kappa_1}-\frac{\Lambda_{\gamma_{\ell}+1}(\Lambda_{\gamma_{\ell}+1}-1)}{2}+1}. 
\end{equation*}
Next, for a given choice of the matrix $[V^{(\ell-2)}]_{\gamma_{\ell}+1}$  satisfying \eqref{e3.40p3.5Kodd}--\eqref{e3.42p3.5Kodd},  we can write
\begin{equation}\label{e3.43Kodd}
 \mathtt{F}^{(\ell-2)}+[\mathtt{W}^{(0)}]_{\gamma_{\ell}+1}[\mathtt{V}^{(\ell-2)}]_{\gamma_{\ell}+1}^t+[\mathtt{V}^{(\ell-2)}]_{\gamma_{\ell}+1}[\mathtt{W}^{(0)}]_{\gamma_{\ell}+1}^t ~\equiv ~ uQ_1 \pmod{u^2}
\end{equation}
for some  $Q_1\in \mathtt{Sym}_{\Lambda_{\gamma_{\ell}+1}}(\mathcal{T}_{m}) .$ 
This, by \eqref{e3.30Kodd},  implies that
\begin{eqnarray}\label{e3.44Kodd}
    Q_1+\mathtt{F}^{(\ell-1)}+[\mathtt{W}^{(0)}]_{\gamma_{\ell}+1}[\mathtt{V}^{(\ell-1)}]_{\gamma_{\ell}+1}^t+[\mathtt{V}^{(\ell-1)}]_{\gamma_{\ell}+1}[\mathtt{W}^{(0)}]_{\gamma_{\ell}+1}^t&&\nonumber\\+[\mathtt{V}^{(1)}]_{\gamma_{\ell}+1}[\mathtt{V}^{(\ell-2)}]_{\gamma_{\ell}+1}^t+[\mathtt{V}^{(\ell-2)}]_{\gamma_{\ell}+1}[\mathtt{V}^{(1)}]_{\gamma_{\ell}+1}^t&\equiv& \mathbf{0}\pmod{u}.
    \end{eqnarray}
Now, working as in Lemma 4.1  of Yadav and Sharma   \cite{quasi}, we observe that there exists a matrix $[\mathtt{V}^{(\ell-1)}]_{\gamma_{\ell}+1}$ satisfying  \eqref{e3.32Kodd} and \eqref{e3.44Kodd}. Moreover, the matrix $[\mathtt{V}^{(\ell-1)}]_{\gamma_{\ell}+1}$ satisfying  \eqref{e3.32Kodd} and \eqref{e3.44Kodd}  has precisely \vspace{-1mm}\begin{equation*}
  \vspace{-1mm}  \displaystyle (2^m)^{\sum\limits_{i=\ell+1}^{\gamma_{\ell}+\ell}\lambda_i\Lambda_{i-\ell}+\Lambda_{\gamma_{\ell}+1}(n-\Lambda_{\gamma_{\ell}+\ell})-\Lambda_{s-2\mathrm{f}_\ell+1}-\frac{\Lambda_{\gamma_{\ell}+1}(\Lambda_{\gamma_{\ell}+1}-1)}{2}}
\end{equation*} distinct choices. Further, working as in Proposition 4.3 of Yadav and Sharma \cite{Galois}, we see that the matrices 
$\Big[
  ~0~ \cdots ~0 ~ \mathtt{B}_{\gamma_{\ell}+g,\gamma_{\ell}+\ell}^{(\ell-g)} ~\cdots ~ \mathtt{B}_{\gamma_{\ell}+g,e}^{(\ell-g)}  
\Big] $ for $2 \leq g \leq \ell-1$ and $\mathtt{W}_{\gamma_{\ell}+\ell}^{(\ell)}$
satisfying \eqref{e3.34Kodd} and \eqref{e3.35Kodd} have  precisely \vspace{-1mm}\begin{equation*}
\vspace{-1mm}(2^m)^{(\Lambda_{\gamma_{\ell}+\ell-1}-\Lambda_{\gamma_{\ell}+1})(n-\Lambda_{\gamma_{\ell}+\ell}-\Lambda_{\gamma_{\ell}+1})}{\lambda_{\gamma_{\ell}+\ell}+n-\Lambda_{\gamma_{\ell}+\ell}-\Lambda_{\gamma_{\ell}+1} \brack \lambda_{\gamma_{\ell}+\ell}}_{2^m}\end{equation*}  distinct choices.  From this, the desired result follows in the case when $\ell$ is even.
 \vspace{-2mm} \end{itemize} 
  \item[(II)] When $\ell$ is odd, the desired result follows by working as in case (I). \vspace{-2mm}
\end{itemize}\vspace{-2mm}\vspace{-2mm} \end{proof}

 The following proposition  lifts a  self-orthogonal code of type $\{\Lambda_{s-\kappa_1+1},\lambda_{s-\kappa_1+2}, \ldots, \lambda_{s+\kappa_1+\theta_e}\}$   and  length $n$ over $\mathscr{R}_{\kappa-1+\theta_e,m}$ satisfying  property $(\mathfrak{P})$ to a   self-orthogonal code  of type $\{\Lambda_{s-\kappa_1},\lambda_{s-\kappa_1+1}, \ldots, \lambda_{s+\kappa_1+1+\theta_e}\}$ and the same length $n$ over $\mathscr{R}_{\kappa+1+\theta_e,m}$ that also satisfies property $(\mathfrak{P})$ when  $2 \kappa \leq e.$ It also enumerates the number of all possible distinct ways to perform this lifting.
\vspace{-1mm}\begin{proposition}\label{p3.6Kodd} Let $2\kappa \leq e,$   and let $\ell=\kappa+1+\theta_e.$  Let $\mathscr{D}_{\kappa-1+\theta_e}$ be a self-orthogonal code of type $\{\Lambda_{s-\kappa_1+1},\lambda_{s-\kappa_1+2}, \\ \ldots, \lambda_{s+\kappa_1+\theta_e}\}$   and  length $n$ over $\mathscr{R}_{\kappa-1+\theta_e,m}$ satisfying  property $(\mathfrak{P}).$ Let $\mathcal{C}^{(s-\kappa+\theta_e)} \subseteq \mathcal{C}^{(s-\kappa+1)}\subseteq \mathcal{C}^{(s-\kappa_1)}$ be a chain of linear subcodes of $Tor_{1}(\mathscr{D}_{\kappa-1+\theta_e}), $ such that $\dim \mathcal{C}^{(s-\kappa+\theta_e)}= \Lambda_{s-\kappa+\theta_e},$ $\dim \mathcal{C}^{(s-\kappa+1)}= \Lambda_{s-\kappa+1}$, $\dim \mathcal{C}^{(s-\kappa_1)}=\Lambda_{s-\kappa_1} $ and the code $\mathcal{C}^{(s-\kappa_1)}$ is doubly even, (such a chain $\mathcal{C}^{(s-\kappa+\theta_e)} \subseteq \mathcal{C}^{(s-\kappa+1)}\subseteq \mathcal{C}^{(s-\kappa_1)}$ of codes exists for each choice of $\mathscr{D}_{\kappa-1+\theta_e},$ by Lemma \ref{l3.1}).    
 The following hold.
 \begin{enumerate}\vspace{-1mm}\item[(a)] There exists a  self-orthogonal code $\mathscr{D}_{\kappa+1+\theta_e}$ of type $\{\Lambda_{s-\kappa_1},\lambda_{s-\kappa_1+1}, \ldots, \lambda_{s+\kappa_1+1+\theta_e}\}$ and length $n$ over $\mathscr{R}_{\kappa+1+\theta_e,m}$ satisfying property $(\mathfrak{P})$ with  $Tor_{1}(\mathscr{D}_{\kappa+1+\theta_e})=\mathcal{C}^{(s-\kappa_1)}$ and     $Tor_{j+1}(\mathscr{D}_{\kappa+1+\theta_e})=Tor_{j}(\mathscr{D}_{\kappa-1+\theta_e})$ for $1 \leq j \leq  \kappa-1+\theta_e,$  where such a code exists: \begin{itemize}\item unconditionally, if $\textbf{1}\notin \mathcal{C}^{(s-\kappa+\theta_e)};$ \item  if and only if either $n\equiv 0\pmod 8$ or 
  $n \equiv 4 \pmod{8}$ with $m$  even when $\textbf{1}\in \mathcal{C}^{(s-\kappa+\theta_e)}.$ \end{itemize}
\vspace{-1mm}\item[(b)] Moreover,  each such choice of the codes   $\mathscr{D}_{\kappa-1+\theta_e},$ $\mathcal{C}^{(s-\kappa_1)},$  $\mathcal{C}^{(s-\kappa+1)}$  and $\mathcal{C}^{(s-\kappa+\theta_e)}$ gives rise to precisely 
\vspace{-1mm}\begin{equation*}
\vspace{-1mm}\hspace{-4mm}2^{\epsilon}(2^m)^{\sum\limits_{i=\kappa+1+\theta_e}^{s+\kappa_1+1+\theta_e}\lambda_i\Lambda_{i-\kappa-\theta_e}+\sum\limits_{j=\kappa+2+\theta_e}^{s+\kappa_1+1+\theta_e}\lambda_j\Lambda_{j-\kappa-1-\theta_e}+\Lambda_{s-\kappa_1}^2+\Lambda_{s-\kappa_1}+Y_{\kappa+1+\theta_e}}  {\lambda_{s+\kappa_1+1+\theta_e}+n-\Lambda_{s+\kappa_1+1+\theta_e}-\Lambda_{s-\kappa_1} \brack \lambda_{s+\kappa_1+1+\theta_e}}_{2^m}
\end{equation*}  distinct  self-orthogonal codes $\mathscr{D}_{\kappa+1+\theta_e}$ of type $\{\Lambda_{s-\kappa_1},\lambda_{s-\kappa_1+1}, \ldots, \lambda_{s+\kappa_1+1+\theta_e}\}$ and length $n$ over $\mathscr{R}_{\kappa+1+\theta_e,m}$ satisfying property $(\mathfrak{P})$ with $Tor_{1}(\mathscr{D}_{\kappa+1+\theta_e})=\mathcal{C}^{(s-\kappa_1)}$ and   $Tor_{j+1}(\mathscr{D}_{\kappa+1+\theta_e})=Tor_{j}(\mathscr{D}_{\kappa-1+\theta_e})$ for $1 \leq j \leq  \kappa-1+\theta_e.$ Here  $\epsilon=1$ if $\textbf{1} \in \mathcal{C}^{(s-\kappa+\theta_e)},$ 
while $\epsilon=0$ otherwise. The number $Y_{\kappa+1+\theta_e}$ is given by $$Y_{\kappa+1+\theta_e}=\left\{ \begin{array}{ll} (\Lambda_{s+\kappa_1+1}+\Lambda_{s-\kappa_1})(n-\Lambda_{s+\kappa_1+2}-\Lambda_{s-\kappa_1})-\Lambda_{s-\kappa+1}-\Lambda_{s-\kappa} & \text{if }e \text{ is odd}; \\
(\Lambda_{s+\kappa_1}+\Lambda_{s-\kappa_1})(n-\Lambda_{s+\kappa_1+1}-\Lambda_{s-\kappa_1})-\Lambda_{s-\kappa} -2\Lambda_{s-\kappa+1}  +\omega_{\kappa+1} & \text{if } e \text{ is even},\end{array}\right.$$ where, in the case that $e$ is even, the number $\omega_{\kappa+1}$ is given by
 $\omega_{\kappa+1}=1$ if $\textbf{1} \in \mathcal{C}^{(s-\kappa+1)},$  
 while $\omega_{\kappa+1}=0$ otherwise.\end{enumerate} \vspace{-1mm}\end{proposition}
\vspace{-1mm}\begin{proof}  To prove the result, we  will consider the following two cases separately: (I) $e$ is even , and (II) $e$ is odd. 
 \vspace{-1mm}\begin{itemize}
     \item[(I)] Let $e$ be even. In this case, we assume, without any loss of generality, that the code 
$\mathscr{D}_{\kappa-1}$ has a generator matrix
\begin{equation*}\label{e00.1}
\mathtt{G}_{\kappa-1}=\begin{bmatrix}
\mathtt{W}_1^{(\kappa-1)}\\\mathtt{W}_2^{(\kappa-1)} \\ \vdots\\ \mathtt{W}_{s-\kappa_1+1}^{(\kappa-1)}\vspace{0.5mm}\\ u\mathtt{W}_{s-\kappa_1+2}^{(\kappa-1)}\vspace{0.5mm}\\  \vdots\\u^{\kappa-2}\mathtt{W}_{s+\kappa_1}^{(\kappa-1)}
\end{bmatrix}=\begin{bmatrix}
\mathtt{W}_1^{(0)}+u\mathtt{V}_1^{(1)}+u^2\mathtt{V}_1^{(2)}+\cdots +u^{\kappa-2}\mathtt{V}_1^{(\kappa-2)}\\\mathtt{W}_2^{(0)}+u\mathtt{V}_2^{(1)}+u^2\mathtt{V}_2^{(2)}+\cdots +u^{\kappa-2}\mathtt{V}_2^{(\kappa-2)} \\ \vdots\\ \mathtt{W}_{s-\kappa_1+1}^{(0)}+u\mathtt{V}_{s-\kappa_1+1}^{(1)}+u^2\mathtt{V}_{s-\kappa_1+1}^{(2)}+\cdots +u^{\kappa-2}\mathtt{V}_{s-\kappa_1+1}^{(\kappa-2)}\vspace{0.5mm}\\ u\mathtt{W}_{s-\kappa_1+2}^{(\kappa-1)}\vspace{0.5mm}\\ \vdots\\u^{\kappa-2}\mathtt{W}_{s+\kappa_1}^{(\kappa-1)}
\end{bmatrix},\end{equation*} 
 where  
\vspace{-1mm}\begin{equation*}\label{e00.2}\vspace{-1mm}
[\mathtt{W}^{(0)}]_{s-\kappa_1+1}=\begin{bmatrix}
\mathtt{W}_1^{(0)}\\\mathtt{W}_2^{(0)}\\ \vdots\\ \mathtt{W}_{s-\kappa_1+1}^{(0)}
\end{bmatrix}  =\begin{bmatrix}
{\mathtt{I}}_{\lambda_1}&{\mathtt{B}}_{1,1}^{(0)}&\cdots&{\mathtt{B}}_{1,s-\kappa_1}^{(0)}&\cdots & {\mathtt{B}}_{1,e-1}^{(0)}& {\mathtt{B}}_{1,e}^{(0)}\\
\mathbf{0} &{\mathtt{I}}_{\lambda_2}  &\cdots &{\mathtt{B}}_{2,s-\kappa_1}^{(0)}&\cdots & {\mathtt{B}}_{2,e-1}^{(0)}&{\mathtt{B}}_{2,e}^{(0)}\\
	\vdots& \vdots   &\vdots&\vdots &\vdots& \vdots&\vdots\\
\mathbf{0}&\mathbf{0}&\cdots& {\mathtt{I}}_{\lambda_{s-\kappa_1+1}}& \cdots & {\mathtt{B}}_{s-\kappa_1+1,e-1}^{(0)}&{\mathtt{B}}_{s-\kappa_1+1,e}^{(0)}
\end{bmatrix}\end{equation*} 
with  ${\mathtt{I}}_{\lambda_i}$ as the $\lambda_i\times \lambda_i$ identity matrix over ${\mathcal{T}}_{m},$    $ {\mathtt{B}}_{i,j}^{(0)} \in \mathcal{M}_{\lambda_i\times \lambda_{j+1}}({\mathcal{T}}_{m})$  for $1 \leq i \leq s-\kappa_1+1$ and $i \leq j \leq e,$
$[\mathtt{V}^{(a)}]_{s-\kappa_1+1}\in \mathcal{M}_{\Lambda_{s-\kappa_1+1}\times n}(\mathcal{T}_{m})$ for $1 \leq a \leq \kappa-2,$ and  the matrix $\mathtt{W}_{s-\kappa_1-1+g}^{(\kappa-1)}\in \mathcal{M}_{\lambda_{s-\kappa_1-1+g}\times n}(\mathscr{R}_{\kappa-1,m})$   to be considered modulo $u^{\kappa+1-g}$   for  $3 \leq g \leq \kappa.$ Note that the torsion code $Tor_1(\mathscr{D}_{\kappa-1})$ is a $\Lambda_{s-\kappa_1+1}$-dimensional code over $\mathcal{T}_{m}$ with a generator matrix $[\mathtt{W}^{(0)}]_{s-\kappa_1+1}.$ 
Here, we also assume, without any loss of generality,  that the code  $\mathcal{C}^{(s-\kappa_1)}$ has a generator matrix $[\mathtt{W}^{(0)}]_{s-\kappa_1},$ the code $\mathcal{C}^{(s-\kappa+1)}$  has a generator matrix $[\mathtt{W}^{(0)}]_{s-\kappa+1} $ and  the code $\mathcal{C}^{(s-\kappa)}$   has a generator matrix $[\mathtt{W}^{(0)}]_{s-\kappa}. $ Moreover, by  Remark \ref{r3.1},  we assume that the matrix $(\mathtt{B}^{(0)})_ { s-\kappa_1,s+\kappa_1+1} $ has full row-rank.   
 Since $\mathscr{D}_{\kappa-1}$ is a self-orthogonal code over $\mathscr{R}_{\kappa-1,m}$ satisfying  property $(\mathfrak{P}),$  we have
\begin{eqnarray*}
[\mathtt{W}^{(\kappa-1)}]_{s-\kappa_1}[\mathtt{W}^{(\kappa-1)}]_{s-\kappa_1}^t &\equiv &\mathbf{0} \pmod{u^{\kappa-1}},\label{}\\
  \displaystyle \mathtt{Diag}\left([\mathtt{W}^{(\kappa-1)}]_{s-\kappa_1+1-\mathrm{f}_i}[\mathtt{W}^{(\kappa-1)}]_{s-\kappa_1+1-\mathrm{f}_i}^t\right) &\equiv & \mathbf{0} \pmod{u^{\kappa-1+i}} \text{ ~~for ~} i\in \{2,4,6,\ldots,\kappa-1\},~\\ 
   \hspace{-2mm} \displaystyle  \pi_{2\kappa-3}\left( \mathtt{Diag}\big([\mathtt{W}^{(\kappa-1)}]_{s-\kappa+2}[\mathtt{W}^{(\kappa-1)}]_{s-\kappa+2}^t\big) \right)&=&\mathbf{0} ,
  \\ 
  ~[\mathtt{W}^{(\kappa-1)}]_{ s-\kappa_1} \mathtt{W}_{s-\kappa_1-1+g}^{(\kappa-1)t}&\equiv &\mathbf{0} \pmod{u^{\kappa+1-g}} \text{ ~for } 2 \leq g \leq  \kappa ,\\
\mathtt{W}_{s-\kappa_1-1+i}^{(\kappa-1)}
 \mathtt{W}_{s-\kappa_1-1+j}^{(\kappa-1)t}&\equiv &\mathbf{0}  \pmod{u^{\kappa+3-i-j}} \text{ for } 2 \leq i,j\leq \kappa \text{ and } i+j\leq \kappa+2.\end{eqnarray*} 
   This implies that \begin{eqnarray*}
[\mathtt{W}^{(\kappa-1)}]_{s-\kappa_1} [\mathtt{W}^{(\kappa-1)}]^{t}_{s-\kappa_1}&\equiv&  u^{\kappa-1}\mathtt{F}^{(\kappa-1)}+u^{\kappa}\mathtt{F}^{(\kappa)}+\cdots+u^{2\kappa}\mathtt{F}^{(2\kappa)} \pmod{u^{2\kappa+1}},\\ ~[\mathtt{W}^{(\kappa-1)}]_{ s-\kappa_1} \mathtt{W}_{s-\kappa_1-1+g}^{(\kappa-1)t}&\equiv &u^{\kappa+1-g}R_{g}  \pmod{u^{\kappa+2-g}} \text{ ~for } 2 \leq g \leq \kappa ,\\
\mathtt{W}_{s-\kappa_1-1+i}^{(\kappa-1)} \mathtt{W}_{s-\kappa_1-1+j}^{(\kappa-1)t}&\equiv &\mathbf{0} \pmod{u^{\kappa+3-i-j}} \text{ for } 2 \leq i,j\leq \kappa \text{ and } i+j\leq \kappa+2,
\end{eqnarray*} where $\mathtt{F}^{(\kappa-1)},\mathtt{F}^{(\kappa)}, \ldots, \mathtt{F}^{(2\kappa)} \in \mathtt{Sym}_{\Lambda_{s-\kappa_1}}(\mathcal{T}_{m})$ with  $\mathtt{F}_{j,j}^{(\kappa-3+i)}=\mathtt{F}^{(\kappa-2+i)}_{j,j}=0 $ for $1 \leq j \leq \Lambda_{s-\kappa_1+1-\mathrm{f}_i }$ and $ i \in \{ 2,4,\ldots,\kappa-1\} ,$  and the matrix $R_{g} \in \mathcal{M}_{\Lambda_{s-\kappa_1} \times \lambda_{s-\kappa_1-1+g}}(\mathcal{T}_{m})  $ for $ 2 \leq g \leq \kappa.$   

Now, to prove the result,  we define a matrix  $\mathtt{G}_{\kappa+1}$ over $\mathscr{R}_{\kappa+1,m}$ as follows:
\vspace{-1mm}\begin{equation*} \vspace{-1mm}\mathtt{G}_{\kappa+1}=\begin{bmatrix}
\mathtt{W}_1^{(\kappa+1)}\\\mathtt{W}_2^{(\kappa+1)}\\ \vdots\\ \mathtt{W}_{s-\kappa_1}^{(\kappa+1)}\\ \vspace{0.5mm}u\mathtt{W}_{s-\kappa_1+1}^{(\kappa+1)}\\\vdots\\u^{\kappa}\mathtt{W}_{s+\kappa_1+1}^{(\kappa+1)}
\end{bmatrix}=\begin{bmatrix} \mathtt{W}_1^{(\kappa-1)}+u^{\kappa-1}\mathtt{V}_1^{(\kappa-1)}+u^{\kappa}\mathtt{V}_1^{(\kappa)}\\ \mathtt{W}_2^{(\kappa-1)}+u^{\kappa-1}\mathtt{V}_2^{(\kappa-1)}+u^{\kappa}\mathtt{V}_2^{(\kappa)}\\ \vdots\\ \mathtt{W}_{s-\kappa_1}^{(\kappa-1)}+u^{\kappa-1}\mathtt{V}_{s-\kappa_1}^{(\kappa-1)}+u^{\kappa}\mathtt{V}_{s-\kappa_1}^{(\kappa)}\\ \vspace{0.5mm}u\mathtt{W}_{s-\kappa_1+1}^{(\kappa+1)}\\\vdots\\u^{\kappa}\mathtt{W}_{s+\kappa_1+1}^{(\kappa+1)}   \end{bmatrix},\end{equation*} 
where the matrices  $[\mathtt{V}^{(\mu)}]_{s-\kappa_1}\in \mathcal{M}_{\Lambda_{s-\kappa_1}\times n}(\mathcal{T}_{m})$ for $ \mu \in \{\kappa-1,\kappa\},$   $\mathtt{W}_{s-\kappa_1-1+g}^{(\kappa+1)}\in \mathcal{M}_{\lambda_{s-\kappa_1-1+g}\times n}(\mathscr{R}_{\kappa+1,m})$   for  $2 \leq g \leq \kappa$ and  $\mathtt{W}_{s+\kappa_1+1}^{(\kappa+1)}\in \mathcal{M}_{\lambda_{s+\kappa_1+1}\times n}(\mathcal{T}_{m})$ are of the forms 
\begin{equation*}\begin{bmatrix}
\mathtt{V}_1^{(\mu)}\\\mathtt{V}_2^{(\mu)}\\ \vdots\\ \mathtt{V}_{s-\kappa_1}^{(\mu)}
\end{bmatrix}= \begin{bmatrix}
0&\cdots&0&\mathtt{B}_{1,\mu+1}^{(\mu)}& \mathtt{B}_{1,\mu+2}^{(\mu)}&\cdots&\mathtt{B}_{1,s-\kappa_1+\mu}^{(\mu)}&\cdots &  \mathtt{B}_{1,e}^{(\mu)}\\
\mathbf{0} & \cdots&\mathbf{0}&\mathbf{0}& \mathtt{B}_{2,\mu+2}^{(\mu)}&  \cdots &\mathtt{B}_{2,s-\kappa_1+\mu}^{(\mu)}&\cdots &\mathtt{B}_{2,e}^{(\mu)}\\
	\vdots& \cdots &\vdots &\vdots &\vdots&\vdots & \vdots&\vdots&\vdots\\
\mathbf{0}&\cdots&\mathbf{0}&\mathbf{0}&\mathbf{0}&\cdots&  \mathtt{B}_{s-\kappa_1,s-\kappa_1+\mu}^{(\mu)}&\cdots &  \mathtt{B}_{s-\kappa_1,e}^{(\mu)}\\
\end{bmatrix},\end{equation*} 
 \begin{equation*} \mathtt{W}_{s-\kappa_1-1+g}^{(\kappa+1)}= \mathtt{W}_{s-\kappa_1-1+g}^{(\kappa-1)}+u^{\kappa+1-g}\begin{bmatrix}
\mathbf{0}&\cdots&\mathbf{0}& \mathtt{B}_{s-\kappa_1-1+g,s+\kappa_1+1}^{(\kappa+1-g)}&\cdots & \mathtt{B}_{s-\kappa_1-1+g,e}^{(\kappa+1-g)}
\end{bmatrix} \text{ and } \end{equation*} 
\begin{equation*} \mathtt{W}_{s+\kappa_1+1}^{(\kappa+1)}= \begin{bmatrix}
\mathbf{0}&\cdots&\mathbf{0}&\mathtt{I}_{\lambda_{s+\kappa_1+1}}& \mathtt{B}_{s+\kappa_1+1,s+\kappa_1+1}^{(0)}& \cdots & \mathtt{B}_{s+\kappa_1+1,e}^{(0)}
\end{bmatrix} \end{equation*}
with   $\mathtt{B}_{i,j}^{(\mu)} \in \mathcal{M}_{\lambda_i \times \lambda_{j+1}}(\mathcal{T}_{m})$ for $1 \leq i \leq s-\kappa_1$ and $ i+\mu \leq j \leq e,$  $\mathtt{B}_{s-\kappa_1-1+g,h}^{(\kappa+1-g)} \in \mathcal{M}_{\lambda_{s-\kappa_1-1+g} \times  \lambda_{h+1}}(\mathcal{T}_{m}) $ for $s+\kappa_1+1 \leq h \leq e,$  and  $\mathtt{B}_{s+\kappa_1+1,a}^{(0)} \in \mathcal{M}_{\lambda_{s+\kappa_1+1} \times  \lambda_{a+1}}(\mathcal{T}_{m})$  for $s+\kappa_1+1 \leq a \leq e. $
Let  $\mathscr{D}_{\kappa+1}$  be a  linear code of length $n$ over $\mathscr{R}_{\kappa+1,m}$ with a generator matrix $\mathtt{G}_{\kappa+1}.$   Note that the code $\mathscr{D}_{\kappa+1}$ is  of  type $\{\Lambda_{s-\kappa_1},\lambda_{s-\kappa_1+1}, \ldots, \lambda_{s+\kappa_1+1}\}$  over $\mathscr{R}_{\kappa+1,m}$ satisfying $Tor_1(\mathscr{D}_{\kappa+1})=\mathcal{C}^{(s-\kappa_1)}$ and $Tor_{i+1}(\mathscr{D}_{\kappa+1})=Tor_{i}(\mathscr{D}_{\kappa-1})$ for $1 \leq i \leq \kappa-1.$  Moreover, by  Lemma \ref{l2.2},  one can easily observe  that the code  $\mathscr{D}_{\kappa+1}$ is a self-orthogonal code over $\mathscr{R}_{\kappa+1,m}$ satisfying property $(\mathfrak{P})$   if and only if there exist matrices  $[\mathtt{V}^{(\kappa-1)}]_{s-\kappa_1},$ $[\mathtt{V}^{(\kappa)}]_{s-\kappa_1},$ $[
\mathbf{0}~\mathbf{0}~\cdots~\mathbf{0}~ \mathtt{B}_{s-\kappa_1-1+g,s+\kappa_1+1}^{(\kappa+1-g)}~\cdots ~\mathtt{B}_{s-\kappa_1-1+g,e}^{(\kappa+1-g)}]
$ $\left[
\mathbf{0}~\mathbf{0}~\cdots~\mathbf{0}~ \mathtt{B}_{s-\kappa_1-1+g,s+\kappa_1+1}^{(\kappa+1-g)}~\cdots ~\mathtt{B}_{s-\kappa_1-1+g,e}^{(\kappa+1-g)}
\right]$ for $2 \leq g \leq \kappa$ and  $\mathtt{W}_{s+\kappa_1+1}^{(\kappa+1)}$
 satisfying  the following system of matrix equations:
\begin{eqnarray}
 \mathtt{F}^{(\kappa-1)}+[\mathtt{W}^{(0)}]_{s-\kappa_1}[\mathtt{V}^{(\kappa-1)}]_{s-\kappa_1}^t+[\mathtt{V}^{(\kappa-1)}]_{s-\kappa_1}[\mathtt{W}^{(0)}]_{s-\kappa_1}^t+
 u\left([\mathtt{W}^{(0)}]_{s-\kappa_1}[\mathtt{V}^{(\kappa)}]_{s-\kappa_1}^t\right.\nonumber\\\left.+[\mathtt{V}^{(\kappa)}]_{s-\kappa_1}[\mathtt{W}^{(0)}]_{s-\kappa_1}^t+[\mathtt{V}^{(1)}]_{s-\kappa_1}[\mathtt{V}^{(\kappa-1)}]_{s-\kappa_1}^t+[\mathtt{V}^{(\kappa-1)}]_{s-\kappa_1}[\mathtt{V}^{(1)}]_{s-\kappa_1}^t+\mathtt{F}^{(\kappa)}\right)&\equiv& \mathbf{0}\pmod{u^2},\label{e3.45Kodd}\\
\mathtt{Diag}\left(\mathtt{E}^{(2\kappa-2)}+[\mathtt{V}^{(\kappa-1)}]_{s-\kappa+1}[\mathtt{V}^{(\kappa-1)}]_{s-\kappa+1}^t \right) &\equiv & \mathbf{0} \pmod{u},\label{e3.46Kodd}~~~~~~~\\
\mathtt{Diag}\left(\mathtt{E}^{(2\kappa-1)}+\eta_0[\mathtt{W}^{(0)}]_{s-\kappa+1}[\mathtt{V}^{(\kappa-1)}]_{s-\kappa+1}^t \right) &\equiv & \mathbf{0} \pmod{u},\label{e3.47Kodd}~~~~~~~\\
\mathtt{Diag}\left(\mathtt{E}^{(2\kappa)}+[\mathtt{V}^{(\kappa)}]_{s-\kappa}[\mathtt{V}^{(\kappa)}]_{s-\kappa}^t +\eta_0[\mathtt{W}^{(0)}]_{s-\kappa}[\mathtt{V}^{(\kappa)}]_{s-\kappa}^t\right.&&\nonumber\\\left.+\eta_0[\mathtt{V}^{(1)}]_{s-\kappa}[\mathtt{V}^{(\kappa-1)}]_{s-\kappa}^t+\eta_1[\mathtt{W}^{(0)}]_{s-\kappa}[\mathtt{V}^{(\kappa-1)}]_{s-\kappa}^t\right) &\equiv & \mathbf{0} \pmod{u},\label{e3.48Kodd}~~~~~~~\\ 
R_{g}+[\mathtt{W}^{(0)}]_{ s-\kappa_1} \begin{bmatrix} \mathbf{0}~ \cdots ~\mathbf{0} ~\mathtt{B}_{s-\kappa_1-1+g,s+\kappa_1+1}^{(\kappa+1-g)} ~\cdots ~ \mathtt{B}_{s-\kappa_1-1+g,e}^{(\kappa+1-g)}\end{bmatrix} ^t&\equiv &\mathbf{0}  \pmod{u}, ~~\label{e3.49Kodd}\\
~[\mathtt{W}^{(0)}]_{s-\kappa_1} \mathtt{W}_{s+\kappa_1+1}^{(\kappa+1)t} &\equiv&~ \mathbf{0}\pmod{u},\label{e3.50Kodd}\end{eqnarray}
 where   $\mathtt{E}^{(2\kappa-2)},$   $\mathtt{E}^{(2\kappa-1)}$ and $\mathtt{E}^{(2\kappa)}$ are  $\Lambda_{s-\kappa+1} \times \Lambda_{s-\kappa+1},$  $\Lambda_{s-\kappa+1} \times \Lambda_{s-\kappa+1}$ and $\Lambda_{s-\kappa} \times \Lambda_{s-\kappa}$ matrices over $\mathcal{T}_{m}$ whose rows are the  first $\Lambda_{s-\kappa+1} ,$  $\Lambda_{s-\kappa+1} $ and $\Lambda_{s-\kappa} $   rows of the matrices $\mathtt{F}^{(2\kappa-2)}, $ $\mathtt{F}^{(2\kappa-1)}$  and $ \mathtt{F}^{(2\kappa)},$     respectively. 
 
We will now establish the existence of the  matrices $[\mathtt{V}^{(\kappa-1)}]_{s-\kappa_1},$ $[\mathtt{V}^{(\kappa)}]_{s-\kappa_1},$ $\Big[\mathbf{0}~ \cdots ~\mathbf{0} ~\mathtt{B}_{s-\kappa_1-1+g,s+\kappa_1+1}^{(\kappa+1-g)} ~\cdots ~ \\ \mathtt{B}_{s-\kappa_1-1+g,e}^{(\kappa+1-g)}\Big] $ for $2 \leq g \leq \kappa$  and $\mathtt{W}_{s+\kappa_1+1}^{(\kappa+1)}$ satisfying the system  of matrix equations \eqref{e3.45Kodd}--\eqref{e3.50Kodd} and count their choices.
To do this,  we  first recall  that  the matrix $(\mathtt{B}^{(0)})_{s-\kappa_1,s+\kappa_1+1}$ has full row-rank over $\mathcal{T}_{m}.$ Working as  in Proposition \ref{p3.4Kodd},  we see that there exists a matrix $[\mathtt{V}^{(\kappa-1)}]_{s-\kappa_1}$  satisfying  the  following system of matrix equations:  
\vspace{-2mm}\begin{eqnarray}
  \mathtt{F}^{(\kappa-1)}+[\mathtt{W}^{(0)}]_{s-\kappa_1}[\mathtt{V}^{(\kappa-1)}]_{s-\kappa_1}^t+[\mathtt{V}^{(\kappa-1)}]_{s-\kappa_1}[\mathtt{W}^{(0)}]_{s-\kappa_1}^t &\equiv & \mathbf{0} \pmod{u},\label{e3.51Kodd}\\
\mathtt{Diag}\left(\mathtt{E}^{(2\kappa-2)}+[\mathtt{V}^{(\kappa-1)}]_{s-\kappa+1}[\mathtt{V}^{(\kappa-1)}]_{s-\kappa+1}^t \right) &\equiv & \mathbf{0} \pmod{u},\label{e3.52Kodd}~~~~~~~\\
\mathtt{Diag}\left(\mathtt{E}^{(2\kappa-1)}+\eta_0[\mathtt{W}^{(0)}]_{s-\kappa+1}[\mathtt{V}^{(\kappa-1)}]_{s-\kappa+1}^t \right) &\equiv & \mathbf{0} \pmod{u}.\label{e3.53Kodd}
\end{eqnarray} 
   Moreover, the matrix  $[\mathtt{V}^{(\kappa-1)}]_{s-\kappa_1}$  satisfying \eqref{e3.51Kodd}--\eqref{e3.53Kodd} has precisely \vspace{-1mm}\begin{equation*}\vspace{-1mm}
  \displaystyle (2^m)^{\sum\limits_{i=\kappa+1}^{s+\kappa_1+1}\lambda_i\Lambda_{i-\kappa}+\Lambda_{s-\kappa_1}(n-\Lambda_{s+\kappa_1+1})-2\Lambda_{s-\kappa+1}-\frac{\Lambda_{s-\kappa_1}(\Lambda_{s-\kappa_1}-1)}{2} +\omega_{\kappa+1}}  
\end{equation*} distinct choices, where $\omega_{\kappa+1}=1 $ if $\textbf{1}\in \mathcal{C}^{(s-\kappa+1)},$   while $\omega_{\kappa+1}=0$ otherwise. 
 Further, for a given choice of the matrix $[\mathtt{V}^{(\kappa-1)}]_{s-\kappa_1}$  satisfying \eqref{e3.51Kodd}--\eqref{e3.53Kodd},  we can write
\begin{equation}\label{e3.54Kodd}
 \mathtt{F}^{(\kappa-1)}+[\mathtt{W}^{(0)}]_{s-\kappa_1}[\mathtt{V}^{(\kappa-1)}]_{s-\kappa_1}^t+[\mathtt{V}^{(\kappa-1)}]_{s-\kappa_1}[\mathtt{W}^{(0)}]_{s-\kappa_1}^t ~\equiv ~ uQ \pmod{u^2}
\end{equation}
for some  $Q\in \mathtt{Sym}_{\Lambda_{s-\kappa_1}}(\mathcal{T}_{m}) .$ 
This,  by \eqref{e3.45Kodd},  implies that
\vspace{-1mm}\begin{eqnarray}
Q+\mathtt{F}^{(\kappa)}+[\mathtt{W}^{(0)}]_{s-\kappa_1}[\mathtt{V}^{(\kappa)}]_{s-\kappa_1}^t+[\mathtt{V}^{(\kappa)}]_{s-\kappa_1}[\mathtt{W}^{(0)}]_{s-\kappa_1}^t&&\nonumber\\+[\mathtt{V}^{(1)}]_{s-\kappa_1}[\mathtt{V}^{(\kappa-1)}]_{s-\kappa_1}^t+[\mathtt{V}^{(\kappa-1)}]_{s-\kappa_1}[\mathtt{V}^{(1)}]_{s-\kappa_1}^t&\equiv& \mathbf{0}\pmod{u}.\label{e3.55Kodd}
\end{eqnarray}
Furthermore, working as in Lemma 4.1 of Yadav and Sharma \cite{Galois},  we observe that there exists a matrix $[\mathtt{V}^{(\kappa)}]_{s-\kappa_1}$ satisfying  \eqref{e3.48Kodd} and \eqref{e3.55Kodd} if and only if exactly one the following two conditions is satisfied: (I) $\textbf{1}\notin \mathcal{C}^{(s-\kappa)},$ and (II) $\textbf{1}\in \mathcal{C}^{(s-\kappa)}$ with either $n\equiv 0\pmod8$ or $n\equiv 4\pmod{8} $ and $m$  even.    Moreover, the matrix $[\mathtt{V}^{(\kappa)}]_{s-\kappa_1}$  satisfying  \eqref{e3.48Kodd} and \eqref{e3.55Kodd} has precisely \vspace{-1mm}\begin{equation*}
  \vspace{-1mm}  \displaystyle 2^{\epsilon}(2^m)^{\sum\limits_{j=\kappa+2}^{s+\kappa_1+1}\lambda_j\Lambda_{j-\kappa-1}+\Lambda_{s-\kappa_1}(n-\Lambda_{s+\kappa_1+1})-\Lambda_{s-\kappa}-\frac{\Lambda_{s-\kappa_1}(\Lambda_{s-\kappa_1}-1)}{2}} 
\end{equation*} distinct choices,  where $\epsilon=1$ if $\textbf{1}\in \mathcal{C}^{(s-\kappa)}$  with either $n\equiv 0\pmod{8}$ or $n\equiv 4\pmod8$ and $m$ is even, while  $\epsilon=0 $ otherwise.   
Further, working as in Proposition 4.3 of Yadav and Sharma \cite{Galois}, we see that the  matrices 
$\begin{bmatrix} \mathbf{0}~ \cdots ~\mathbf{0} ~\mathtt{B}_{s-\kappa_1-1+g,s+\kappa_1+1}^{(\kappa+1-g)} ~\cdots ~ \mathtt{B}_{s-\kappa_1-1+g,e}^{(\kappa+1-g)}\end{bmatrix} $ for $2 \leq g \leq \kappa$ and $\mathtt{W}_{s+\kappa_1+1}^{(\kappa+1)}$
satisfying \eqref{e3.49Kodd} and \eqref{e3.50Kodd} have precisely $$ (2^m)^{(\Lambda_{s+\kappa_1}-\Lambda_{s-\kappa_1})(n-\Lambda_{s+\kappa_1+1}-\Lambda_{s-\kappa_1})}{\lambda_{s+\kappa_1+1}+n-\Lambda_{s+\kappa_1+1}-\Lambda_{s-\kappa_1} \brack \lambda_{s+\kappa_1+1}}_{2^m}$$ distinct choices. 
From this,  the desired result follows immediately. 
     \item[(II)]  When $e$ is odd, the desired result follows by working as in case (I).
\vspace{-3mm} \end{itemize}\vspace{-4mm}\end{proof}
The following proposition lifts a  self-orthogonal code of type $\{\Lambda_{\gamma_{\ell}+2},\lambda_{\gamma_{\ell}+3}, \ldots, \lambda_{\gamma_{\ell}+\ell-1}\}$   and  length $n$ over $\mathscr{R}_{\ell-2,m}$ satisfying  property $(\mathfrak{P})$ to a   self-orthogonal code  of type $\{\Lambda_{\gamma_{\ell}+1},\lambda_{\gamma_{\ell}+2}, \ldots, \lambda_{\gamma_{\ell}+\ell}\}$ and the same length $n$ over $\mathscr{R}_{\ell,m}$ that also satisfies property $(\mathfrak{P}),$ where the integer $\ell$ satisfies  $\ell \equiv e\pmod2$ and $\kappa+3 \leq \ell \leq e $ when $2\kappa \leq e$, while  it satisfies  $\ell \equiv e\pmod2$ and  $e-\kappa+2-2\theta_e \leq  \ell \leq e$   when $2\kappa >e.$ Additionally, it determines the number of distinct ways in which this lifting can be performed.
\vspace{-1mm}\begin{proposition}\label{p3.9KoddReplace} Let  $\ell$ be a fixed integer satisfying $\ell \equiv e\pmod2$ and $\kappa+3 \leq \ell \leq e $ if $2\kappa \leq e$, while the integer $\ell$ satisfies $\ell \equiv e\pmod2$ and $e-\kappa+2-2\theta_e \leq  \ell \leq e$   if $2\kappa >e.$ Let us define 
$\gamma_{\ell}=s-\mathrm{f}_{\ell}.$
 Let $\mathscr{D}_{\ell-2}$ be a self-orthogonal code of type $\{\Lambda_{\gamma_{\ell}+2},\lambda_{\gamma_{\ell}+3}, \ldots, \lambda_{\gamma_{\ell}+\ell-1}\}$   and  length $n$ over $\mathscr{R}_{\ell-2,m}$ satisfying  property $(\mathfrak{P}).$ Let $\mathcal{C}^{(\gamma_{\ell}+1)}$ be a  $\Lambda_{\gamma_{\ell}+1}$-dimensional  linear subcode of the torsion code $Tor_1(\mathscr{D}_{\ell-2}).$ 
The following hold.
 \begin{enumerate}\item[(a)] There exists a  self-orthogonal code $\mathscr{D}_{\ell}$ of type $\{\Lambda_{\gamma_{\ell}+1},\lambda_{\gamma_{\ell}+2}, \ldots, \lambda_{\gamma_{\ell}+\ell}\}$ and length $n$ over $\mathscr{R}_{\ell,m}$ satisfying property $(\mathfrak{P})$ with $Tor_1(\mathscr{D}_{\ell})=\mathcal{C}^{(\gamma_{\ell}+1)}$ and  $Tor_{i+1}(\mathscr{D}_{\ell})=Tor_{i}(\mathscr{D}_{\ell-2})$ for $1 \leq i \leq  \ell-2.$
\item[(b)] Moreover,  each such pair  of codes $\mathscr{D}_{\ell-2}$ and $\mathcal{C}^{(\gamma_{\ell}+1)}$ gives rise to precisely \vspace{-2mm}\begin{eqnarray*}
\displaystyle(2^m)^{\sum\limits_{i=\ell}^{\gamma_{\ell}+\ell}\lambda_i\Lambda_{i-\ell+1}+\sum\limits_{j=\ell+1}^{\gamma_{\ell}+\ell}\lambda_j\Lambda_{j-\ell}+(\Lambda_{\gamma_{\ell}+\ell-1}+\Lambda_{\gamma_{\ell}+1})(n-\Lambda_{\gamma_{\ell}+\ell}-\Lambda_{\gamma_{\ell}+1})+\Lambda_{\gamma_{\ell}+1}^2+\Lambda_{\gamma_{\ell}+1}-Y_{\ell}} {\lambda_{\gamma_{\ell}+\ell}+n-\Lambda_{\gamma_{\ell}+\ell}-\Lambda_{\gamma_{\ell}+1} \brack \lambda_{\gamma_{\ell}+\ell}}_{2^m}
\vspace{-2mm}\end{eqnarray*}  distinct  self-orthogonal codes $\mathscr{D}_{\ell}$ of type $\{\Lambda_{\gamma_{\ell}+1},\lambda_{\gamma_{\ell}+2}, \ldots, \lambda_{\gamma_{\ell}+\ell}\}$ and length $n$ over $\mathscr{R}_{\ell,m}$ satisfying property $(\mathfrak{P})$ with  $Tor_1(\mathscr{D}_{\ell})=\mathcal{C}^{(\gamma_{\ell}+1)}$ and $Tor_{i+1}(\mathscr{D}_{\ell})=Tor_{i}(\mathscr{D}_{\ell-2})$ for $1 \leq i \leq  \ell-2,$ where 
\vspace{-1mm}\begin{equation*}
    \vspace{-1mm}Y_{\ell}=\left\{ \begin{array}{ll}
     \Lambda_{\gamma_{\ell}+1-\kappa_1}+\Lambda_{\gamma_{\ell}-\kappa_1} & \text{if } 2\kappa\leq e \text{ and } \kappa+3 \leq \ell \leq e-\kappa+1;  \\
    \Lambda_{s-2\mathrm{f}_{\ell}+2}+\Lambda_{s-2\mathrm{f}_{\ell}+1} &  \text{if }2\kappa> e \text{ and } e-\kappa+2-2\theta_e \leq \ell \leq \kappa-\mathrm{f}_{2\kappa-e}+1;\\
     0 & \text{if either }2\kappa\leq e \text{ and } e-\kappa+2  \leq \ell \leq e \text{ or } 2\kappa >e \text{ and } \kappa-\mathrm{f}_{2\kappa-e}+1 < \ell \leq e.
 \end{array}\right.\end{equation*}
 \end{enumerate} \end{proposition}
 \begin{proof}When $2\kappa >e$ and $e-\kappa+2-2\theta_e \leq \ell \leq \kappa-\mathrm{f}_{2\kappa-e}+1,$ the desired result follows immediately  by working as in Proposition \ref{p3.4Kodd} and applying Lemma 4.1 of Yadav and Sharma \cite{quasi}.  On the other hand,
 when either $2\kappa \leq e$ and $\kappa+3\leq \ell \leq e$ or $2\kappa >e $ and $\kappa-\mathrm{f}_{2\kappa-e}+1 <\ell \leq e,$ the desired result follows directly by
  proceeding as in Proposition 4.3 of Yadav and Sharma \cite{Galois}.  
 \end{proof}

\vspace{2mm} \tikzstyle{startstop} = [rectangle, rounded corners, minimum width=3.5cm, minimum height=1cm,text centered, draw=black, fill=white!30, text width=7.5cm, align=left]
\tikzstyle{process} = [rectangle, minimum width=3.5cm, minimum height=1cm, text centered, draw=black, fill=white!20, text width=7.5cm, align=left]
\tikzstyle{arrow} = [thick,->,>=stealth]

\noindent
\begin{minipage}{0.48\textwidth}
\centering
\begin{tikzpicture}[node distance=2cm, font=\small]
\node (start1) [startstop]  {Input: A chain  $\mathcal{C}^{(1)}\subseteq \mathcal{C}^{(2)}\subseteq\cdots \subseteq \mathcal{C}^{(s+\theta_e)}$ of self-orthogonal codes of length $n$ over $\mathcal{T}_m$, such that $\dim \mathcal{C}^{(i)}=\Lambda_i$ for $1 \leq i \leq s+\theta_e,$ with $\mathcal{C}^{(s-\kappa_1)}$ being doubly even, and satisfying the additional condition that $\textbf{1}\notin \mathcal{C}^{(s-\kappa+\theta_e)}$ whenever $n\equiv 4\pmod 8$ and $m$ is odd.};
\node (step1a) [process, below of=start1, node distance=2.6cm] {Step 1: If $e$ is even, construct a self-orthogonal code $\mathscr{D}_2$ satisfying property $(\mathfrak{P})$ using Proposition \ref{p3.2}; if $e$ is odd, construct a self-orthogonal code $\mathscr{D}_3$ satisfying property $(\mathfrak{P})$  using Proposition \ref{p3.3a}.};
\node (step2a) [process, below of=step1a, node distance=2cm] {Step 2: For an integer $\ell$ with $4 \leq \ell \leq \kappa$ and $\ell \equiv e \pmod 2$, construct a self-orthogonal code $\mathscr{D}_\ell$ satisfying property $(\mathfrak{P})$ from $\mathscr{D}_{\ell-2}$  using Proposition \ref{p3.4Kodd}.};
\node (step3a) [process, below of=step2a, node distance=2cm] {Step 3: Construct a self-orthogonal code $\mathscr{D}_{\kappa+1}$ (if $e$ is even) or $\mathscr{D}_{\kappa+2}$ (if $e$ is odd) satisfying property $(\mathfrak{P})$ using Proposition \ref{p3.6Kodd}.};
\node (step4a) [process, below of=step3a, node distance=2cm] {Step 4: For an integer $\ell$ with $\kappa+3 \leq \ell \leq e$ and $\ell \equiv e \pmod 2$, construct a self-orthogonal code $\mathscr{D}_\ell$ satisfying property $(\mathfrak{P})$ from $\mathscr{D}_{\ell-2}$ using Proposition \ref{p3.9KoddReplace}.};

\node (output1) [startstop, below of=step4a] {Output: A self-orthogonal code $\mathscr{D}_e$ of type 
$\{\lambda_1,\lambda_2, \ldots,\lambda_e\}$ and length $n$ over  $\mathscr{R}_{e,m}$ with $Tor_i(\mathscr{D}_e) = \mathcal{C}^{(i)}$ for $1 \le i \le s+\theta_e.$};
\draw [arrow] (start1) -- (step1a);
\draw [arrow] (step1a) -- (step2a);
\draw [arrow] (step2a) -- (step3a);
\draw [arrow] (step3a) -- (step4a);
\draw [arrow] (step4a) -- (output1);
\end{tikzpicture}
\captionof{figure}{Flowchart illustrating the construction of a self-orthogonal code of type $\{\lambda_1,\lambda_2, \ldots,\lambda_e\}$ and length $n$ over $\mathscr{R}_{e,m}$  when $2\kappa \leq e$}
\label{F1}
\end{minipage}%
\hfill
\begin{minipage}{0.48\textwidth}
\centering
\begin{tikzpicture}[node distance=2cm, font=\small]
\node (start2) [startstop] {Input: A chain  $\mathcal{C}^{(1)}\subseteq \mathcal{C}^{(2)}\subseteq \cdots\subseteq \mathcal{C}^{(s+\theta_e)}$ of self-orthogonal codes of length $n$ over $\mathcal{T}_m,$ such that $\dim \mathcal{C}^{(i)}=\Lambda_i$ for $1 \leq i \leq s+\theta_e,$ and the code $\mathcal{C}^{(s-\kappa_1)}$ is doubly even. };
\node (step1b) [process, below of=start2, node distance=2.4cm] {Step 1: If $e$ is even, construct a self-orthogonal code $\mathscr{D}_2$ satisfying property $(\mathfrak{P})$ using Proposition \ref{p3.2}; if $e$ is odd, construct a self-orthogonal code $\mathscr{D}_3$ satisfying property $(\mathfrak{P})$ using Proposition \ref{p3.3a}.};
\node (step2b) [process, below of=step1b, node distance=2.4cm] {Step 2:  For an integer $\ell$ with $4 \leq \ell \leq e-\kappa+1-2\theta_e$ and $\ell \equiv e \pmod 2$, construct a self-orthogonal code $\mathscr{D}_\ell$ satisfying property $(\mathfrak{P})$ from $\mathscr{D}_{\ell-2}$ using Proposition \ref{p3.4Kodd}.};
\node (step3b) [process, below of=step2a, node distance=2.5cm] {Step 3: For an integer $\ell$ with $e-\kappa+2-2\theta_e \leq \ell \leq e$ and $\ell \equiv e \pmod 2$, construct a self-orthogonal code $\mathscr{D}_\ell$  satisfying property $(\mathfrak{P})$  from $\mathscr{D}_{\ell-2}$ using Proposition \ref{p3.9KoddReplace}.};
\node (output2) [startstop, below of=step3b, node distance=2.2cm] {Output:  A 
self-orthogonal code $\mathscr{D}_e$ of type $\{\lambda_1,\lambda_2,\ldots,\lambda_e\}$ and length $n$  over $\mathscr{R}_{e,m}$ satisfying $Tor_i(\mathscr{D}_e)=\mathcal{C}^{(i)}$ for $1 \leq i \leq s+\theta_e.$};

\draw [arrow] (start2) -- (step1b);
\draw [arrow] (step1b) -- (step2b);
\draw [arrow] (step2b) -- (step3b);
\draw [arrow] (step3b) -- (output2);
\end{tikzpicture}
\captionof{figure}{Flowchart illustrating the construction of a self-orthogonal code of type $\{\lambda_1,\lambda_2, \ldots,\lambda_e\}$ and length $n$ over $\mathscr{R}_{e,m}$  when   $2\kappa > e$}
\label{F2}
\end{minipage}

\vspace{2mm}  In the following lemma, we show that for every self-orthogonal code  of type $\{\lambda_1,\lambda_2,\ldots,\lambda_e\}$ and length $n$  over $\mathscr{R}_{e,m},$  there exists a chain $\mathcal{C}^{(1)}\subseteq \mathcal{C}^{(2)} \subseteq \cdots \subseteq \mathcal{C}^{(s+\theta_e)} $ of self-orthogonal codes of length $n$ over $\mathcal{T}_m,$  such that $\dim \mathcal{C}^{(i)}=\Lambda_i$ for $1 \leq i \leq s+\theta_e,$  the code $\mathcal{C}^{(s-\kappa_1)}$ is  doubly even, and the code $\mathcal{C}^{(s-\kappa+\theta_e)}$ satisfies  the additional condition  $\textbf{1} \notin \mathcal{C}^{(s-\kappa+\theta_e)}$ provided that $2\kappa \leq e,$  $n\equiv 4\pmod 8$ and $m$ is odd.  
\begin{lemma}\label{l3.4} 
    For every self-orthogonal code  of type $\{\lambda_1,\lambda_2,\ldots,\lambda_e\}$ and length $n$  over $\mathscr{R}_{e,m},$ there exists a chain $\mathcal{C}^{(1)}\subseteq \mathcal{C}^{(2)} \subseteq \cdots \subseteq \mathcal{C}^{(s+\theta_e)} $ of self-orthogonal codes of length $n$ over $\mathcal{T}_m,$ where $\dim \mathcal{C}^{(i)}=\Lambda_i$ for $1 \leq i \leq s+\theta_e,$  the code $\mathcal{C}^{(s-\kappa_1)}$ is  doubly even, and the code $\mathcal{C}^{(s-\kappa+\theta_e)}$ satisfies the additional condition  $\textbf{1} \notin \mathcal{C}^{(s-\kappa+\theta_e)}$ provided that $2\kappa \leq e,$  $n\equiv 4\pmod 8$ and $m$ is odd.  
 \end{lemma}
 \begin{proof} 
To prove the result, let $\mathscr{D}_e$ be a self-orthogonal code of type $\{\lambda_1,\lambda_2,\ldots,\lambda_e\}$ and length $n$  over $\mathscr{R}_{e,m}.$ 
     By Lemma \ref{l1.2}, we see  that the $i$-th torsion code $Tor_i(\mathscr{D}_e)$ is a self-orthogonal code of length $n$ and dimension $\Lambda_i$ over $\mathcal{T}_m$ for $1 \leq i \leq s+\theta_e.$ It is easy to see that $Tor_1(\mathscr{D}_e)\subseteq Tor_2(\mathscr{D}_e) \subseteq \cdots \subseteq Tor_{s+\theta_e}(\mathscr{D}_e).$ Further, by Lemma \ref{l3.3}, we note that the torsion code $Tor_{s-\kappa_1}(\mathscr{D}_e)$ is a doubly even code over $\mathcal{T}_m.$

   We next assert that $\textbf{1} \notin Tor_{s-\kappa+\theta_e}(\mathscr{D}_e)$  when $2\kappa \leq e,$ $n\equiv 4 \pmod 8$ and $m$ is odd.  To prove this assertion, let us suppose that the code $\mathscr{D}_e$ has a generator matrix $\mathcal{G}_e$ as defined in \eqref{Ge}. Further,  corresponding to the code $\mathscr{D}_e,$ let us  consider a linear code $\mathscr{D}_{\ell}$ of type $\{\Lambda_{\gamma_{\ell}+1},\lambda_{\gamma_{\ell}+2}, \ldots, \lambda_{\gamma_{\ell}+\ell}\}$ and length $n$ over $\mathscr{R}_{\ell,m}$  with a generator matrix $\mathcal{G}_{\ell}$ (as defined in \eqref{Gl}) for each integer $\ell$ satisfying  $2 \leq \ell \leq e-2$ and $\ell \equiv e \pmod 2.$ 
By Lemma \ref{l3.2}, we see that  $\mathscr{D}_{\ell}$ is  a self-orthogonal code over $\mathscr{R}_{\ell,m}$ satisfying property $(\mathfrak{P}).$   In particular,  $\mathscr{D}_{\kappa-1+\theta_e}$ is a self-orthogonal code of type $\{\Lambda_{s-\kappa_1+1},\lambda_{s-\kappa_1+2}, \ldots, \lambda_{s+\kappa_1+\theta_e}\}$  over $\mathscr{R}_{\kappa-1+\theta_e,m}$ satisfying property $(\mathfrak{P})$ and $\mathscr{D}_{\kappa+1+\theta_e}$ is a  self-orthogonal  code  of type $\{\Lambda_{s-\kappa_1},\lambda_{s-\kappa_1+1}, \ldots, \lambda_{s+\kappa_1+1+\theta_e}\}$ over $\mathscr{R}_{\kappa+1+\theta_e,m}$ satisfying  property $(\mathfrak{P}).$
Now, working as in Proposition  \ref{p3.6Kodd}, it is easy to see that $\textbf{1} \notin Tor_{s-\kappa+\theta_e}(\mathscr{D}_e)$ when $2 \kappa \leq e,$ $n\equiv 4 \pmod 8$ and $m$ is odd. On taking $\mathcal{C}^{(i)}=Tor_i(\mathscr{D}_e)$ for $1 \leq i \leq s+\theta_e,$ the desired result follows.
 \end{proof}
 Now, let $\mathcal{C}^{(1)}\subseteq \mathcal{C}^{(2)} \subseteq \cdots \subseteq \mathcal{C}^{(s+\theta_e)} $ be a chain  of self-orthogonal codes of length $n$ over $\mathcal{T}_m,$ where $\dim \mathcal{C}^{(i)}=\Lambda_i$ for $1 \leq i \leq s+\theta_e,$  the code $\mathcal{C}^{(s-\kappa_1)}$ is doubly even, and  the code $\mathcal{C}^{(s-\kappa+\theta_e)}$ satisfies the additional condition $\textbf{1} \notin \mathcal{C}^{(s-\kappa+\theta_e)},$ provided  that $2\kappa \leq e,$  $n\equiv 4\pmod 8$ and $m$ is odd. 
 Note that the proofs of Propositions \ref{p3.2}--\ref{p3.9KoddReplace} provide a method to construct a 
self-orthogonal code $\mathscr{D}_e$ of type $\{\lambda_1,\lambda_2,\ldots,\lambda_e\}$ and length $n$  over $\mathscr{R}_{e,m}$ satisfying $Tor_i(\mathscr{D}_e)=\mathcal{C}^{(i)}$ for $1 \leq i \leq s+\theta_e.$ Figures \ref{F1} and \ref{F2} present flowcharts outlining the construction method, distinguishing between the two following cases: (I) $2 \kappa \leq e,$ and  (II) $2 \kappa > e.$

 When  $\lambda_j=\lambda_{e-j+2}$ for $1 \leq j \leq e,$ we see, by Lemma  \ref{l2.2}, that the construction methods outlined in Figures \ref{F1} and \ref{F2} give rise to  a self-dual code $\mathscr{D}_e$ of type  $\{\lambda_1,\lambda_2,\ldots,\lambda_{e}\}$ and length $n$ over $\mathscr{R}_{e,m},$ satisfying $Tor_i(\mathscr{D}_e)=\mathcal{C}^{(i)}$ for $1 \leq i \leq s+\theta_e,$ in the cases $2 \kappa \leq e$ and $2 \kappa >e,$ respectively.  Additionally, the recursive construction methods   coincide with that used in \cite{quasi} when
$\mathfrak{s}=1.$

In the following example, we illustrate that the  self-orthogonal code $\mathcal{C}_0,$ as defined  in Example \ref{example1}, can be lifted to a self-orthogonal code over $\mathscr{R}_{8,2}$ by applying  the recursive construction method outlined in Figure \ref{F1}, since $2 \kappa=6 < 8=e$ in this case.

\begin{example}\label{Examle3.1}
    
 Let  $\mathscr{R}_{8,2}$  be the finite commutative chain ring 
as defined in Example \ref{example1}.   
 Let $n=4,$ $\lambda_1=1$  and $\lambda_2=\lambda_3=\lambda_4=\lambda_5=\lambda_6=\lambda_7=\lambda_8=0.$ Let $\mathcal{C}_0$ be a self-orthogonal  code of length $4$ and dimension $1$ over $\mathcal{T}_2$ with a generator matrix $\begin{bmatrix}1& 1&1&1\end{bmatrix}$ (as defined in Example \ref{example1}). Note  that $\pi_{1}(\textbf{d}\cdot \textbf{d})=\pi_{3}(\textbf{d}\cdot \textbf{d})=0$ for all $\textbf{d}\in \mathcal{C}_0,$ where each $\textbf{d}\cdot \textbf{d}$ is viewed as an element of $\mathscr{R}_{8,2}.$ 
 
 Here, we will show that the code $\mathcal{C}_0$ can be lifted to a self-orthogonal code of type $\{1,0,0,0,0,0,0,0\}$ and length $4$ over $\mathscr{R}_{8,2}$ by applying the recursive construction method outlined in Figure \ref{F1}.
Towards this, let us consider a linear code $\mathscr{C}_2$ of type $\{1,0\} $ and length $4$ over $\mathscr{R}_{2,2}$ with a generator matrix \vspace{-1.5mm}\begin{equation*}\vspace{-1.5mm}\mathcal{G}_2=\begin{bmatrix}
    1&1&1&1
\end{bmatrix}+u\begin{bmatrix}
    0&a_1&b_1&c_1
\end{bmatrix},  \end{equation*}  where $a_1,b_1,c_1\in \mathcal{T}_2.$ Now, by the recursive construction method outlined in Figure \ref{F1}, we see that the code $\mathcal{C}_0$ can be lifted to a self-orthogonal code of type $\{1,0,0,0,0,0,0,0\}$ over $\mathscr{R}_{8,2}$ if and only if $\mathscr{C}_2$ is a   self-orthogonal code   over $\mathscr{R}_{2,2}$ satisfying property $(\mathfrak{P}),$  which holds if and only if   $a_1^2+b_1^2+c_1^2\equiv 0\pmod{ u^2},$ or equivalently,  $a_1+b_1+c_1\equiv 0\pmod u$ holds. It is easy to see  that $a_1,b_1,c_1\in \mathcal{T}_2$ satisfying $a_1+b_1+c_1\equiv 0\pmod u$ have precisely $16$ distinct choices.  Thus,  for a fixed choice of $a_1,b_1,c_1\in \mathcal{T}_2$ satisfying $a_1+b_1+c_1\equiv 0\pmod u,$ the code $\mathscr{C}_2$ of length $4$ over $\mathscr{R}_{2,2}$ with a generator matrix $\mathcal{G}_2$ is a self-orthogonal code satisfying property $(\mathfrak{P}).$ Further, note that $\mathcal{G}_2\mathcal{G}_2^t \equiv u^2(a_1^2+b_1^2+c_1^2)+u^4(a_1+b_1+c_1)+u^6 \pmod{u^7}.$ Since $a_1+b_1+c_1\equiv 0\pmod u,$ let us suppose that $a_1+b_1+c_1\equiv u\gamma_1+u^2\gamma_2+u^3\gamma_3+u^4\gamma_4+u^5\gamma_5+u^6\gamma_6\pmod{u^7},$ where $\gamma_i\in \mathcal{T}_2$ for $1\leq i \leq 6.$ This implies that $\mathcal{G}_2\mathcal{G}_2^t \equiv u^4\gamma_1^2+u^5\gamma_1+u^6(1+\gamma_2+\gamma_2^2)\pmod{u^7}.$

Now, corresponding to such a choice of the code $\mathscr{C}_2$,  consider a linear code $\mathscr{C}_4$ of type $\{1,0,0,0\} $ and length $4$ over $\mathscr{R}_{4,2}$ with a generator matrix \vspace{-1.5mm}\begin{equation*}
    \vspace{-1.5mm}\mathcal{G}_4=\mathcal{G}_2+u^2\begin{bmatrix}
0&a_2&b_2&c_2\end{bmatrix}+u^3\begin{bmatrix}
    0&a_3&b_3&c_3
\end{bmatrix},
\end{equation*}    where $a_2,b_2,c_2,a_3,b_3,c_3\in \mathcal{T}_2.$ Using again the recursive construction method outlined in Figure \ref{F1},   we see that the code $\mathscr{C}_4$ can be lifted to a self-orthogonal code of type $\{1,0,0,0,0,0,0,0\}$ over $\mathscr{R}_{8,2}$ if and only if $\mathscr{C}_4$ is a   self-orthogonal code   over $\mathscr{R}_{4,2}$ satisfying property $(\mathfrak{P}),$ which holds if and only if $\mathcal{G}_4\mathcal{G}_4^t\equiv 0\pmod{u^7}.$ This holds if and only if  $\gamma_1^2+a_2^2+b_2^2+c_2^2 +u(\gamma_1+a_2+b_2+c_2)+u^2(1+\gamma_2+\gamma_2^2+a_3^2+b_3^2+c_3^2+a_3+b_3+c_3)\equiv 0\pmod{u^3}.$
Let us first choose $a_2,b_2,c_2\in \mathcal{T}_2$ satisfying $a_2+b_2+c_2+\gamma_1\equiv 0\pmod u,$ (note that $a_2,b_2,c_2$ have precisely $16$ distinct choices). 
Next, for a fixed choice of $a_2,b_2,c_2\in \mathcal{T}_2$ satisfying $a_2+b_2+c_2+\gamma_1\equiv 0\pmod u,$ let us suppose that $a_2+b_2+c_2+\gamma_1\equiv uz_1+u^2z_2\pmod{u^3}.$  Thus, we  need to choose $a_3,b_3,c_3$ satisfying  \vspace{-2mm}\begin{equation}\label{exampleequation}
1+\gamma_2+\gamma_2^2+a_3^2+b_3^2+c_3^2+a_3+b_3+c_3+z_1+z_1^2\equiv 0\pmod u.
\vspace{-2mm}\end{equation}
It is easy to see that there are precisely $32$ distinct choices for $a_3,b_3,c_3\in \mathcal{T}_2$ satisfying \eqref{exampleequation}.
Now,  for  a fixed choice of $a_3,b_3,c_3$ satisfying \eqref{exampleequation}, we see that the code $\mathscr{C}_4$ of length $4$ over $\mathscr{R}_{4,2}$ with a generator matrix $\mathcal{G}_4$ is a self-orthogonal code satisfying property $(\mathfrak{P}),$ and we have that $\mathcal{G}_4\mathcal{G}_4^t\equiv u^7w\pmod{u^8},$ where $w\in \mathcal{T}_2.$ 

Further, corresponding to such a choice of the code $\mathscr{C}_4$,  consider a linear code $\mathscr{C}_6$ of type $\{1,0,0,0,0,0\} $ and length $4$ over $\mathscr{R}_{6,2}$ with a generator matrix \vspace{-1mm}$$\mathcal{G}_6=\mathcal{G}_4+u^4\begin{bmatrix}
0&a_4&b_4&c_4\end{bmatrix}+u^5\begin{bmatrix}
    0&a_5&b_5&c_5
\end{bmatrix},  \vspace{-1mm}$$  where $a_4,b_4,c_4,a_5,b_5,c_5\in \mathcal{T}_2.$  By the recursive construction method outlined in Figure \ref{F1},  we see that the code $\mathscr{C}_6$ can be lifted to a self-orthogonal code of type $\{1,0,0,0,0,0,0,0\}$ over $\mathscr{R}_{8,2}$ if and only if $\mathscr{C}_6$ is a   self-orthogonal code   over $\mathscr{R}_{6,2}$ satisfying property $(\mathfrak{P}),$ which holds  if and only if $w+a_4+b_4+c_4\equiv 0 \pmod u.$ It is easy to see that $a_4,b_4,c_4,a_5,b_5,c_5\in \mathcal{T}_2$ satisfying $w+a_4+b_4+c_4\equiv 0 \pmod u$ have precisely $1024$ distinct choices. Thus, for a fixed choice of $a_4,b_4,c_4,a_5,b_5,c_5 \in \mathcal{T}_2$ satisfying $w+a_4+b_4+c_4\equiv 0 \pmod u,$  we see that the  code $\mathscr{C}_6$ of length $4$ over  $\mathscr{R}_{6,2}$ with a generator matrix $\mathcal{G}_6$ is a self-orthogonal code   over $\mathscr{R}_{6,2}$ satisfying property $(\mathfrak{P}).$ 

Finally,  consider a linear code $\mathscr{C}_8$ of type $\{1,0,0,0,0,0,0,0\} $ and length $4$ over $\mathscr{R}_{8,2}$ with a generator matrix \vspace{-1mm}$$\mathcal{G}_8=\mathcal{G}_6+u^6\begin{bmatrix}
0&a_6&b_6&c_6\end{bmatrix}+u^7\begin{bmatrix}
    0&a_7&b_7&c_7
\end{bmatrix},  $$ \vspace{-1mm} where $a_6,b_6,c_6,a_7,b_7,c_7\in \mathcal{T}_2.$ Note that the code $\mathscr{C}_8$ is  self-orthogonal for all choices of $a_6,b_6,c_6,a_7,b_7,c_7\in \mathcal{T}_2.$ This shows that the code $\mathcal{C}_0$ can be lifted to a self-orthogonal code of type $\{1,0,0,0,0,0,0,0\}$ and length $4$ over $\mathscr{R}_{8,2}$ with the help of the recursive construction  method described in Figure \ref{F1}.
\end{example}
\section{Enumeration formulae for self-orthogonal and self-dual codes of
length $n$ over $\mathscr{R}_{e,m}$}\label{counting}
In this section, we will
obtain  explicit enumeration formulae for self-orthogonal and self-dual codes of an arbitrary length over $\mathscr{R}_{e,m}.$ In particular, when  $e=2,$ one can easily see that  either $\mathscr{R}_{e,m}\simeq \mathbb{F}_{2^m}[u]/\langle u^{2}\rangle$ or      $\mathscr{R}_{e,m}\simeq GR(4,m). $ The enumeration formula for  self-dual codes  over $\mathbb{F}_{2^m}[u]/\langle u^{2}\rangle$ was derived by Gaborit  \cite[Th. 3]{Zp},   while 
the enumeration formula for  self-orthogonal codes over  $\mathbb{F}_{2^m}[u]/\langle u^2\rangle $ was obtained by Galvez {\etal} \cite[Th. 2]{GBN}.   Yadav and Sharma \cite[Th. 5.1]{Galois} counted all  self-orthogonal and self-dual codes of an arbitrary length over  $GR(4,m).$  In view of this, we assume, throughout this section, that  $e \geq 3.$ We also assume that $n$ is a positive integer and $\lambda_1,\lambda_2,\ldots,\lambda_{e+1}$ are non-negative integers satisfying $n=\lambda_1+\lambda_2+\cdots+\lambda_{e+1}.$ Now, let  $\mathfrak{S}_e(n;\lambda_1,\lambda_2,\ldots,\lambda_e)$ and $\mathfrak{U}_e(n;\lambda_1,\lambda_2,\ldots,\lambda_e)$  denote  the number of distinct  self-orthogonal and self-dual codes of type $\{\lambda_1,\lambda_2,\ldots,\lambda_e\}$ and length $n$ over $\mathscr{R}_{e,m},$ respectively.  We note, by Lemma \ref{l2.2} and Remark \ref{Rem1}, that   $\mathfrak{S}_e(n;\lambda_1,\lambda_2,\ldots,\lambda_e)=0$ if $2\lambda_1+2\lambda_2+\cdots+2\lambda_{e-j+1}+\lambda_{e-j+2}+\cdots+\lambda_j > n$ holds for some integer $j$ satisfying $s+1\leq j\leq e,$  while $\mathfrak{U}_e(n;\lambda_1,\lambda_2,\ldots,\lambda_e)=0$ if $\lambda_j\neq \lambda_{e-j+2}$ for some integer $j$ satisfying $1 \leq j \leq e.$
Further, if $\mathfrak{S}_e(n)$ and $\mathfrak{U}_{e}(n)$ denote the total number of  self-orthogonal and self-dual codes of length $n$ over $\mathscr{R}_{e,m},$ respectively, 
then by Lemma \ref{l2.2} and Remark \ref{Rem1}, we get  \begin{equation}\label{summation}\mathfrak{S}_e(n)= \mathlarger{\Sigma}_{1}~\mathfrak{S}_e(n;\lambda_1,\lambda_2,\ldots,\lambda_e) \text{ ~and~ } \mathfrak{U}_e(n)= \mathlarger{\Sigma}_{2}~ \mathfrak{U}_e(n;\lambda_1,\lambda_2,\ldots,\lambda_e),\end{equation}
where the summation $\mathlarger{\Sigma}_{1}$ runs over all $e$-tuples $(\lambda_1, \lambda_2, \ldots,  \lambda_{e})$ of non-negative integers satisfying $ \lambda_1+\lambda_2+\cdots+\lambda_{e}\leq n$ and  $2\lambda_1+2\lambda_2+\cdots+2\lambda_{e-j+1}+\lambda_{e-j+2}+\cdots+\lambda_j \leq n$ for $s+1\leq j\leq e,$ while the summation $\mathlarger{\Sigma}_{2}$ runs over all $e$-tuples $(\lambda_1, \lambda_2, \ldots,  \lambda_{e})$ of non-negative integers  satisfying $\lambda_j=\lambda_{e-j+2}$ for $2 \leq j \leq e$ and $ 2(\lambda_1+\lambda_2+\cdots+\lambda_{s})+(1+\theta_e)\lambda_{s+1}=n.$
Therefore, to determine the numbers  $\mathfrak{S}_e(n)$ and $\mathfrak{U}_{e}(n),$ it suffices to obtain  enumeration formulae for the numbers $\mathfrak{S}_e(n;\lambda_1,\lambda_2,\ldots,\lambda_e)$ and $\mathfrak{U}_e(n;\lambda_1,\lambda_2,\ldots,\lambda_e).$ Further, in view of Remark \ref{Rem1}, from now on, throughout this section,  we assume  that $n$ is a positive integer and    $\lambda_1, \lambda_2, \ldots,  \lambda_{e+1}$ are  non-negative integers satisfying $n=\lambda_1+\lambda_2+\cdots+\lambda_{e+1}$ and  $2\lambda_1+2\lambda_2+\cdots+2\lambda_{e-j+1}+\lambda_{e-j+2}+\cdots+\lambda_j \leq n$ for $s+1\leq j\leq e.$ We also assume that  $\Lambda_0=0$ and $\Lambda_i=\lambda_1+\lambda_2+\cdots+\lambda_{i}$ for $1 \leq i \leq e+1.$ Now, to obtain  enumeration formulae for the numbers $\mathfrak{S}_e(n;\lambda_1,\lambda_2,\ldots,\lambda_e)$ and $\mathfrak{U}_e(n;\lambda_1,\lambda_2,\ldots,\lambda_e),$ we need to prove the following lemmas.

 In the following lemma, we assume that    $2\kappa \leq e,$ and we consider a chain $\mathcal{C}^{(1)}\subseteq \mathcal{C}^{(2)} \subseteq \cdots \subseteq \mathcal{C}^{(s+\theta_e)} $ of self-orthogonal codes of length $n$ over $\mathcal{T}_m,$ such that (i) $\dim \mathcal{C}^{(i)}=\Lambda_i$ for $1 \leq i \leq s+\theta_e,$ (ii) the code $\mathcal{C}^{(s-\kappa_1)}$ is  doubly even,  and (iii)  $\textbf{1} \notin \mathcal{C}^{(s-\kappa+\theta_e)}$ if   $n\equiv 4\pmod 8$ and $m$ is odd. Here, we will  count the    choices for a self-orthogonal code $\mathscr{D}_{e}$ of type $\{\lambda_1,\lambda_2,\ldots, \lambda_e\}$ and length $n$ over $\mathscr{R}_{e,m}$ satisfying $Tor_i(\mathscr{D}_e)=\mathcal{C}^{(i)}$ for $1 \leq i \leq s+\theta_e.$

\begin{lemma}\label{t3.1Kodd}
Let $e \geq 3$ be an integer  satisfying $2\kappa \leq e.$ 
Let  $\mathcal{C}^{(1)}\subseteq \mathcal{C}^{(2)} \subseteq \cdots \subseteq \mathcal{C}^{(s+\theta_e)} $ be a chain of self-orthogonal codes of length $n$ over $\mathcal{T}_m,$ such that (i) $\dim \mathcal{C}^{(i)}=\Lambda_{i}$  for $1 \leq i \leq s+\theta_e,$ (ii) the code $\mathcal{C}^{(s-\kappa_1)}$ is doubly even,   and (iii) $\textbf{1} \notin \mathcal{C}^{(s-\kappa+\theta_e)}$ if $n\equiv 4\pmod 8$ and $m$ is odd.  The  chain $\mathcal{C}^{(1)}\subseteq \mathcal{C}^{(2)} \subseteq \cdots \subseteq \mathcal{C}^{(s+\theta_e)} $  of  self-orthogonal codes over $\mathcal{T}_m$ 
gives rise to precisely   \begin{equation*}   2^{\epsilon}(2^m)^{\sum\limits_{i=1}^{s}\Lambda_{i}(n-\Lambda_{i+1})+\sum\limits_{j=1}^{s-1+\theta_e}\Lambda_{s+j}(n-\Lambda_{s+j+1}-\Lambda_{s+\theta_e-j})-\sum\limits_{a=1}^{s-\kappa_1-1}\Lambda_a-(1-\theta_e)\frac{\Lambda_s(\Lambda_s-1)}{2}+\mu} \hspace{-2mm}
\prod\limits_{\ell=s+1+\theta_e}^{e}{\lambda_{\ell}+n-\Lambda_{\ell}-\Lambda_{e+1-\ell}\brack \lambda_{\ell}}_{2^m} \vspace{-1mm}\end{equation*}  
 distinct   self-orthogonal codes $\mathscr{D}_e$ of type $\{\lambda_1,\lambda_2,\ldots,\lambda_e\}$ and length $n$ over $\mathscr{R}_{e,m}$ satisfying $Tor_i(\mathscr{D}_e)=\mathcal{C}^{(i)}$   for $1 \leq i \leq s+\theta_e,$  where 
\begin{equation*}(\epsilon,\mu)=\left\{\begin{array}{ll} (0,\omega) & \text{if there exists an integer }\omega \text{ satisfying }1\leq \omega \leq \kappa_1-\theta_e, \textbf{1} \in \mathcal{C}^{(s-\kappa_1-\omega)}\text{ and } \textbf{1}\notin \mathcal{C}^{(s-\kappa_1-\omega-1)};\\ (1,\kappa_1) & \text{if } \textbf{1} \in \mathcal{C}^{(s-\kappa+\theta_e)}\text{ with either }n\equiv 0\pmod 8\text{ or }n\equiv 4 \pmod 8\text{  and }m\text{ being even};\\(0,0) & \text{otherwise.} \end{array} \right.\end{equation*}
 \end{lemma}
\begin{proof} The desired result follows  by applying Propositions \ref{p3.2} and \ref{p3.4Kodd}--\ref{p3.9KoddReplace} when $e$ is even, whereas the desired result follows by applying Propositions \ref{p3.3a}--\ref{p3.9KoddReplace} when $e$ is odd. 
 \vspace{-1mm}  \end{proof}

In the following lemma, we assume that    $2\kappa > e,$ and we consider a chain $\mathcal{C}^{(1)}\subseteq \mathcal{C}^{(2)} \subseteq \cdots \subseteq \mathcal{C}^{(s+\theta_e)} $ of self-orthogonal codes of length $n$ over $\mathcal{T}_m,$  where $\dim \mathcal{C}^{(i)}=\Lambda_i$ for $1 \leq i \leq s+\theta_e$ and the code $\mathcal{C}^{(s-\kappa_1)}$ is  doubly even.
Here, we count the choices for a self-orthogonal code $\mathscr{D}_{e}$ of type $\{\lambda_1,\lambda_2,\ldots, \lambda_e\}$ and length $n$ over $\mathscr{R}_{e,m}$ satisfying $Tor_i(\mathscr{D}_e)=\mathcal{C}^{(i)}$ for $1 \leq i \leq s+\theta_e.$

\begin{lemma}\label{t3.2Kodd}
Let $e \geq 3$ be an integer  satisfying $2\kappa > e.$ 
Let  $\mathcal{C}^{(1)}\subseteq \mathcal{C}^{(2)} \subseteq \cdots \subseteq \mathcal{C}^{(s+\theta_e)} $ be a chain of self-orthogonal codes of length $n$ over $\mathcal{T}_m,$ such that (i) $\dim \mathcal{C}^{(i)}=\Lambda_{i}$  for $1 \leq i \leq s+\theta_e,$ and (ii) the code $\mathcal{C}^{(s-\kappa_1)}$ is doubly even. The chain  $\mathcal{C}^{(1)}\subseteq \mathcal{C}^{(2)} \subseteq \cdots \subseteq \mathcal{C}^{(s+\theta_e)} $   of  self-orthogonal codes over $\mathcal{T}_m$   gives rise to precisely   \begin{equation*}   (2^m)^{\sum\limits_{i=1}^{s}\Lambda_{i}(n-\Lambda_{i+1})+\sum\limits_{j=1}^{s-1+\theta_e}\Lambda_{s+j}(n-\Lambda_{s+j+1}-\Lambda_{s+\theta_e-j})-\sum\limits_{a=1}^{s-\kappa_1-1}\Lambda_a-(1-\theta_e)\frac{\Lambda_s(\Lambda_s-1)}{2}+\mu} 
\hspace{-1mm}\prod\limits_{\ell=s+1+\theta_e}^{e}{\lambda_{\ell}+n-\Lambda_{\ell}-\Lambda_{e+1-\ell}\brack \lambda_{\ell}}_{2^m} \vspace{-1mm}\end{equation*}  
 distinct   self-orthogonal codes $\mathscr{D}_e$ of type $\{\lambda_1,\lambda_2,\ldots,\lambda_e\}$ and length $n$ over $\mathscr{R}_{e,m}$ satisfying $Tor_i(\mathscr{D}_e)=\mathcal{C}^{(i)}$   for $1 \leq i \leq s+\theta_e,$  where
 \begin{equation*}
     \mu=\left\{ \begin{array}{cl}
     \omega & \text{if there exists an integer }\omega\text{ satisfying }1\leq \omega \leq s-\kappa_1-2, \textbf{1} \in \mathcal{C}^{(s-\kappa_1-\omega)}\text{ and } \textbf{1}\notin \mathcal{C}^{(s-\kappa_1-\omega-1)};  \\
     s-\kappa_1-1 & \text{if }\textbf{1}\in \mathcal{C}^{(1)};\\0 & \text{otherwise.}
 \end{array}\right.
 \end{equation*}
\end{lemma}
\begin{proof}When $e$ is even, the desired result follows by applying Propositions \ref{p3.2}, \ref{p3.4Kodd} and \ref{p3.9KoddReplace}, while in the case when $e$ is odd, the desired result follows by applying Propositions \ref{p3.3a}, \ref{p3.4Kodd}  and \ref{p3.9KoddReplace}.
    \end{proof}

\begin{remark}\label{REMM}
By Lemma \ref{LEM}, we see that there are precisely $\widehat{\sigma}_m(n;\mathfrak{d})$ distinct doubly even codes of length $n$ and dimension $\mathfrak{d}$  over $\mathcal{T}_m$ containing the all-one vector and that there are precisely $\sigma_m(n;\mathfrak{d})$ distinct doubly even codes of length $n$ and dimension $\mathfrak{d}$ over $\mathcal{T}_m$ that do not contain  the all-one vector, where the numbers $\widehat{\sigma}_m(n;\mathfrak{d})$ and $\sigma_m(n;\mathfrak{d})$  are as obtained in Theorems 3.2 and 3.3 of Yadav and Sharma \cite{Galois}, respectively. \end{remark}

  We next proceed to count the choices for the chain $\mathcal{C}^{(1)}\subseteq \mathcal{C}^{(2)} \subseteq \cdots \subseteq \mathcal{C}^{(s+\theta_e)} $ of self-orthogonal codes of length $n$  over $\mathcal{T}_m,$ where   (i) $\dim \mathcal{C}^{(i)}=\Lambda_i$ for $1 \leq i \leq s+\theta_e,$ (ii) the code $\mathcal{C}^{(s-\kappa_1)}$ is  doubly even, and (iii)  $\textbf{1} \notin \mathcal{C}^{(s-\kappa+\theta_e)}$ if $2\kappa \leq e,$  $n\equiv 4\pmod 8$ and $m$ is odd. Towards this, we note that every self-orthogonal code of length $n$ over $\mathcal{T}_{m}$ 
is contained in the set  $\mathcal{I}(\mathcal{T}_m^n)=\{\textbf{a} \in \mathcal{T}_{m}^n: B_{m}(\textbf{a},\textbf{a})=0 \},$ where  $B_{m}(\textbf{v},\textbf{w})=\pi_0(\textbf{v}\cdot \textbf{w})=\pi_0(\sum\limits_{i=1}^{n}v_iw_i)$ for all $\textbf{v}=(v_1,v_2,\ldots,v_n),$ $\textbf{w}=(w_1,w_2,\ldots,w_n) \in \mathcal{T}_{m}^n.$    Furthermore, one can easily see that the set $\mathcal{I}(\mathcal{T}_m^n)$ is an $(n-1)$-dimensional $\mathcal{T}_m$-linear subspace of $\mathcal{T}_m^n$  and 
that $\textbf{1} \in \mathcal{I}(\mathcal{T}_{m}^n)$ if and only if $n$ is even.  We will now distinguish the following two cases:   (\textbf{I})   $\textbf{1} \notin \mathcal{C}^{(s-\kappa_1)}$ and (\textbf{II}) $\textbf{1}\in\mathcal{C}^{(s-\kappa_1)}.$  

In the following lemma, we count the choices for the chain $\mathcal{C}^{(1)}\subseteq \mathcal{C}^{(2)} \subseteq \cdots \subseteq \mathcal{C}^{(s+\theta_e)} $ of self-orthogonal codes of length $n$  over $\mathcal{T}_m,$  such that (i) $\dim \mathcal{C}^{(i)}=\Lambda_i$ for $1 \leq i \leq s+\theta_e,$ and (ii) $\mathcal{C}^{(s-\kappa_1)}$ is a doubly even code satisfying $\textbf{1} \notin \mathcal{C}^{(s-\kappa_1)}.$
\begin{lemma}\label{p5.1} Let $e\geq 3$ be an  integer.  
Let  $N(\lambda_1,\lambda_2,\ldots,\lambda_{s+\theta_e}) $ denote the number of distinct choices  for the chain $\mathcal{C}^{(1)}\subseteq \mathcal{C}^{(2)} \subseteq \cdots \subseteq \mathcal{C}^{(s+\theta_e)} $ of self-orthogonal codes of length $n$  over $\mathcal{T}_m,$  such that (i) $\dim \mathcal{C}^{(i)}=\Lambda_i$ for $1 \leq i \leq s+\theta_e,$ and (ii)  $\mathcal{C}^{(s-\kappa_1)}$ is a doubly even code satisfying $\textbf{1} \notin \mathcal{C}^{(s-\kappa_1)}.$  We have

\vspace{-2mm}\begin{equation*} \vspace{-1mm}N(\lambda_1,\lambda_2,\ldots,\lambda_{s+\theta_e})= \left\{\begin{array}{l}
		\displaystyle 	\sigma_m\big(n;\Lambda_{s-\kappa_1}\big) \prod\limits_{i=1}^{s-\kappa_1}{\Lambda_i \brack \lambda_i}_{2^m}\prod\limits_{j=s-\kappa_1+1}^{s+\theta_e}{\Lambda_j-\Lambda_{s-\kappa_1} \brack \lambda_j}_{2^m}   
         \prod\limits_{\ell=\Lambda_{s-\kappa_1}}^{\Lambda_{s+\theta_e}-1}\left(\frac{2^{m(n-2\ell-1)}-1}{2^{m(\ell+1-\Lambda_{s-\kappa_1})}-1} \right)   \vspace{1mm}\\ \text{if  } n \text{ is odd};  \vspace{1mm}\\
\displaystyle  \left(\mathfrak{D}_0\left( \frac{2^{m(n-\Lambda_{s+\theta_e}-\Lambda_{s-\kappa_1})}-1}{2^{m(\Lambda_{s+\theta_e}-\Lambda_{s-\kappa_1})}-1} \right) + \mathfrak{B}_0\left( \frac{2^{m(n-2\Lambda_{s+\theta_e})}+2^{m(\Lambda_{s+\theta_e}-\Lambda_{s-\kappa_1})}-2}{2^{m(\Lambda_{s+\theta_e}-\Lambda_{s-\kappa_1})}-1} \right)\right)   \vspace{0.5mm} \\  \displaystyle  \times \prod\limits_{i=1}^{s-\kappa_1}{\Lambda_i \brack \lambda_i}_{2^m}\hspace{-1mm}\prod\limits_{j=s-\kappa_1+1}^{s+\theta_e}\hspace{-1mm}{\Lambda_j-\Lambda_{s-\kappa_1} \brack \lambda_j}_{2^m}  \prod\limits_{\ell=\Lambda_{s-\kappa_1}}^{\Lambda_{s+\theta_e}-2}\hspace{-2mm}\left(\frac{2^{m(n-2\ell-2)}-1}{2^{m(\ell+1-\Lambda_{s-\kappa_1})}-1} \right) \vspace{1mm}\\ 
  \text{if  } n \text{ is even with } \Lambda_{s+\theta_e}\neq \Lambda_{s-\kappa_1};\\
  \displaystyle  \big(\mathfrak{D}_0 + \mathfrak{B}_0\big)     \prod\limits_{i=1}^{s-\kappa_1}{\Lambda_i \brack \lambda_i}_{2^m}  
  ~~~\text{if  } n \text{ is even and  }\Lambda_{s+\theta_e}=\Lambda_{s-\kappa_1},
\end{array}\right.\end{equation*}


where  
\vspace{-1mm}\small{\begin{equation}\label{D0}\hspace{-3mm} \mathfrak{D}_0=\left\{\begin{array}{ll} 
~~~1 &\text{if } \Lambda_{s-\kappa_1}= 0;\\ 
\displaystyle \frac{(2^{m(\frac{n}{2}-1)}-1)(2^{m(\frac{n}{2}-\Lambda_{s-\kappa_1}-1)}+1)}{2^m-1}\prod\limits_{i=1}^{\Lambda_{s-\kappa_1}-1}\left(\frac{2^{m(n-2-2i)}-1}{2^{m(i+1)}-1} \right)&\text{if  }  1\leq \Lambda_{s-\kappa_1}\leq  \frac{n}{2}-1 \text{ and  } n\equiv 2,6~ (\bmod ~8);\\
\displaystyle \prod\limits_{i=0}^{\Lambda_{s-\kappa_1}-1} \left(\frac{2^{m(n-2i-3)}+2^{m(\frac{n}{2}-1-i)}-2^{m(\frac{n}{2}-i-2)}-1}{2^{m(i+1)}-1}\right)  &\text{if } 1\leq \Lambda_{s-\kappa_1}\leq  \frac{n}{2}-1 \text{ with either }n\equiv 0~(\bmod ~8)  \vspace{-3mm}\\&\text{or } n\equiv 4 ~(\bmod ~8)  \text{ and } m \text{ is even}  ;\\
\displaystyle\prod\limits_{i=0}^{\Lambda_{s-\kappa_1}-1}\left( \frac{2^{m(n-2i-3)}-2^{m(\frac{n}{2}-1-i)}+2^{m(\frac{n}{2}-i-2)}-1}{2^{m(i+1)}-1}\right)& \text{if }1\leq \Lambda_{s-\kappa_1}\leq  \frac{n}{2}-1, n\equiv 4~(\bmod ~8) \text{ and }\vspace{-3mm}\\& m \text{ is odd;}\\
~~0 & \text{if } \Lambda_{s-\kappa_1}>\frac{n}{2}-1

\end{array}\right.  \vspace{-1mm}
\end{equation} }\normalsize and 
\vspace{-2mm}\small{\begin{equation}\label{D1} \hspace{-6mm}\mathfrak{B}_{0}=\left\{\begin{array}{l} 
\displaystyle \left(2^{m(n-2\Lambda_{s-\kappa_1}-1)}-2^{m(\frac{n}{2}-\Lambda_{s-\kappa_1}-1)}\right) \prod\limits_{i=0}^{\Lambda_{s-\kappa_1}-2}\left(\frac{2^{m(n-2i-3)}+2^{m(\frac{n}{2}-1-i)}-2^{(m\frac{n}{2}-2-i)}-1}{2^{m(i+1)}-1} \right) \vspace{0.25mm}\\ \text{if } 1\leq \Lambda_{s-\kappa_1}\leq \frac{n}{2}-1 \text{ and } n\equiv 2,6~ (\bmod ~8);\vspace{0.25mm}\\
\displaystyle \left(2^{m(n-2\Lambda_{s-\kappa_1}-1)}+2^{m(\frac{n}{2}-\Lambda_{s-\kappa_1})}-2^{m(\frac{n}{2}-\Lambda_{s-\kappa_1}-1)}-1\right)\hspace{-0.5mm}\prod\limits_{i=0}^{\Lambda_{s-\kappa_1}-2} \hspace{-0.5mm}\left(\frac{2^{m(n-2i-3)}+2^{m(\frac{n}{2}-1-i)}-2^{m(\frac{n}{2}-i-2)}-1}{2^{m(i+1)}-1}\right)  \vspace{0.50mm} \\ \text{if }  1\leq \Lambda_{s-\kappa_1}\leq \frac{n}{2}-1  \text{ with either }n\equiv 0~(\bmod ~8) \text{ or } n\equiv 4~(\bmod ~8) \text{ and } m \text{ is even};\vspace{0.50mm}\\
\displaystyle  \left(2^{m(n-2\Lambda_{s-\kappa_1}-1)}-2^{m(\frac{n}{2}-\Lambda_{s-\kappa_1})}+2^{m(\frac{n}{2}-\Lambda_{s-\kappa_1}-1)}-1\right) \hspace{-0.5mm}\prod\limits_{i=0}^{\Lambda_{s-\kappa_1}-2} \hspace{-0.5mm}\left(\frac{2^{m(n-2i-3)}+2^{m(\frac{n}{2}-i-1)}-2^{m(\frac{n}{2}-i-2)}-1}{2^{m(i+1)}-1} \right) \vspace{0.5mm}\\ \text{if }  1\leq \Lambda_{s-\kappa_1}\leq \frac{n}{2}-1 ,  n\equiv 4~(\bmod ~8) \text{ and } m \text{ is odd;}\vspace{1mm}\\
~0 ~~~~~~~\text{ if either } \Lambda_{s-\kappa_1}=0 \text{ or }\Lambda_{s-\kappa_1}>\frac{n}{2}-1.
\end{array}\right.
\vspace{-2mm}\end{equation}}\normalsize
\end{lemma}
\begin{proof} To prove the result, we recall  that 
  $\textbf{1} \in \mathcal{I}(\mathcal{T}_{m}^n)$ if and only if $n$ is even.  Accordingly, we will consider the following two cases separately: (I) $n$ is odd, and (II) $n$ is even.  
\begin{enumerate}
\item[(I)] Let $n$ be odd. Here, we note that $(\mathcal{I}(\mathcal{T}_m^n),B_m(\cdot,\cdot){\restriction_{\mathcal{I}(\mathcal{T}_m^n) \times (\mathcal{I}(\mathcal{T}_m^n)}})$ 
 is a symplectic space of dimension $n-1$ over $\mathcal{T}_m.$ Further, by Remark \ref{REMM}, one can easily see that there are precisely  $\sigma_m(n;\Lambda_{s-\kappa_1})\prod\limits_{i=1}^{s-\kappa_1} {\Lambda_i \brack \lambda_i}_{2^m}$ distinct choices for  the chain  $\mathcal{C}^{(1)}\subseteq \mathcal{C}^{(2)}\subseteq \cdots \subseteq \mathcal{C}^{(s-\kappa_1)}$ of self-orthogonal codes of length $n$ over $\mathcal{T}_m,$ where $\mathcal{C}^{(s-\kappa_1)}$ is a doubly even code satisfying $\textbf{1} \notin \mathcal{C}^{(s-\kappa_1)}.$   

 Now, for a given choice of $\mathcal{C}^{(s-\kappa_1)},$ 
we will count the choices for a self-orthogonal code $\mathcal{C}^{(s+\theta_e)}$ satisfying $\mathcal{C}^{(s-\kappa_1)}\subseteq \mathcal{C}^{(s+\theta_e)} \subseteq (\mathcal{C}^{(s-\kappa_1)})^{\perp_{B_m}}.$ Towards this, we see, working as in case (I) in the proof of Theorem 3.1   of Yadav and Sharma \cite{Galois},  that $\mathcal{C}^{(s-\kappa_1)}=\langle \textbf{v}_1,\textbf{v}_2,\ldots,\textbf{v}_{\Lambda_{s-\kappa_1}}\rangle,$ where $\textbf{v}_1,\textbf{v}_2,\ldots,\textbf{v}_{\Lambda_{s-\kappa_1}} $  are mutually orthogonal singular vectors in $\mathcal{I}(\mathcal{T}_m^n)$ that are linearly independent over $\mathcal{T}_{m}.$
We further see,
by \cite[pp. 69-70]{Taylor},  that the symplectic space $(\mathcal{I}(\mathcal{T}_m^n),B_m(\cdot,\cdot){\restriction_{\mathcal{I}(\mathcal{T}_m^n) \times (\mathcal{I}(\mathcal{T}_m^n)}})$  admits an orthogonal direct sum decomposition of the form: $\mathcal{I}(\mathcal{T}_m^n)=\mathcal{X}\perp \mathcal{Y}$ with $\mathcal{X}=\langle \textbf{v}_1,\textbf{v}_1^{\prime}\rangle \perp \langle \textbf{v}_2,\textbf{v}_2^{\prime}\rangle \perp\cdots \perp \langle \textbf{v}_{\Lambda_{s-\kappa_1}},\textbf{v}_{\Lambda_{s-\kappa_1}}^{\prime}\rangle $ and $\mathcal{Y}=\langle \textbf{v}_{\Lambda_{s-\kappa_1}+1},\textbf{v}_{\Lambda_{s-\kappa_1}+1}^{\prime}\rangle \perp \langle \textbf{v}_{\Lambda_{s-\kappa_1}+2},\textbf{v}_{\Lambda_{s-\kappa_1}+2}^{\prime}\rangle \perp  \cdots \perp \langle \mathbf{v}_{\frac{n-1}{2}},\mathbf{v}_{\frac{n-1}{2}}^{\prime} \rangle ,$ where $(\textbf{v}_i,\textbf{v}_i^\prime)$ is a hyperbolic pair in $\mathcal{I}(\mathcal{T}_m^n)$ for $1 \leq  i \leq \frac{n-1}{2}.$ We further note that each $\textbf{w}\in (\mathcal{C}^{(s-\kappa_1)})^{\perp_{B_m}}$ can be uniquely written as $\textbf{w}=\sum\limits_{i=1}^{\frac{n-1}{2}}(\alpha_i\textbf{v}_i+\beta_i\textbf{v}_i^{\prime}),$ where  $\alpha_i,\beta_i \in \mathcal{T}_m$ for $1 \leq i \leq \frac{n-1}{2}.$ As $\textbf{w}\in (\mathcal{C}^{(s-\kappa_1)})^{\perp_{B_m}},$ we have $B_m(\textbf{w},\textbf{v}_j)=0$ for $1\leq  j \leq \Lambda_{s-\kappa_1},$ which implies that $\beta_j=0$ for $1 \leq j\leq \Lambda_{s-\kappa_1}.$ Thus, each $\textbf{w}\in (\mathcal{C}^{(s-\kappa_1)})^{\perp_{B_m}}$ is of the form $\textbf{w}=\sum\limits_{i=1}^{\Lambda_{s-\kappa_1}}\alpha_i\textbf{v}_i +\sum\limits_{j=\Lambda_{s-\kappa_1}+1}^{\frac{n-1}{2}}(\alpha_j\textbf{v}_j+\beta_j\textbf{v}_j^{\prime}),$ which implies that $\langle \textbf{v}_1,\textbf{v}_2,\ldots,\textbf{v}_{\Lambda_{s-\kappa_1}},\textbf{w}\rangle=\langle \textbf{v}_1,\textbf{v}_2,\ldots,\textbf{v}_{\Lambda_{s-\kappa_1}}, \sum\limits_{j=\Lambda_{s-\kappa_1}+1}^{\frac{n-1}{2}}(\alpha_j\textbf{v}_j+\beta_j\textbf{v}_j^{\prime})\rangle. $ In view of this, we assume,  without any loss of generality, that $\mathcal{C}^{(s+\theta_e)}=\langle \textbf{v}_1,\textbf{v}_2,\ldots,\textbf{v}_{\Lambda_{s-\kappa_1}}\rangle \perp \langle \textbf{w}_1,\textbf{w}_2,\ldots, \textbf{w}_{\Lambda_{s+\theta_e}-\Lambda_{s-\kappa_1}} \rangle ,$ where $\textbf{w}_1,\textbf{w}_2,\ldots, \textbf{w}_{\Lambda_{s+\theta_e}-\Lambda_{s-\kappa_1}}$ are mutually orthogonal  isotropic vectors in $\mathcal{Y}$ that are linearly independent over $\mathcal{T}_m.$ 
We further note that $\mathcal{Y}$  is a symplectic space over $\mathcal{T}_m$ of dimension $n-1-2\Lambda_{s-\kappa_1}$ and Witt index $\big(\frac{n-1}{2}\big)-\Lambda_{s-\kappa_1}.$ 
One can easily observe that the number of choices for a self-orthogonal code $\mathcal{C}^{(s+\theta_e)}$ satisfying $\mathcal{C}^{(s-\kappa_1)}\subseteq \mathcal{C}^{(s+\theta_e)} \subseteq (\mathcal{C}^{(s-\kappa_1)})^{\perp_{B_m}}$ is equal to the number of  
$(\Lambda_{s+\theta_e}-\Lambda_{s-\kappa_1})$-dimensional totally isotropic subspaces of the symplectic space $\mathcal{Y},$ which, by Exercise 8.1 (ii) of \cite{Taylor}, has precisely  $ \prod\limits_{i=\Lambda_{s-\kappa_1}}^{\Lambda_{s+\theta_e}-1}\left(\frac{2^{m(n-2i-1)}-1}{2^{m(i+1-\Lambda_{s-\kappa_1})}-1}\right)$  distinct choices.
Finally, for a given choice of $\mathcal{C}^{(s-\kappa_1)},$ we note that the chain  $\mathcal{C}^{(s-\kappa_1+1)}\subseteq \mathcal{C}^{(s-\kappa_1+2)}\subseteq \cdots\subseteq \mathcal{C}^{(s+\theta_e)}$ of codes has precisely   \vspace{-1mm}\begin{equation*}\vspace{-1mm}
   \prod\limits_{i=\Lambda_{s-\kappa_1}}^{\Lambda_{s+\theta_e}-1}\left(\frac{2^{m(n-2i-1)}-1}{2^{m(i+1-\Lambda_{s-\kappa_1})}-1}\right)\prod\limits_{j=s-\kappa_1+1}^{s+\theta_e}{\Lambda_j-\Lambda_{s-\kappa_1} \brack \lambda_j}_{2^m} 
\end{equation*}  distinct choices. From this, the desired result follows.

\item[(II)] Next, let $n$ be even.  In this case, we have  $\textbf{1}\in \mathcal{I}(\mathcal{T}_{m}^n).$ Let us choose  an $(n-2)$-dimensional $\mathcal{T}_{m}$-linear subspace  $\mathcal{U}_{m}^{\prime}$ of $\mathcal{I}(\mathcal{T}_{m}^n)$ satisfying $\textbf{1} \notin \mathcal{U}_{m}^{\prime}.$ It is easy to see that $\mathcal{I}(\mathcal{T}_{m}^n)=\mathcal{U}_{m}^{\prime} \perp \langle \textbf{1} \rangle.$ Note that $(\mathcal{U}_{m}^{\prime}, B_{m}(\cdot,\cdot){\restriction_{\mathcal{U}_{m}^{\prime}\times \mathcal{U}_{m}^{\prime}}})$ is an  $(n-2)$-dimensional symplectic space  over $\mathcal{T}_{m}.$

Now, we will count the choices for the chain $\mathcal{C}^{(1)}\subseteq\mathcal{C}^{(2)}\subseteq \cdots \subseteq \mathcal{C}^{(s+\theta_e)}$ of self-orthogonal codes of length $n$ over $\mathcal{T}_m,$  where  $\mathcal{C}^{(s-\kappa_1)}$ is a doubly even code satisfying $\textbf{1} \notin \mathcal{C}^{(s-\kappa_1)}.$ 

Towards this, we first note that $N(\lambda_1,\lambda_2,\ldots,\lambda_{s+\theta_e})=1$ if $\Lambda_{s+\theta_e}=0.$ So  we assume, throughout the proof,  that $\Lambda_{s+\theta_e}\geq 1.$ 

When $\Lambda_{s-\kappa_1}=0,$ working as in case (II) in the proof of 
Proposition 3.1 of Sharma and Kaur \cite{Sharma},  we get
\vspace{-2mm}\begin{equation*}\vspace{-2mm}
N(\lambda_1,\lambda_2,\ldots,\lambda_{s+\theta_e})=\prod\limits_{i=0}^{\Lambda_{s+\theta_e}-2}\left(\frac{2^{m(n-2i-2)}-1}{2^{m(i+1)}-1}\right) \vspace{1mm} \left(\frac{2^{m(n-\Lambda_{s+\theta_e})}-1}{2^{m\Lambda_{s+\theta_e}}-1}\right)  \prod\limits_{j=s-\kappa_1+1}^{s+\theta_e}{\Lambda_j \brack \lambda_j}_{2^m}.\end{equation*}  

On the other hand,  when $\Lambda_{s-\kappa_1}\geq 1,$ working as in case (II) in the proof of Theorem 3.1 of Yadav and Sharma \cite{Galois}, we see that the code $\mathcal{C}^{(s-\kappa_1)}$ is  of the following two forms: \begin{itemize}
    \item $\langle \textbf{v}_1,\textbf{v}_2,\ldots,\textbf{v}_{\Lambda_{s-\kappa_1}} \rangle, $ where
$\textbf{v}_1,\textbf{v}_2,\ldots, \textbf{v}_{\Lambda_{s-\kappa_1}} \in \mathcal{U}_{m}^{\prime}$ are  mutually orthogonal singular vectors that are linearly independent
 over  $\mathcal{T}_{m}.$   
\item $\langle \textbf{v}_1,\textbf{v}_2,\ldots,\textbf{v}_{\Lambda_{s-\kappa_1}-1},\textbf{1}+\textbf{v}_{\Lambda_{s-\kappa_1}} \rangle, $ where
$\textbf{v}_1,\textbf{v}_2,\ldots, \textbf{v}_{\Lambda_{s-\kappa_1}} \in \mathcal{U}_{m}^{\prime}\setminus\{\textbf{0}\}$ are such that   $\textbf{v}_1,\textbf{v}_2,\ldots,\textbf{v}_{\Lambda_{s-\kappa_1}-1},$ $\textbf{1}+\textbf{v}_{\Lambda_{s-\kappa_1}}$ are mutually orthogonal singular vectors that are linearly independent  over $\mathcal{T}_{m}.$ 
\end{itemize}
Accordingly, we will consider the following two cases:
\begin{enumerate}
    \item[(i)] First of all, suppose that the code $\mathcal{C}^{(s-\kappa_1)}$ is of the form $\langle \textbf{v}_1,\textbf{v}_2,\ldots,\textbf{v}_{\Lambda_{s-\kappa_1}} \rangle,$ where
$\textbf{v}_1,\textbf{v}_2,\ldots, \textbf{v}_{\Lambda_{s-\kappa_1}} \in \mathcal{U}_{m}^{\prime}$ are mutually orthogonal singular vectors that are linearly independent 
 over  $\mathcal{T}_{m}.$ Using Lemma \ref{LEM} and working as in case (II) in the proof of Theorem 3.1 of Yadav and Sharma \cite{Galois}, we see that such a code $\mathcal{C}^{(s-\kappa_1)}$ has precisely $\mathfrak{D}_0$ distinct choices, where the number $\mathfrak{D}_{0}$ is given by \eqref{D0}. Given such a choice of $\mathcal{C}^{(s-\kappa_1)},$  we will  count the choices for the chain 
$\mathcal{C}^{(1)}\subseteq\mathcal{C}^{(2)}\subseteq \cdots \subseteq \mathcal{C}^{(s+\theta_e)}$ of self-orthogonal codes of length $n$ over $\mathcal{T}_m.$  

Towards this,  one can easily see that  there are precisely $\mathfrak{D}_0\prod\limits_{j=1}^{s-\kappa_1}{\Lambda_j \brack \lambda_j}_{2^m}$ distinct choices for the chain  
$\mathcal{C}^{(1)}\subseteq \mathcal{C}^{(2)}\subseteq \cdots \subseteq \mathcal{C}^{(s-\kappa_1)}$ of self-orthogonal codes of length $n$ over $\mathcal{T}_m.$

Now, if $\Lambda_{s+\theta_e}=\Lambda_{s-\kappa_1},$ it is easy to see that the chain $\mathcal{C}^{(1)}\subseteq \mathcal{C}^{(2)}\subseteq \cdots \subseteq \mathcal{C}^{(s+\theta_e)}$ 
 has  precisely $\mathfrak{D}_0\prod\limits_{j=1}^{s-\kappa_1}{\Lambda_j \brack \lambda_j}_{2^m}$ distinct choices.
 
Next, if $\Lambda_{s+\theta_e}\neq \Lambda_{s-\kappa_1},$ for a given choice of $ \mathcal{C}^{(s-\kappa_1)},$   we need to count the choices for a self-orthogonal code $\mathcal{C}^{(s+\theta_e)}$ satisfying  $\mathcal{C}^{(s-\kappa_1)}\subseteq \mathcal{C}^{(s+\theta_e)} \subseteq (\mathcal{C}^{(s-\kappa_1)})^{\perp_{B_m}}.$ Here, we see, by
\cite[pp. 69-70]{Taylor}, that the space $(\mathcal{I}(\mathcal{T}_m^n),B_m(\cdot,\cdot){\restriction_{\mathcal{I}(\mathcal{T}_m^n) \times (\mathcal{I}(\mathcal{T}_m^n)}})$  admits an orthogonal direct sum decomposition of the form: $\mathcal{I}(\mathcal{T}_m^n)=\mathcal{X}_0\perp \mathcal{Y}_0\perp \langle \textbf{1}\rangle$ with  $\mathcal{X}_0=\langle \textbf{v}_1,\textbf{v}_1^{\prime}\rangle \perp \langle \textbf{v}_2,\textbf{v}_2^{\prime}\rangle \perp\cdots \perp \langle \textbf{v}_{\Lambda_{s-\kappa_1}},\textbf{v}_{\Lambda_{s-\kappa_1}}^{\prime}\rangle $ and $\mathcal{Y}_0=\langle \textbf{v}_{\Lambda_{s-\kappa_1}+1},\textbf{v}_{\Lambda_{s-\kappa_1}+1}^{\prime}\rangle \perp \langle \textbf{v}_{\Lambda_{s-\kappa_1}+2},\textbf{v}_{\Lambda_{s-\kappa_1}+2}^{\prime}\rangle \perp  \cdots \perp \langle \mathbf{v}_{\frac{n-2}{2}},\mathbf{v}_{\frac{n-2}{2}}^{\prime} \rangle ,$ where $(\textbf{v}_i,\textbf{v}_i^\prime)$ is a hyperbolic pair in $\mathcal{I}(\mathcal{T}_m^n)$ for $1 \leq  i \leq \frac{n-2}{2}.$
We further note that each $\textbf{z}\in (\mathcal{C}^{(s-\kappa_1)})^{\perp_{B_m}}$ can be uniquely written as $\textbf{z}=\sum\limits_{i=1}^{\frac{n-2}{2}}(\alpha_i\textbf{v}_i+\beta_i\textbf{v}_i^{\prime})+\gamma \textbf{1},$ where  $\alpha_i,\beta_i,\gamma \in \mathcal{T}_m$ for $1 \leq i \leq \frac{n-2}{2}.$ Since $\textbf{z}\in (\mathcal{C}^{(s-\kappa_1)})^{\perp_{B_m}},$ we have $B_m(\textbf{z},\textbf{v}_j)=0$ for $1\leq  j \leq \Lambda_{s-\kappa_1},$ which implies that $\beta_j=0$ for $1 \leq j\leq \Lambda_{s-\kappa_1}.$ Thus, each $\textbf{z}\in (\mathcal{C}^{(s-\kappa_1)})^{\perp_{B_m}}$ is of the form  $\textbf{z}=\sum\limits_{i=1}^{\Lambda_{s-\kappa_1}}\alpha_i\textbf{v}_i +\sum\limits_{j=\Lambda_{s-\kappa_1}+1}^{\frac{n-2}{2}}(\alpha_j\textbf{v}_j+\beta_j\textbf{v}_j^{\prime})+\gamma \textbf{1},$ which implies that  $\langle \textbf{v}_1,\textbf{v}_2,\ldots,\textbf{v}_{\Lambda_{s-\kappa_1}},\textbf{z}\rangle=\langle \textbf{v}_1,\textbf{v}_2,\ldots,\textbf{v}_{\Lambda_{s-\kappa_1}}, \sum\limits_{j=\Lambda_{s-\kappa_1}+1}^{\frac{n-2}{2}}(\alpha_j\textbf{v}_j+\beta_j\textbf{v}_j^{\prime})+\gamma \textbf{1}\rangle. $ In view of this, we assume,  without any loss of generality,  that the code $\mathcal{C}^{(s+\theta_e)}$  is  of the following two forms: 
\begin{itemize}
    \item[($\star$)] 
$\langle \textbf{v}_1,\textbf{v}_2,\ldots,\textbf{v}_{\Lambda_{s-\kappa_1}}\rangle \perp \langle \textbf{w}_1,\textbf{w}_2,\ldots, \textbf{w}_{\Lambda_{s+\theta_e}-\Lambda_{s-\kappa_1}} \rangle ,$ where $\textbf{w}_1,\textbf{w}_2,\ldots, \textbf{w}_{\Lambda_{s+\theta_e}-\Lambda_{s-\kappa_1}}$ are mutually orthogonal isotropic vectors in $\mathcal{Y}_0$ that are linearly independent over $\mathcal{T}_m.$  
\item[($\dagger$)] $\langle \textbf{v}_1,\textbf{v}_2,\ldots,\textbf{v}_{\Lambda_{s-\kappa_1}}\rangle \perp \langle \textbf{w}_1,\textbf{w}_2,\ldots, \textbf{w}_{\Lambda_{s+\theta_e}-\Lambda_{s-\kappa_1}-1},\textbf{1}+\textbf{w}_{\Lambda_{s+\theta_e}-\Lambda_{s-\kappa_1}} \rangle ,$ where $\textbf{w}_1,\textbf{w}_2,\ldots,$\\ 
$ \textbf{w}_{\Lambda_{s+\theta_e}-\Lambda_{s-\kappa_1}-1}$ are mutually orthogonal isotropic vectors in $\mathcal{Y}_0,$   $\textbf{w}_{\Lambda_{s+\theta_e}-\Lambda_{s-\kappa_1}}\in \mathcal{Y}_0,$ and the vectors $\textbf{w}_1,\textbf{w}_2,\ldots, \textbf{w}_{\Lambda_{s+\theta_e}-\Lambda_{s-\kappa_1}-1},\textbf{1}+\textbf{w}_{\Lambda_{s+\theta_e}-\Lambda_{s-\kappa_1}}$ are linearly independent over $\mathcal{T}_m.$ \end{itemize}
Here, working as in case (II) in the proof of 
Proposition 3.1 of Sharma and Kaur \cite{Sharma} and by  Exercise 8.1 (ii) of \cite{Taylor},  we observe that the subspaces of the forms ($\star$) and ($\dagger$), and hence the code $\mathcal{C}^{(s+\theta_e)}$ has  precisely 
 \vspace{-1mm}
\begin{equation*}
 \vspace{-1mm} \prod\limits_{i=\Lambda_{s-\kappa_1}}^{\Lambda_{s+\theta_e}-2}\left(\frac{2^{m(n-2i-2)}-1}{2^{m(i+1-\Lambda_{s-\kappa_1})}-1}\right) \vspace{1mm} \left(\frac{2^{m(n-\Lambda_{s+\theta_e}-\Lambda_{s-\kappa_1})}-1}{2^{m(\Lambda_{s+\theta_e}-\Lambda_{s-\kappa_1})}-1}\right)  
\end{equation*}  distinct choices. Furthermore, for a given choice of $\mathcal{C}^{(s-\kappa_1)},$  it is easy to see that the desired chain  $\mathcal{C}^{(s-\kappa_1+1)}\subseteq \mathcal{C}^{(s-\kappa_1+2)}\subseteq  \cdots\subseteq  \ \mathcal{C}^{(s+\theta_e)}$ of codes has precisely
\vspace{-1mm}\begin{equation*}\vspace{-1mm}
\prod\limits_{i=\Lambda_{s-\kappa_1}}^{\Lambda_{s+\theta_e}-2}\left(\frac{2^{m(n-2i-2)}-1}{2^{m(i+1-\Lambda_{s-\kappa_1})}-1}\right) \vspace{1mm} \left(\frac{2^{m(n-\Lambda_{s+\theta_e}-\Lambda_{s-\kappa_1})}-1}{2^{m(\Lambda_{s+\theta_e}-\Lambda_{s-\kappa_1})}-1}\right)  \prod\limits_{j=s-\kappa_1+1}^{s+\theta_e}{\Lambda_j-\Lambda_{s-\kappa_1} \brack \lambda_j}_{2^m}\end{equation*}  distinct choices.

\item[(ii)] Let us suppose that the code $\mathcal{C}^{(s-\kappa_1)}$ is of the form 
 $\langle \textbf{v}_1,\textbf{v}_2,\ldots,\textbf{v}_{\Lambda_{s-\kappa_1}-1},\textbf{1}+\textbf{v}_{\Lambda_{s-\kappa_1}} \rangle$ with
$\textbf{v}_1,\textbf{v}_2,\ldots,$ $ \textbf{v}_{\Lambda_{s-\kappa_1}} \in \mathcal{U}_{m}^{\prime}\setminus\{\textbf{0}\},$ where $\textbf{v}_1,\textbf{v}_2,\ldots,\textbf{v}_{\Lambda_{s-\kappa_1}-1},\textbf{1}+\textbf{v}_{\Lambda_{s-\kappa_1}}$ are  mutually orthogonal singular vectors that are linearly independent over $\mathcal{T}_{m}.$ Using Lemma \ref{LEM} and working as in case (II) in the proof of  Theorem 3.1 of Yadav and Sharma \cite{Galois}, we see that such a code $\mathcal{C}^{(s-\kappa_1)}$ has precisely $\mathfrak{B}_{0}$ distinct choices, where the number  $\mathfrak{B}_{0}$ is given by \eqref{D1}. Given such a choice of $\mathcal{C}^{(s-\kappa_1)}$,  we will  count the choices for the chain 
$\mathcal{C}^{(1)}\subseteq\mathcal{C}^{(2)}\subseteq \cdots \subseteq \mathcal{C}^{(s+\theta_e)}$ of self-orthogonal codes of length $n$ over $\mathcal{T}_m.$   Here,  one can easily see that  there are precisely $\mathfrak{B}_{0}\prod\limits_{j=1}^{s-\kappa_1}{\Lambda_j \brack \lambda_j}_{2^m}$ distinct choices for the chain  
$\mathcal{C}^{(1)}\subseteq \mathcal{C}^{(2)}\subseteq \cdots \subseteq \mathcal{C}^{(s-\kappa_1)}$ of self-orthogonal codes of length $n$ over $\mathcal{T}_m.$

Now, if $\Lambda_{s+\theta_e}=\Lambda_{s-\kappa_1},$ it is easy to see that the chain $\mathcal{C}^{(1)}\subseteq \mathcal{C}^{(2)}\subseteq \cdots \subseteq \mathcal{C}^{(s+\theta_e)}$ 
 has  precisely $\mathfrak{B}_0\prod\limits_{j=1}^{s-\kappa_1}{\Lambda_j \brack \lambda_j}_{2^m}$ distinct choices.
 
Next, if $\Lambda_{s+\theta_e}\neq \Lambda_{s-\kappa_1},$   for a given choice of $ \mathcal{C}^{(s-\kappa_1)},$   we need to count the choices for a self-orthogonal code $\mathcal{C}^{(s+\theta_e)}$ satisfying  $\mathcal{C}^{(s-\kappa_1)}\subseteq 
   \mathcal{C}^{(s+\theta_e)} \subseteq (\mathcal{C}^{(s-\kappa_1)})^{\perp_{B_m}}.$ 
 Here, working as in case (i), we observe that the code $\mathcal{C}^{(s+\theta_e)}$ 
is either of the form $\langle \textbf{v}_1,\textbf{v}_2,\ldots,\textbf{v}_{\Lambda_{s-\kappa_1}-1},\textbf{1}+\textbf{v}_{\Lambda_{s-\kappa_1}}\rangle \perp \langle \textbf{w}_1,\textbf{w}_2,\ldots, \textbf{w}_{\Lambda_{s+\theta_e}-\Lambda_{s-\kappa_1}} \rangle ,$ or of the form  $\langle \textbf{v}_1,\textbf{v}_2,\ldots,\textbf{v}_{\Lambda_{s-\kappa_1}-1},\textbf{1}+\textbf{v}_{\Lambda_{s-\kappa_1}}\rangle \perp \langle \textbf{w}_1,\textbf{w}_2,\ldots, \textbf{w}_{\Lambda_{s+\theta_e}-\Lambda_{s-\kappa_1}-1},\textbf{1} \rangle ,$ where $\textbf{w}_1,\textbf{w}_2,\ldots,$ $ \textbf{w}_{\Lambda_{s+\theta_e}-\Lambda_{s-\kappa_1}}\in \mathcal{Y}_0$ are mutually orthogonal, isotropic and linearly independent vectors  over $\mathcal{T}_m.$  Again, working as in case (II) in the proof of
Proposition 3.1 of Sharma and Kaur \cite{Sharma} and by  Exercise 8.1 (ii) of \cite{Taylor},  we see that such a subspace, and hence the code $\mathcal{C}^{(s+\theta_e)}$ has  precisely \vspace{-2mm}
\begin{equation*}
 \vspace{-2mm} \left(\frac{2^{m(n-2\Lambda_{s+\theta_e})}+2^{m(\Lambda_{s+\theta_e}-\Lambda_{s-\kappa_1})}-2}{2^{m(\Lambda_{s+\theta_e}-\Lambda_{s-\kappa_1})}-1} \right)\prod\limits_{i=\Lambda_{s-\kappa_1}}^{\Lambda_{s+\theta_e}-2}\left(\frac{2^{m(n-2i-2)}-1}{2^{m(i+1-\Lambda_{s-\kappa_1})}-1}\right)  
\end{equation*}  distinct choices. 
Furthermore, for a given choice of $\mathcal{C}^{(s-\kappa_1)},$  it is easy to see that the desired chain  $\mathcal{C}^{(s-\kappa_1+1)}\subseteq \mathcal{C}^{(s-\kappa_1+2)}\subseteq  \cdots\subseteq  \ \mathcal{C}^{(s+\theta_e)}$ of codes has precisely
\vspace{-2mm}\begin{equation*}\vspace{-2mm}
\left(\frac{2^{m(n-2\Lambda_{s+\theta_e})}+2^{m(\Lambda_{s+\theta_e}-\Lambda_{s-\kappa_1})}-2}{2^{m(\Lambda_{s+\theta_e}-\Lambda_{s-\kappa_1})}-1} \right)\prod\limits_{i=\Lambda_{s-\kappa_1}}^{\Lambda_{s+\theta_e}-2}\left(\frac{2^{m(n-2i-2)}-1}{2^{m(i+1-\Lambda_{s-\kappa_1})}-1}\right)   \prod\limits_{j=s-\kappa_1+1}^{s+\theta_e}{\Lambda_j-\Lambda_{s-\kappa_1} \brack \lambda_j}_{2^m}\end{equation*}  distinct choices.
\end{enumerate}   On combining the cases (i) and (ii) above, the desired result follows in the case when $n$ is even. 
\vspace{-2mm}\end{enumerate}
\vspace{-4mm}\end{proof}

We next see, by Theorem 3.2 of Yadav and Sharma \cite{Galois}, that if a doubly even code of length $n$ over $\mathcal{T}_m$ contains $\textbf{1},$ then  $n\equiv 0,4\pmod 8.$ We will now consider the case $n\equiv 0,4 \pmod 8$ and count all distinct choices for the chain $\mathcal{C}^{(1)}\subseteq \mathcal{C}^{(2)} \subseteq \cdots \subseteq \mathcal{C}^{(s+\theta_e)} $ of self-orthogonal codes of length $n$  over $\mathcal{T}_m,$  where (i) $\dim \mathcal{C}^{(i)}=\Lambda_i$ for $1 \leq i \leq s+\theta_e,$ and (ii)  $\mathcal{C}^{(s-\kappa_1)}$ is a  doubly even code satisfying $\textbf{1} \in \mathcal{C}^{(s-\kappa_1)}.$  
Here, the following three cases arise: (\textbf{I}) there exists an integer $\omega$ satisfying $0\leq \omega \leq \kappa_1-\theta_e$ if $2\kappa \leq e,$ while $0 \leq \omega \leq s-\kappa_1-2$ if $2\kappa > e,$ $\textbf{1}\in \mathcal{C}^{(s-\kappa_1-\omega)}$ and $\textbf{1}\notin \mathcal{C}^{(s-\kappa_1-\omega-1)},$  (\textbf{II}) $\textbf{1}\in  \mathcal{C}^{(s-\kappa+\theta_e)}$ and  $2\kappa \leq e$ with  either $n\equiv 0\pmod 8$  or $n\equiv 4 \pmod 8$ and $m$ being even, and (\textbf{III}) $\textbf{1}\in\mathcal{C}^{(1)}$ and $2\kappa>e.$
In the following lemma, we   count the choices for the chain $\mathcal{C}^{(1)}\subseteq \mathcal{C}^{(2)} \subseteq \cdots \subseteq \mathcal{C}^{(s+\theta_e)} $ of self-orthogonal codes of length $n$  over $\mathcal{T}_m$ corresponding to case (\textbf{I}).
\begin{lemma}\label{p5.2}
    Let $e\geq 3$ be an  integer.  
    Let $\omega$ be a fixed integer satisfying  $0 \leq \omega \leq \kappa_1-\theta_e$  if $2\kappa \leq e,$ while $0 \leq \omega \leq s-\kappa_1-2$ if $2\kappa > e.$  Let  $Y_\omega(\lambda_1,\lambda_2,\ldots,\lambda_{s+\theta_e}) $ denote the number of distinct choices  for the chain $\mathcal{C}^{(1)}\subseteq \mathcal{C}^{(2)} \subseteq \cdots \subseteq \mathcal{C}^{(s+\theta_e)} $ of self-orthogonal codes of length $n$  over $\mathcal{T}_m,$  where (i) $\dim \mathcal{C}^{(i)}=\Lambda_i$ for $1 \leq i \leq s+\theta_e,$ (ii) the code $\mathcal{C}^{(s-\kappa_1)}$ is  doubly even, and  (iii)   $\textbf{1} \in \mathcal{C}^{(s-\kappa_1-\omega)}$ and $\textbf{1} \notin \mathcal{C}^{(s-\kappa_1-\omega-1)}.$ We have  $n\equiv 0,4\pmod 8,$  and $Y_\omega(\lambda_1,\lambda_2,\ldots,\lambda_{s+\theta_e})=0$ if $\Lambda_{s-\kappa_1-\omega}=0.$ Further, when $\Lambda_{s-\kappa_1-\omega}\geq 1,$ we have 
\begin{eqnarray*}
Y_\omega(\lambda_1,\lambda_2,\ldots,\lambda_{s+\theta_e})&=& \widehat{\sigma}_m(n;\Lambda_{s-\kappa_1})  (2^m)^{\Lambda_{s-\kappa_1-\omega-1}}   {\Lambda_{s-\kappa_1}-1 \brack \Lambda_{s-\kappa_1}-\Lambda_{s-\kappa_1-\omega}}_{2^m}{\Lambda_{s-\kappa_1-\omega}-1 \brack \Lambda_{s-\kappa_1-\omega-1}}_{2^m}  \prod\limits_{i=1}^{s-\kappa_1-\omega-1}{\Lambda_i \brack \lambda_i}_{2^m}\\ &&\times \hspace{-1mm}\prod\limits_{a=s-\kappa_1-\omega+1}^{s-\kappa_1}{\Lambda_a-\Lambda_{s-\kappa_1-\omega} \brack  \lambda_a}_{2^m}\prod\limits_{b=s-\kappa_1+1}^{s+\theta_e}\hspace{-1mm}{\Lambda_b-\Lambda_{s-\kappa_1} \brack \lambda_b}_{2^m} \prod\limits_{g=\Lambda_{s-\kappa_1}}^{\Lambda_{s+\theta_e}-1}\left( \frac{2^{m(n-2g)}-1}{2^{m(g+1-\Lambda_{s-\kappa_1})}-1}\right)\end{eqnarray*} 
\end{lemma}
 \begin{proof} To prove the result, we first note that $\textbf{1} \in \mathcal{C}^{(s-\kappa_1-\omega)},$ which, by Theorem 3.2 of Yadav and Sharma \cite{Galois}, implies that   $n\equiv 0,4\pmod 8.$
    We further observe, by Remark \ref{REMM}, that the chain  $\mathcal{C}^{(1)}\subseteq  \mathcal{C}^{(2)}\subseteq \cdots\subseteq \mathcal{C}^{(s-\kappa_1-\omega)}$ of self-orthogonal codes of length $n$ over $\mathcal{T}_m$ has precisely \vspace{-1mm}\begin{equation*}
  \vspace{-1mm}  \widehat{\sigma}_m(n; \Lambda_{s-\kappa_1-\omega}) \left ( {\Lambda_{s-\kappa_1-\omega} \brack \Lambda_{s-\kappa_1-\omega-1}}_{2^m} -{\Lambda_{s-\kappa_1-\omega}-1 \brack \Lambda_{s-\kappa_1-\omega-1}-1}_{2^m}\right)\prod\limits_{i=1}^{s-\kappa_1-\omega-1} {\Lambda_i \brack \lambda_i}_{2^m}\end{equation*}
distinct choices. Next, by the Pascal's Identity for Gaussian binomial coefficients, we note that  \vspace{-1mm}\begin{equation*}\vspace{-1mm}
    {\Lambda_{s-\kappa_1-\omega} \brack \Lambda_{s-\kappa_1-\omega-1}}_{2^m} -{\Lambda_{s-\kappa_1-\omega}-1 \brack \Lambda_{s-\kappa_1-\omega-1}-1}_{2^m}= (2^m)^{\Lambda_{s-\kappa_1-\omega-1}}{\Lambda_{s-\kappa_1-\omega}-1 \brack \Lambda_{s-\kappa_1-\omega-1}}_{2^m}.\end{equation*}  This implies that the chain  $\mathcal{C}^{(1)}\subseteq  \mathcal{C}^{(2)}\subseteq \cdots\subseteq \mathcal{C}^{(s-\kappa_1-\omega)}$ of codes has precisely \vspace{-1mm}\begin{equation*}
  \vspace{-1mm}  \widehat{\sigma}_m(n; \Lambda_{s-\kappa_1-\omega}) (2^m)^{\Lambda_{s-\kappa_1-\omega-1}}{\Lambda_{s-\kappa_1-\omega}-1 \brack \Lambda_{s-\kappa_1-\omega-1}}_{2^m}\prod\limits_{i=1}^{s-\kappa_1-\omega-1} {\Lambda_i \brack \lambda_i}_{2^m}\end{equation*}
distinct choices.
Now, for a given choice of a doubly even code $ \mathcal{C}^{(s-\kappa_1-\omega)}$ containing $\textbf{1},$  we need to count the choices for a doubly even code $\mathcal{C}^{(s-\kappa_1)}$ satisfying $ \mathcal{C}^{(s-\kappa_1-\omega)} \subseteq \mathcal{C}^{(s-\kappa_1)}\subseteq  (\mathcal{C}^{(s-\kappa_1-\omega)})^{\perp_{B_m}}.$ 
Towards this, we see, by the proof of Theorem 3.2  of Yadav and Sharma \cite{Galois},   that $\mathcal{C}^{(s-\kappa_1-\omega)}=\langle \textbf{v}_1,\textbf{v}_2,\ldots,\textbf{v}_{\Lambda_{s-\kappa_1-\omega}-1},\textbf{1}\rangle,$ where $\textbf{v}_1,\textbf{v}_2,\ldots,\textbf{v}_{\Lambda_{s-\kappa_1-\omega}-1} $  are mutually orthogonal singular vectors in $\mathcal{I}(\mathcal{T}_m^n)$ that are linearly independent over $\mathcal{T}_{m}.$
 Now, working as in case (II) in the proof of Lemma \ref{p5.1} and  by applying   Theorem  3.1 of Yadav and Sharma \cite{Galois}, we see that for a given choice of $\mathcal{C}^{(s-\kappa_1-\omega)},$    the chain  $\mathcal{C}^{(s-\kappa_1-\omega+1)}\subseteq \mathcal{C}^{(s-\kappa_1-\omega+2)}\subseteq \cdots\subseteq  \mathcal{C}^{(s-\kappa_1)}$ of codes has precisely 
\vspace{-1mm}\begin{equation*}\vspace{-1mm}
 \prod\limits_{i=\Lambda_{s-\kappa_1-\omega}}^{\Lambda_{s-\kappa_1}-1} \Big(  \frac{2^{m(n-2i-1)}+\epsilon_12^{m(\frac{n}{2}-i)} -\epsilon_1 2^{m(\frac{n}{2}-i-1)}-1 }{2^{m(i-\Lambda_{s-\kappa_1-\omega}+1)}-1}\Big)  \prod\limits_{t=s-\kappa_1-\omega+1}^{s-\kappa_1} {\Lambda_{t}-\Lambda_{s-\kappa_1-\omega} \brack \lambda_t}_{2^m}
 \end{equation*} 
distinct choices, where $\epsilon_1=1 $ if either $n\equiv 0\pmod8$ or $n\equiv 4\pmod 8$ and $m$ is even, while $\epsilon_1=-1$ if $n\equiv 4\pmod 8$ and $m$ is odd.
Furthermore, one can easily  observe that   \vspace{-1mm}\begin{equation*}
   \vspace{-1mm} \widehat{\sigma}_m(n;\Lambda_{s-\kappa_1}){\Lambda_{s-\kappa_1}-1 \brack \Lambda_{s-\kappa_1}-\Lambda_{s-\kappa_1-\omega}}_{2^m} \hspace{-1mm}=\widehat{\sigma}_m(n; \Lambda_{s-\kappa_1-\omega}) \hspace{-1.5mm}\prod\limits_{i=\Lambda_{s-\kappa_1-\omega}}^{\Lambda_{s-\kappa_1}-1} \hspace{-1.5mm}\Big(  \frac{2^{m(n-2i-1)}+\epsilon_1 2^{m(\frac{n}{2}-i)} - \epsilon_12^{m(\frac{n}{2}-i-1)}-1 }{2^{m(i-\Lambda_{s-\kappa_1-\omega}+1)}-1}\Big).\end{equation*}   Finally, working as in case (II) in the proof of Lemma \ref{p5.1},  we see that  for a given choice of $\mathcal{C}^{(s-\kappa_1)},$  the chain  $\mathcal{C}^{(s-\kappa_1+1)}\subseteq  \mathcal{C}^{(s-\kappa_1+2)}\subseteq \cdots\subseteq  \mathcal{C}^{(s+\theta_e)}$  of  codes has precisely 
\begin{equation*}
     \prod\limits_{j=s-\kappa_1+1}^{s+\theta_e}\hspace{-1mm}{\Lambda_j-\Lambda_{s-\kappa_1} \brack \lambda_j}_{2^m}\prod\limits_{i=\Lambda_{s-\kappa_1}}^{\Lambda_{s+\theta_e}-1}\hspace{-1mm}\Big(\frac{2^{m(n-2i)}-1}{2^{m(i+1-\Lambda_{s-\kappa_1})}-1}\Big) \end{equation*} distinct choices.  From this, the desired result follows immediately. 
\vspace{-1mm}  \end{proof}

 In the following lemma, we consider the case $2\kappa \leq e$ and  count the choices for the chain $\mathcal{C}^{(1)}\subseteq \mathcal{C}^{(2)} \subseteq \cdots \subseteq \mathcal{C}^{(s+\theta_e)} $ of self-orthogonal codes of length $n$  over $\mathcal{T}_m, $ such that  (i) $\dim \mathcal{C}^{(i)}=\Lambda_i$ for $1 \leq i \leq s+\theta_e,$ (ii) the code $\mathcal{C}^{(s-\kappa_1)}$ is  doubly even, and  (iii) $\textbf{1}\in  \mathcal{C}^{(s-\kappa+\theta_e)}$ with either $n\equiv 0\pmod 8$  or $n\equiv 4 \pmod 8$ and $m$ being even. 
\begin{lemma}\label{p5.3}
    Let $e\geq 3$ be an  integer satisfying $2\kappa \leq e.$  
    Let  $M(\lambda_1,\lambda_2,\ldots,\lambda_{s+\theta}) $ denote the number of  distinct choices  for the chain $\mathcal{C}^{(1)}\subseteq \mathcal{C}^{(2)} \subseteq \cdots \subseteq \mathcal{C}^{(s+\theta_e)} $ of self-orthogonal codes of length $n$  over $\mathcal{T}_m,$  such that (i) $\dim \mathcal{C}^{(i)}=\Lambda_i$ for $1 \leq i \leq s+\theta_e,$ (ii) the code $\mathcal{C}^{(s-\kappa_1)}$ is  doubly even, and  (iii)  $\textbf{1}\in  \mathcal{C}^{(s-\kappa+\theta_e)}$ with either $n\equiv 0\pmod 8$  or $n\equiv 4 \pmod 8$ and $m$  being even. We have $M(\lambda_1,\lambda_2,\ldots,\lambda_{s+\theta_e})=0$ if $\Lambda_{s-\kappa+\theta_e}=0.$ Further, when $\Lambda_{s-\kappa+\theta_e}\geq 1,$ we have 
\begin{eqnarray*}
M(\lambda_1,\lambda_2,\ldots,\lambda_{s+\theta_e})&=& \widehat{\sigma}_m(n;\Lambda_{s-\kappa_1})\prod\limits_{g=\Lambda_{s-\kappa_1}}^{\Lambda_{s+\theta_e}-1}\left( \frac{2^{m(n-2g)}-1}{2^{m(g+1-\Lambda_{s-\kappa_1})}-1}\right)  {\Lambda_{s-\kappa_1}-1 \brack \Lambda_{s-\kappa_1}-\Lambda_{s-\kappa+\theta_e}}_{2^m} \\&&\times \prod\limits_{\ell=1}^{s-\kappa+\theta_e}{\Lambda_{\ell} \brack \lambda_{\ell}}_{2^m} \prod\limits_{b=s-\kappa+1+\theta_e}^{s-\kappa_1}{\Lambda_b-\Lambda_{s-\kappa+\theta_e} \brack \lambda_b}_{2^m}       \prod\limits_{d=s-\kappa_1+1}^{s+\theta_e}\hspace{-1mm}{\Lambda_d-\Lambda_{s-\kappa_1} \brack \lambda_d}_{2^m}. \end{eqnarray*} 
\end{lemma}
 \begin{proof}
Working as in Lemma \ref{p5.2}, the desired result follows immediately.
\end{proof}

 In the following lemma, we consider the case $2\kappa > e$ and count the choices for the chain $\mathcal{C}^{(1)}\subseteq \mathcal{C}^{(2)} \subseteq \cdots \subseteq \mathcal{C}^{(s+\theta_e)} $ of self-orthogonal codes of length $n$  over $\mathcal{T}_m, $ such that  (i) $\dim \mathcal{C}^{(i)}=\Lambda_i$ for $1 \leq i \leq s+\theta_e,$ (ii) the code $\mathcal{C}^{(s-\kappa_1)}$ is  doubly even, and  (iii) $\textbf{1}\in  \mathcal{C}^{(1)}.$
\begin{lemma}\label{p5.4}
    Let $e\geq 3$ be an  integer satisfying $2\kappa > e.$  
    Let  $Z(\lambda_1,\lambda_2,\ldots,\lambda_{s+\theta}) $ denote the number of distinct choices  for the chain $\mathcal{C}^{(1)}\subseteq \mathcal{C}^{(2)} \subseteq \cdots \subseteq \mathcal{C}^{(s+\theta_e)} $ of self-orthogonal codes of length $n$  over $\mathcal{T}_m,$  such that (i) $\dim \mathcal{C}^{(i)}=\Lambda_i$ for $1 \leq i \leq s+\theta_e$ (ii) the code $\mathcal{C}^{(s-\kappa_1)}$ is  doubly even, and  (iii)  $\textbf{1}\in  \mathcal{C}^{(1)}.$ We have   $n\equiv 0,4\pmod 8,$ and  $Z(\lambda_1,\lambda_2,\ldots,\lambda_{s+\theta_e})=0$ if $\Lambda_1=0.$ Further, when $\Lambda_1\geq 1,$ we have  
\begin{eqnarray*}
   Z(\lambda_1,\lambda_2,\ldots,\lambda_{s+\theta_e})&=&  \widehat{\sigma}_m(n;\Lambda_{s-\kappa_1}) {\Lambda_{s-\kappa_1}-1 \brack \Lambda_{s-\kappa_1}-\Lambda_{1}}_{2^m} \prod\limits_{g=\Lambda_{s-\kappa_1}}^{\Lambda_{s+\theta_e}-1}\left( \frac{2^{m(n-2g)}-1}{2^{m(g+1-\Lambda_{s-\kappa_1})}-1}\right)    \\
  && \times \prod\limits_{d=2}^{s-\kappa_1}{\Lambda_d-\Lambda_{1} \brack \lambda_d}_{2^m} \prod\limits_{b=s-\kappa_1+1}^{s+\theta_e}\hspace{-1mm}{\Lambda_b-\Lambda_{s-\kappa_1} \brack \lambda_b}_{2^m}  . \end{eqnarray*} 
\end{lemma}
 \begin{proof}
Working as in Lemma \ref{p5.2}, the desired result follows immediately.
\end{proof}

In the following theorem, we provide an explicit enumeration formula for the number $\mathfrak{S}_e(n;\lambda_1,\lambda_2,\ldots,\lambda_e).$ 
\begin{theorem}\label{t4.1Kodd} 
Let $e\geq 3$ be an  integer.  Let $n$ be a positive integer, and   let $\lambda_1, \lambda_2, \ldots,  \lambda_{e+1}$ be non-negative integers satisfying $n=\lambda_1+\lambda_2+\cdots+\lambda_{e+1}$ and  $2\lambda_1+2\lambda_2+\cdots+2\lambda_{e-j+1}+\lambda_{e-j+2}+\cdots+\lambda_j \leq n$ for $s+1\leq j\leq e.$ Let $\Lambda_0=0$ and   $\Lambda_i=\lambda_1+\lambda_2+\cdots+\lambda_{i}$ for $1 \leq i \leq e+1.$ We have  
\begin{eqnarray*} \vspace{-1mm} \mathfrak{S}_e(n;\lambda_1,\lambda_2,\ldots,\lambda_e)=\displaystyle
	 \displaystyle  (2^m)^{\sum\limits_{i=1}^{s}\Lambda_{i}(n-\Lambda_{i+1})+\sum\limits_{j=1}^{s-1+\theta_e}\Lambda_{s+j}(n-\Lambda_{s+j+1}-\Lambda_{s-j+\theta_e})-\sum\limits_{a=1}^{s-\kappa_1-1}\Lambda_a-(1-\theta_e)\frac{\Lambda_s(\Lambda_s-1)}{2}} \\
	 \displaystyle  \times \mathfrak{B}_{\theta_e}(\lambda_1,\lambda_2,\ldots,\lambda_{s+\theta_e}) \prod\limits_{\ell=s+1+\theta_e}^{e}{\lambda_{\ell}+n-\Lambda_{\ell}-\Lambda_{e+1-\ell}\brack \lambda_{\ell}}_{2^m}, ~~~~~~~~~~~~~~~ \end{eqnarray*} where 
 \vspace{-1mm} 
\small{\begin{equation*}\mathfrak{B}_{\theta_e}(\lambda_1,\lambda_2,\ldots, \lambda_{s+\theta_e})= \left\{\begin{array}{ll}
		\displaystyle N(\lambda_1,\lambda_2,\ldots,\lambda_{s+\theta_e})	   & \text{if  } n\equiv 1,2,3,5,6,7 ~(\bmod~8); \vspace{1mm} \\
        \displaystyle N(\lambda_1,\lambda_2,\ldots,\lambda_{s+\theta_e})+  2(2^{m})^{\kappa_1}M(\lambda_1,\lambda_2,\ldots,\lambda_{s+\theta_e})& \\  +  \sum\limits_{\omega=0}^{\kappa_1-\theta_e} (2^m)^{\omega}Y_\omega(\lambda_1,\lambda_2,\ldots,\lambda_{s+\theta_e})& \text{if  } 2\kappa \leq e \text{ with   either } n\equiv 0 ~(\bmod~8) \vspace{-1mm}\\ & \text{or } n\equiv 4 ~(\bmod~8) \text{ and }  m \text{ is even;}  \\
N(\lambda_1,\lambda_2,\ldots,\lambda_{s+\theta_e})+  \sum\limits_{\omega=0}^{\kappa_1-\theta_e} (2^m)^{\omega}Y_\omega(\lambda_1,\lambda_2,\ldots,\lambda_{s+\theta_e})& \text{if  } 2\kappa \leq e ,  n\equiv 4 ~(\bmod~8) \text{ and }  \vspace{-1mm}\\ & m \text{ is odd;}\\
N(\lambda_1,\lambda_2,\ldots,\lambda_{s+\theta_e}) +(2^{m})^{s-\kappa_1-1}Z(\lambda_1,\lambda_2,\ldots,\lambda_{s+\theta_e}) & \\ 
 + \sum\limits_{\omega=0}^{s-\kappa_1-2} (2^m)^\omega Y_\omega(\lambda_1,\lambda_2,\ldots,\lambda_{s+\theta_e}) &\text{if } 2\kappa >e \text{ and } n\equiv 0,4 \pmod 8.
\end{array}\right.\vspace{-2mm}\end{equation*} }\normalsize 
\vspace{-2mm}\end{theorem}
\vspace{-2mm}\begin{proof} When $2\kappa\leq e,$ the desired result follows  by   Lemmas \ref{t3.1Kodd}, \ref{p5.1}, \ref{p5.2} and \ref{p5.3}.
On the other hand, when  $2\kappa>e,$ we get the desired result by  Lemmas \ref{t3.2Kodd}, \ref{p5.1}, \ref{p5.2} and \ref{p5.4}.
\end{proof}

 \small
\begin{table}[h!]
\centering
\resizebox{\textwidth}{!}{%
\begin{tabular}{|c|c||c|c||c|c|}
\hline
$\{\lambda_1,\lambda_2,\lambda_3,\lambda_4\}$ & $\mathfrak{S}_4(3; \lambda_1,\lambda_2,\lambda_3,\lambda_4)$ &
$\{\lambda_1,\lambda_2,\lambda_3,\lambda_4\}$ & $\mathfrak{S}_4(3; \lambda_1,\lambda_2,\lambda_3,\lambda_4)$ &
$\{\lambda_1,\lambda_2,\lambda_3,\lambda_4\}$ & $\mathfrak{S}_4(3; \lambda_1,\lambda_2,\lambda_3,\lambda_4)$ \\
\hline
$\{0, 0, 0, 1\} $ & $7$ & $\{0, 0, 0, 2\}$ & $7$ & $\{0, 0, 0, 3\}$ & $1$ \\ \hline
$\{0,0,1,0 \} $ & $28$ & $\{ 0,0,1,1\}$ & $42$ & $\{ 0,0,1,2 \}$ & $7$ \\ \hline
$\{ 0,0,2,0 \} $ & $28  $ & $\{ 0,0,2,1\}$ & $  7$ & $\{0,0,3,0 \}$ & $ 1$ \\ \hline
$\{ 0,1,0,0\} $ & $ 48$ & $\{0,1,0,1 \}$ & $72 $ & $\{ 0,1,0,2\}$ & $12 $ \\ \hline
$\{ 0,1,1,0\} $ & $ 24$ & $\{0,1,1,1 \}$ & $ 6$ & $\{ 1,0,0,0\}$ & $ 0$ \\ \hline
$\{ 1,0,0,1\} $ & $0 $ & $\{1,0,1,0 \}$ & $ 0$ & & \\ \hline
\end{tabular}%
}
\caption{The number $\mathfrak{S}_4(3; \lambda_1,\lambda_2,\lambda_3,\lambda_4)$ of self-orthogonal codes of length $3$ and type $\{\lambda_1,\lambda_2,\lambda_3,\lambda_4\}$ over $\mathscr{R}_{4,1} = GR(4,1)[y]/\langle y^3+2,2y\rangle$}
\label{Tab1}
\end{table}

 \small
\begin{table}[h!]
\centering
\resizebox{\textwidth}{!}{%
\begin{tabular}{|c|c||c|c||c|c|}
\hline
$\{\lambda_1,\lambda_2,\ldots,\lambda_5\}$ & $\mathfrak{S}_5(3; \lambda_1,\lambda_2,\ldots,\lambda_5)$ &
$\{\lambda_1,\lambda_2,\ldots,\lambda_5\}$ & $\mathfrak{S}_5(3; \lambda_1,\lambda_2,\ldots,\lambda_5)$ &
$\{\lambda_1,\lambda_2,\ldots,\lambda_5\}$ & $\mathfrak{S}_5(3; \lambda_1,\lambda_2,\ldots,\lambda_5)$ \\
\hline
$\{0,0,0,0,1 \} $ & $7 $ & $\{ 0,0,0,0,2\}$ & $7 $ & $\{0,0,0,0,3 \}$ & $1 $ \\ \hline
$\{ 0,0,0,1,0\} $ & $28 $ & $\{0,0,0,1,1 \}$ & $42 $ & $\{ 0,0,0,1,2\}$ & $7 $ \\ \hline
$\{0,0,0,2,0 \} $ & $28 $ & $\{0,0,0,2,1 \}$ & $ 7$ & $\{0,0,0,3,0 \}$ & $ 1$ \\ \hline
$\{0,0,1,0,0 \} $ & $ 48$ & $\{ 0,0,1,0,1\}$ & $72 $ & $\{ 0,0,1,0,2\}$ & $12 $ \\ \hline
$\{ 0,0,1,1,0\} $ & $ 72$ & $\{0,0,1,1,1 \}$ & $18 $ & $\{0,0,1,2,0 \}$ & $ 3$ \\ \hline
$\{0,1,0,0,0 \} $ & $96 $ & $\{ 0,1,0,0,1\}$ & $144 $ & $\{ 0,1,0,0,2\}$ & $24 $ \\ \hline
$\{ 0,1,0,1,0\} $ & $48 $ & $\{0,1,0,1,1 \}$ & $ 12$ & $\{1,0,0,0,0 \}$ & $0$ \\ \hline
$\{ 1,0,0,0,1\} $ & $ 0$ & $\{ 1,0,0,1,0\}$ & $ 0$ & & \\ \hline
\end{tabular}%
}
\caption{The number $\mathfrak{S}_5(3; \lambda_1,\lambda_2,\ldots,\lambda_5)$ of self-orthogonal codes of length $3$ and type $\{\lambda_1,\lambda_2,\ldots,\lambda_5\}$ over $\mathscr{R}_{5,1} = GR(4,1)[y]/\langle y^3+2,2y^2\rangle$}
\label{Tab2}
\end{table}

\vspace{-3mm} \small
\begin{table}[h!]
\centering
\resizebox{\textwidth}{!}{%
\begin{tabular}{|c|c||c|c||c|c|}
\hline
$\{\lambda_1,\lambda_2,\ldots,\lambda_6\}$ & $\mathfrak{S}_6(2; \lambda_1,\lambda_2,\ldots,\lambda_6)$ &
$\{\lambda_1,\lambda_2,\ldots,\lambda_6\}$ & $\mathfrak{S}_6(2; \lambda_1,\lambda_2,\ldots,\lambda_6)$ &
$\{\lambda_1,\lambda_2,\ldots,\lambda_6\}$ & $\mathfrak{S}_6(2; \lambda_1,\lambda_2,\ldots,\lambda_6)$ \\
\hline
$\{0,0,0,0,0,1 \} $ & $ 5 $ & $\{ 0,0,0,0,0,2\}$ & $ 1  $ & $\{ 0,0,0,0,1,0 \}$ & $20 $ \\ \hline
$\{ 0,0,0,0,1,1 \} $ & $ 5 $ & $\{  0,0,0,0,2,0\}$ & $ 1 $ & $\{ 0,0,0,1,0,0 \}$ & $ 80$ \\ \hline
$\{ 0,0,0,1,0,1 \} $ & $ 20 $ & $\{ 0,0,0,1,1,0 \}$ & $5  $ & $\{ 0,0,0,2,0,0 \}$ & $1 $ \\ \hline
$\{ 0,0,1,0,0,0 \} $ & $  64$ & $\{ 0,0,1,0,0,1 \}$ & $ 16 $ & $\{ 0,0,1,0,1,0 \}$ & $4 $ \\ \hline
$\{ 0,1,0,0,0,0 \} $ & $ 0 $ & $\{ 0,1,0,0,0,1 \}$ & $ 0 $ & $\{ 1,0,0,0,0,0 \}$ & $ 0$ \\ \hline
\end{tabular}%
}
\caption{The number $\mathfrak{S}_6(2; \lambda_1,\lambda_2,\ldots,\lambda_6)$ of self-orthogonal codes of length $2$ and type $\{\lambda_1,\lambda_2,\ldots,\lambda_6\}$ over $\mathscr{R}_{6,2} = GR(4,2)[y]/\langle y^3+2,2y^3\rangle$}
\label{Tab3}
\end{table}

 \small{
\begin{table}[h!]
\centering
\resizebox{\textwidth}{!}{%
\begin{tabular}{|c|c||c|c||c|c|}
\hline
$\{\lambda_1,\lambda_2,\lambda_3,\lambda_4\}$ & $\mathfrak{S}_4(4; \lambda_1,\lambda_2,\lambda_3,\lambda_4)$ &
$\{\lambda_1,\lambda_2,\lambda_3,\lambda_4\}$ & $\mathfrak{S}_4(4; \lambda_1,\lambda_2,\lambda_3,\lambda_4)$ &
$\{\lambda_1,\lambda_2,\lambda_3,\lambda_4\}$ & $\mathfrak{S}_4(4; \lambda_1,\lambda_2,\lambda_3,\lambda_4)$ \\
\hline
$\{ 0,0,0,1\} $ & $ 15 $ & $\{ 0,0,0,2\}$ & $   35$ & $\{ 0,0,0,3 \}$ & $ 15$ \\ \hline
$\{ 0,0,0,4 \} $ & $ 1 $ & $\{ 0,0,1,0 \}$ & $ 120 $ & $\{ 0,0,1,1 \}$ & $420 $ \\ \hline
$\{  0,0,1,2 \} $ & $ 210 $ & $\{ 0,0,1,3 \}$ & $ 15  $ & $\{ 0,0,2,0 \}$ & $560 $ \\ \hline
$\{  0,0,2,1 \} $ & $  420$ & $\{ 0,0,2,2 \}$ & $ 35  $ & $\{ 0,0,3,0 \}$ & $ 120$ \\ \hline
$\{  0,0,3,1 \} $ & $15  $ & $\{  0,0,4,0\}$ & $ 1  $ & $\{ 0,1,0,0 \}$ & $ 448$ \\ \hline
$\{ 0,1,0,1  \} $ & $ 1568 $ & $\{ 0,1,0,2 \}$ & $  784 $ & $\{ 0,1,0,3 \}$ & $ 56$ \\ \hline
$\{ 0,1,1,0  \} $ & $ 1344 $ & $\{ 0,1,1,1 \}$ & $ 1008  $ & $\{ 0,1,1,2 \}$ & $84 $ \\ \hline
$\{ 0,1,2,0  \} $ & $ 112 $ & $\{ 0,1,2,1 \}$ & $  14 $ & $\{ 0,2,0,0 \}$ & $384 $ \\ \hline
$\{ 0,2,0,1  \} $ & $ 288 $ & $\{ 0,2,0,2 \}$ & $  24 $ & $\{ 1,0,0,0 \}$ & $256 $ \\ \hline
$\{  1,0,0,1 \} $ & $ 384 $ & $\{ 1,0,0,2 \}$ & $ 64  $ & $\{ 1,0,1,0 \}$ & $384 $ \\ \hline
$\{ 1,0,1,1  \} $ & $ 96 $ & $\{  1,0,2,0\}$ & $ 16  $ & $\{ 1,1,0,0 \}$ & $384 $ \\ \hline
$\{ 1,1,0,1  \} $ & $ 96 $ & $\{ 2,0,0,0 \}$ & $  0 $ &  &\\ \hline
\end{tabular}%
}
\caption{The number $\mathfrak{S}_4(4; \lambda_1,\lambda_2,\lambda_3,\lambda_4)$ of self-orthogonal codes of length $4$ and type $\{\lambda_1,\lambda_2,\lambda_3,\lambda_4\}$ over $\mathscr{R}_{4,1} = GR(4,1)[y]/\langle y^3+2,2y\rangle$}
\label{Tab4}
\end{table}}\normalsize

\begin{example}\label{Example 4.1}
  Using Magma \cite{Mag}, we provide the number $\mathfrak{S}_4(3; \lambda_1,\lambda_2,\lambda_3,\lambda_4)$ of self-orthogonal codes of length $3$ and type $\{\lambda_1,\lambda_2,\lambda_3,\lambda_4\}$ over $\mathscr{R}_{4,1} = GR(4,1)[y]/\langle y^3+2,2y\rangle$  in Table \ref{Tab1},  the number $\mathfrak{S}_5(3; \lambda_1,\lambda_2,\ldots,\lambda_5)$ of self-orthogonal codes of length $3$ and type $\{\lambda_1,\lambda_2,\ldots,\lambda_5\}$ over $\mathscr{R}_{5,1} = GR(4,1)[y]/\langle y^3+2,2y^2\rangle$   in  Table \ref{Tab2},   the number $\mathfrak{S}_6(2; \lambda_1,\lambda_2,\ldots,\lambda_6)$ of self-orthogonal codes of length $2$ and type $\{\lambda_1,\lambda_2,\ldots,\lambda_6\}$ over $\mathscr{R}_{6,2} = GR(4,2)[y]/\langle y^3+2,2y^3\rangle$  in  Table \ref{Tab3}, and  the number $\mathfrak{S}_4(4; \lambda_1,\lambda_2,\lambda_3,\lambda_4)$ of self-orthogonal codes of length $4$ and type $\{\lambda_1,\lambda_2,\lambda_3,\lambda_4\}$ over $\mathscr{R}_{4,1} = GR(4,1)[y]/\langle y^3+2,2y\rangle$ in  Table \ref{Tab4}.  These computed  values agree exactly with Theorem \ref{t4.1Kodd}. 
\end{example}  
\begin{remark}\label{rem9.1}
When $\kappa$ is odd,
Theorem 5.1  of Yadav and Sharma \cite{quasi} can be obtained as a special case of  Theorem \ref{t4.1Kodd} by setting $\mathfrak{s}=1$ and $e=\mathtt{t}=\kappa.$ 
\end{remark}

 Now, in the following theorem, we provide an explicit enumeration formula for the number $\mathfrak{U}_e(n;\lambda_1,\lambda_2,\ldots,\lambda_e),$ where $e\geq 3.$   
\begin{theorem}\label{t4.2Kodd} 
Let $e\geq 3$ be an  integer.  Let $n$ be a positive integer, and   let $\lambda_1, \lambda_2, \ldots,  \lambda_{e+1}$ be non-negative integers satisfying $n=\lambda_1+\lambda_2+\cdots+\lambda_{e+1}$ and  $\lambda_j=\lambda_{e-j+2}$ for $1 \leq j \leq e+1.$ Let $\Lambda_0=0$ and   $\Lambda_i=\lambda_1+\lambda_2+\cdots+\lambda_{i}$ for $1 \leq i \leq e+1.$ We have  
\vspace{-1mm}\begin{equation*} \vspace{-1mm} \mathfrak{U}_e(n;\lambda_1,\lambda_2,\ldots,\lambda_e)=\displaystyle
	 \displaystyle  \mathfrak{B}_{\theta_e}(\lambda_1,\lambda_2,\ldots,\lambda_{s+\theta_e}) (2^m)^{\sum\limits_{i=1}^{s}\Lambda_{i}(n-\Lambda_{i+1})-\sum\limits_{a=1}^{s-\kappa_1-1}\Lambda_a-(1-\theta_e)\frac{\Lambda_s(\Lambda_s-1)}{2}}  , \end{equation*} where the number $\mathfrak{B}_{\theta_e}(\lambda_1,\lambda_2,\ldots, \lambda_{s+\theta_e})$ is as obtained in Theorem \ref{t4.1Kodd}.
 \vspace{-1mm} \end{theorem}
 \begin{proof}
    By Lemma  \ref{l2.2} and on substituting $\lambda_{j}= \lambda_{e-j+2}$  for $1 \leq j \leq e+1$ in Theorem \ref{t4.1Kodd}, we get the desired result.
 \end{proof}
\begin{remark}\label{rem9.2} When $\kappa$ is odd, 
  Theorem 5.3  of Yadav and Sharma \cite{quasi} follows as a direct consequence of Theorem \ref{t4.2Kodd}  upon setting $\mathfrak{s}=1$ and $e=\mathtt{t}=\kappa.$\end{remark}
  
  With the help of the enumeration formulae for self-orthogonal and self-dual codes of type $\{\lambda_1,\lambda_2,\ldots,\lambda_e\}$ and length $n$ over $\mathscr{R}_{e,m},$  obtained in Theorems \ref{t4.1Kodd} and \ref{t4.2Kodd}, respectively, and by applying the classification algorithm described in \cite[Sec. 9.6 and 9.7]{H} and carrying out computations in Magma \cite{Mag}, one can obtain complete lists of inequivalent self-orthogonal and self-dual codes of a given type and length 
 over $\mathscr{R}_{e,m}$ \cite{Galois,quasi,Y}.  In the following example, we will now illustrate this  by classifying all self-orthogonal codes of type $\{0,1,1,0\}$ and length $3$ over $\mathscr{R}_{4,1}= GR(4,1)[y]/\langle y^3+2,2y\rangle,$
 up to monomial equivalence.  We will also explicitly determine a generator matrix for a representative of each equivalence class of such self-orthogonal codes. 
 
 \begin{example}\label{EXclassify}
From Table \ref{Tab1}, we  note that there are precisely $24$ distinct self-orthogonal codes of type $\{0,1,1,0\}$ and length $3$ over $\mathscr{R}_{4,1}= GR(4,1)[y]/\langle y^3+2,2y\rangle.$ We will now apply the  classification algorithm \cite[Sec. 9.6 and 9.7]{H} and carry out computations in Magma \cite{Mag} to obtain all inequivalent self-orthogonal codes of type $\{0,1,1,0\}$ and length $3$ over $\mathscr{R}_{4,1}.$   

For a linear code $\mathcal{C}$ of length $n$ over $\mathscr{R}_{e,m},$  let $\mathcal{M}ono(\mathcal{C})$  denote the set of all monomial transformations that map  $\mathcal{C}$ onto a self-orthogonal code over $\mathscr{R}_{e,m},$ and let $\mathcal{A}ut(\mathcal{C})$ denote the group of all monomial transformations that map the code  $\mathcal{C}$ onto itself.

To begin with, we consider the self-orthogonal code $\mathcal{D}_1$ of type $\{0,1,1,0\}$ and length $3$ over $\mathscr{R}_{4,1}$ with a generator matrix \vspace{-1mm}\begin{equation*}
    \vspace{-1mm}\begin{bmatrix} u&0&u+u^3\\ 0&u^2&u^3 \end{bmatrix}.
\end{equation*} Using Magma \cite{Mag}, we compute $|\mathcal{M}ono(\mathcal{D}_1)|=3072$ and $|\mathcal{A}ut(\mathcal{D}_1)|=256.$ By the orbit–stabilizer principle \cite[Th. 7.4]{Gallian},  there are precisely $\frac{|\mathcal{M}ono(\mathcal{D}_1)|}{|\mathcal{A}ut(\mathcal{D}_1)|}= 12$ distinct self-orthogonal codes of type $\{0,1,1,0\}$ and length $3$  over $\mathscr{R}_{4,1}$ that are monomially equivalent to $\mathcal{D}_1.$ 

Using Magma \cite{Mag} again, we further observe that the code $\mathcal{D}_2$ with a generator  matrix \vspace{-1mm}\begin{equation*}
    \begin{bmatrix}u&0&u+u^2+u^3\\0&u^2&0
\end{bmatrix}\vspace{-1mm}
\end{equation*} is a self-orthogonal code of type $\{0,1,1,0\}$ and length $3$  over $\mathscr{R}_{4,1},$ which is monomially inequivalent to $\mathcal{D}_1.$ With the help of Magma, we compute $|\mathcal{M}ono(\mathcal{D}_2)|=3072$ and $|\mathcal{A}ut(\mathcal{D}_2)|=256.$ By the orbit–stabilizer principle \cite[Th. 7.4]{Gallian} again, there are precisely $\frac{|\mathcal{M}ono(\mathcal{D}_2)|}{|\mathcal{A}ut(\mathcal{D}_2)|}= 12$ distinct self-orthogonal codes of type $\{0,1,1,0\}$ and length $3$  over $\mathscr{R}_{4,1},$ which are monomially equivalent to $\mathcal{D}_2.$

Hence, the total number of self-orthogonal codes of type $\{0,1,1,0\}$ and length $3$  over $\mathscr{R}_{4,1}$ that are monomially equivalent to either $\mathcal{D}_1$ or $\mathcal{D}_2$ is $12+12=24,$ which coincides with  the total number of such codes obtained from Table \ref{Tab1}. Consequently, there are precisely two monomially inequivalent self-orthogonal codes of type $\{0,1,1,0\}$ and length $3$  over $\mathscr{R}_{4,1},$ \textit{viz.} $\mathcal{D}_1$
and $\mathcal{D}_2,$ with generator matrices  \vspace{-1mm}\begin{equation*}
    \begin{bmatrix} u&0&u+u^3\\ 0&u^2&u^3 \end{bmatrix} \text{ ~and ~} \begin{bmatrix}u&0&u+u^2+u^3\\0&u^2&0
\end{bmatrix},\vspace{-1mm}
\end{equation*} respectively.\end{example}

\section{Conclusion and future work}\label{Conclusion}

In this paper, explicit enumeration formulae for all self-orthogonal and self-dual codes of an arbitrary length over a finite commutative chain ring $\mathscr{R}_{e,m}$ of even characteristic are derived, under the condition that
$2\in \langle u^{\kappa}\rangle \setminus \langle u^{\kappa+1}\rangle $ for some odd positive integer $\kappa$ satisfying $3\leq \kappa \leq  e,$ where  $\mathscr{R}_{e,m}$ is a finite commutative chain ring  with maximal ideal $\langle u \rangle$ of  nilpotency index $e \geq 3.$  A subsequent study \cite{YSub} will address the complementary case 
where $\kappa$ is even.   Enumeration formulae for  all self-orthogonal and self-dual codes of an arbitrary length over Galois rings of even characteristic (\textit{i.e.,} when  $\kappa=1$) and  finite commutative chain rings of odd characteristic  are explicitly derived in \cite{Galois} and  \cite{Y}, respectively.  Together, the results derived in this paper, along with those in \cite{YSub,Galois} and \cite{Y}, provide a complete solution to the problem of enumeration of  self-orthogonal and self-dual codes of an arbitrary length over any finite commutative chain ring.   Future research directions include extending enumeration to codes over more general finite commutative Frobenius rings and exploring structural classifications of self-orthogonal and self-dual codes over various classes of finite commutative Frobenius rings.
\section{Acknowledgements}
The authors acknowledge the research support provided by the Department of Science and Technology, India, under Grant no. DST/INT/RUS/RSF/P-41/2021 with TPN 65025.

{}
\end{document}